\documentclass[english]{article}
\usepackage{geometry}
\usepackage{parskip}
\usepackage{chris_style}
\usepackage{xcolor}
\usepackage{bm} 
\usepackage{amssymb}
\usepackage{authblk}
\usepackage{bbm}

\usepackage[authoryear]{natbib}

\let\emptyset\varnothing

\crefname{theorem}{Theorem}{Theorems}
\crefname{lemma}{Lemma}{Lemmas}
\crefname{corollary}{Corollary}{Corollaries}
\crefname{proposition}{Proposition}{Propositions}
\crefname{definition}{Definition}{Definitions}
\crefname{remark}{Remark}{Remarks}
\crefname{assumption}{Assumption}{Assumptions}
\crefname{conjecture}{Conjecture}{Conjectures}

\Crefname{theorem}{Theorem}{Theorems}
\Crefname{lemma}{Lemma}{Lemmas}
\Crefname{corollary}{Corollary}{Corollaries}
\Crefname{proposition}{Proposition}{Propositions}
\Crefname{definition}{Definition}{Definitions}
\Crefname{remark}{Remark}{Remarks}
\Crefname{assumption}{Assumption}{Assumptions}
\Crefname{conjecture}{Conjecture}{Conjectures}

\allowdisplaybreaks
\def\@#1\@{\begin{align}#1\end{align}}
\def\$#1\${\begin{align*}#1\end{align*}}
\title{Sample-Efficient Multiple Testing with Adaptive Data Collection}
\author{Zhanran Lin\textsuperscript{*}}
\author{Wanteng Ma\textsuperscript{*}}
\author{Zhimei Ren}
\author{Yuting Wei}
\affil{Department of Statistics and Data Science, University of Pennsylvania}
\date{\today}

\allowdisplaybreaks
\newcommand{\TPP}{\operatorname{TPP}}
\newcommand{\FDP}{\operatorname{FDP}}

\newcommand{\ind}{\mathds{1}}

\newcommand{\sfU}{\mathsf{U}}
\newcommand{\sfL}{\mathsf{L}}
\newcommand{\sfR}{\mathsf{R}}

\newcommand{\sfB}{\mathsf{B}}
\newcommand{\sfC}{\mathsf{C}}
\newcommand{\sfD}{\mathsf{D}}
\newcommand{\sfO}{\mathsf{O}}
\newcommand{\sfH}{\mathsf{H}}
\newcommand{\sfh}{\mathsf{h}}
\newcommand{\sfl}{\mathsf{l}}
\newcommand{\sfw}{\mathsf{w}}
\newcommand{\di}{\mathrm{d}}
\newcommand{\cC}{\mathcal{C}}
\newcommand{\cH}{\mathcal{H}}

\newcommand{\cA}{\mathcal{A}}
\newcommand{\cB}{\mathcal{B}}
\newcommand{\given}{\,|\,}
\newcommand{\biggiven}{\, \big|\,}
\newcommand{\Biggiven}{\, \Big|\,}
\newcommand{\bigggiven}{\, \bigg|\,}

\newcommand{\sfa}{\mathsf{a}}
\newcommand{\sfb}{\mathsf{b}}
\newcommand{\bmX}{{\bm X}}
\let\hat\widehat
\begin{document}

\maketitle

\begingroup
\renewcommand{\thefootnote}{*}
\footnotetext{Equal contribution.}
\endgroup

\begin{abstract}
This paper studies adaptive experimental design for multiple testing,
where an experimenter sequentially chooses which hypothesis to sample. 
We propose the e-value-based posterior sampling (e-PS) procedure, which uses the empirical average of log e-value increments to guide randomized sampling and applies e-BH to construct rejection sets.
Under conditionally valid e-value increments, the procedure controls the false discovery rate at arbitrary stopping times and produces nested rejection sets.
We establish high-probability bounds on the number of samples needed to discover all nonnull hypotheses in terms of the growth and concentration of the underlying e-processes. We specialize these bounds to simple-versus-simple, composite-versus-simple, and simple-versus-composite testing. Simulations and experiments using joke ratings and watermarked text illustrate the procedure’s power under limited sampling budgets.
\end{abstract}

\setcounter{tocdepth}{2}
\tableofcontents

\section{Introduction}
In many modern applications, a decision maker seeks to identify all signals that satisfy a prescribed criterion from a large pool of candidates. For example, in an A/B testing platform, one may evaluate $K$ policies with the goal of selecting those that meet a target performance threshold~\citep{kohavi2020trustworthy}. Similarly, in LLM evaluation, practitioners compare multiple models across a range of tasks and seek to identify those that outperform a given baseline~\citep{chang2024survey}. These tasks can be naturally framed as a problem of testing multiple hypotheses simultaneously. In such settings, data collection is often costly, making it essential to attain reliable discoveries with minimal sampling budget while maintaining rigorous control of statistical error rates.

While a substantial body of work has been devoted to multiple testing (see e.g.,~\citet{benjamini1995controlling,benjamini2000adaptive,benjamini2001control, storey2002direct,barber2015controlling,candes2018panning,wang2022false} for a highly incomplete list of references), 
the majority of existing approaches focus on data collected under fixed experimental designs, independent of the subsequent analysis.
This separation of design and analysis can be inefficient: intuitively, 
more samples should be allocated to ``hard'' hypotheses near the decision boundary, 
and fewer to those that are clearly null or non-null.
In this work, we instead take a holistic view that integrates {\em data collection} and {\em inference}, jointly designing
adaptive sampling strategies and test statistics to 
improve the sample efficiency of multiple testing 
while maintaining rigorous error control.

To formalize the problem, consider $K$ hypotheses. 
At each time $t \ge 1$, we select one hypothesis to query, 
and observe only the statistic associated with the chosen hypothesis.
The sampling strategy may be adaptive: the choice at time $t$ 
may depend on the data collected up to time $t-1$. 
Based on the accumulated observations for each hypothesis, 
we construct at every time step a rejection set $\mR_t\subseteq \{1,\ldots,K\}$ 
consisting of the indices declared significant. Our task is 
to design an algorithm---including sampling and testing---that 
identifies all non-nulls with a minimal sampling budget, 
while controlling the false discovery rate~\citep[FDR;][]{benjamini1995controlling}, i.e., the expected proportion of false discoveries, 
 at a pre-specified level $\alpha$.
Since the rejection sets $\mR_t$ are updated sequentially and may be stopped in a data-driven manner, 
we impose the stronger requirement of {\em anytime-valid FDR control}, 
which guarantees FDR control at every, possibly data-dependent, 
stopping time.

A key distinction between our setting and  
classical multiple testing is that the sampling strategy 
can be updated adaptively.
While adaptive data collection provides opportunities for
more efficient data usage, 
it also introduces significant challenges for FDR control. 
In particular, the complex dependencies induced by adaptive sampling 
and/or data-driven stopping
can invalidate many classical multiple testing procedures,
such as the Benjamini-Hochberg (BH) procedure~\citep{benjamini1995controlling},
or require substantial correction factors to restore their validity.
Recent work by~\citet{xu2021unified} addresses this challenge 
using the e-values~\citep{ramdas2025hypothesis}. Given 
an adaptive data collection algorithm, their approach 
constructs an {\em e-process} for each hypothesis and  
applies the e-BH procedure~\citep{wang2022false}
to the e-process values at each time point. 
The resulting rejection set satisfies anytime-valid 
FDR control by virtue of the e-processes and the e-BH procedure.
\citet{xu2021unified} decouple sampling 
from inference, focusing on constructing rejection sets 
with anytime-valid FDR control under a {\em given} data collection mechanism.
When the experimenter can also control data collection,  
however, the questions of how to design an efficient sampling algorithm 
and establish corresponding sample-complexity guarantees remain open.
Building on the e-value-based inferential framework of~\citet{xu2021unified}, 
we develop adaptive sampling algorithms that achieve 
high power with provable sample-efficiency guarantees.
As a preview of our results, Figure~\ref{fig:intro_compact}
shows that our proposed method 
can achieve the same target true positive rate (TPR) as the baseline methods
with a significantly smaller sampling budget, while controlling FDR at the desired level.


\begin{figure}[t]
    \centering
\includegraphics[width = 0.6\textwidth]{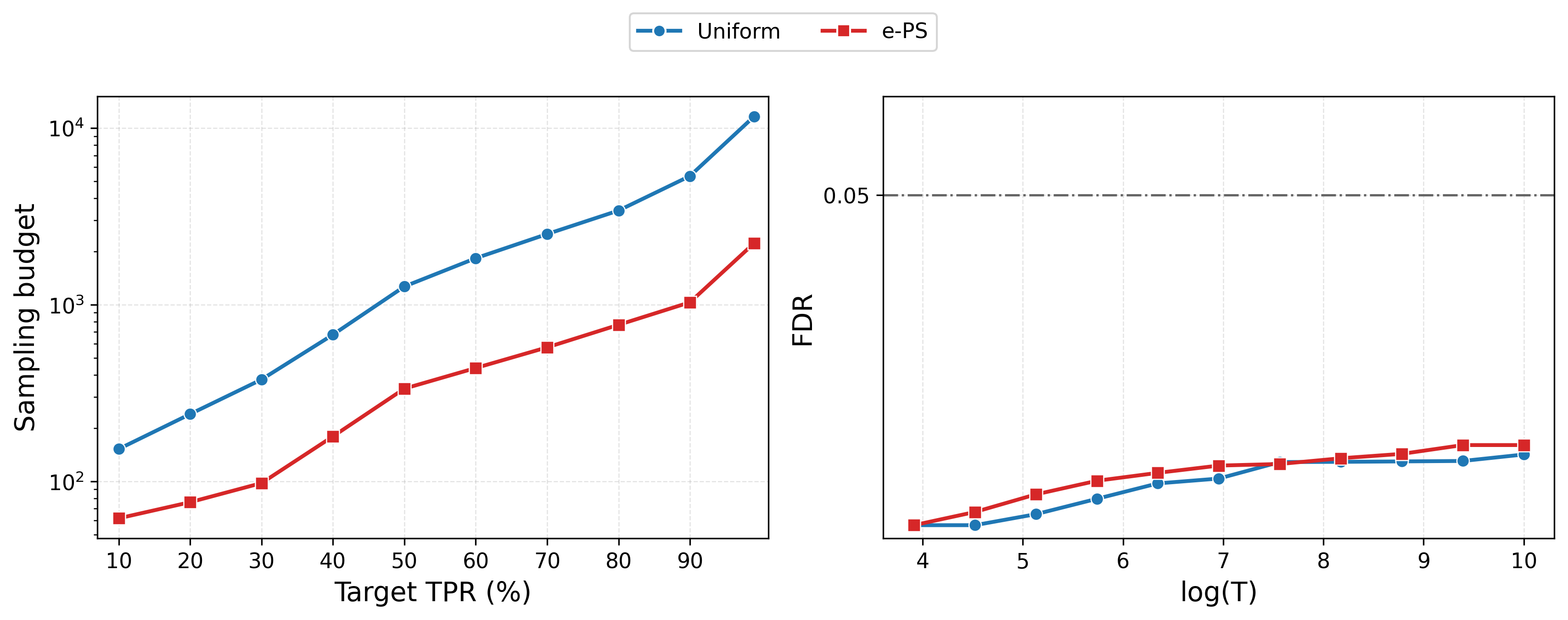}
    \caption{
    A glimpse of our contribution in multivariate Gaussian mean-shift 
    testing problems (the details are provided in Section \ref{sec:simulation}). 
    We show  the average budget required to reach a given target TPR of our proposed framework compared with the uniform benchmark (left), 
    as well as the realized FDR as a function of the budget(right).
    }
\label{fig:intro_compact}
\end{figure}


In the bandit literature, many sampling algorithms have been developed
for minimizing cumulative or simple regret, with respect 
to scalar reward~\citep{lattimore2020bandit}. 
By contrast, adaptive sampling for multiple testing remains less studied,
with notable exceptions of~\citet{jamieson2018bandit,xiang2022multi}.
These works associate each arm with a scalar reward
and test whether its mean exceeds a given threshold,
with sampling decisions directly driven by the  
scalar reward---the upper confidence bound (UCB) 
algorithm of~\citet{jamieson2018bandit} samples the 
the arm with the largest upper confidence bound on its mean.
For hypotheses of more general forms---such as 
whether a parameter vector lies in a certain set
or whether an outcome follows the same distribution 
as the negative controls---there is no explicit scalar reward
and it is unclear how 
existing sampling algorithms generalize.
We address this gap by developing adaptive sampling algorithms for general hypothesis-testing problems beyond tests of scalar means
and providing their sample complexity guarantees.




\subsection{Our contributions}
In this work, we propose {\em e-value-based 
posterior sampling (e-PS)}, a general 
active-sampling framework 
for multiple testing. In e-PS, e-values play a dual role: 
they quantify evidence against each null hypothesis 
and determine rejections through the e-BH procedure,
while also defining a surrogate reward process that guides sampling.
Specifically, e-PS constructs a carefully calibrated 
posterior distribution for each hypothesis, centered at its current 
surrogate reward; draws one synthetic reward from each posterior 
distribution; and samples the hypothesis with the largest draw.
Intuitively, this randomization balances between exploitation, which favors
hypotheses with strong evidence against the null, and 
exploration, which continues to sample hypotheses whose evidence
remains uncertain.


Our framework produces a sequence of
rejection sets with anytime-valid 
FDR control. We show that e-PS yields nested 
rejection sets, $\mR_t \subseteq \mR_{t+1}$, 
for any $t \ge 1$,
so that discoveries can never be revoked.
We further derive high-probability upper bounds on the number 
of samples needed to reject all non-null hypotheses, 
expressed through the growth rates of the 
underlying e-processes. 
We specialize these general bounds to 
simple-versus-simple, composite-versus-simple, and 
simple-versus-composite testing. 
For the simple-versus-simple setting, 
we also provide a complementary lower bound 
that characterizes the fundamental 
sample-complexity limit of multiple testing.
We further validate the effectiveness of our approach through extensive simulations, 
along with two real-world case studies in recommendation system evaluation and LLM evaluation.



\subsection{Notation}
For any $a,b\in \RR$, we write $a \vee b = \max\{a,b\}$ and $a \wedge b = \min\{a,b\}$. 
Denote $x^+ = \max\{x,0\}$ for $x \in \mathbb{R}$  
and $\log_+(x) = \max\{\log x, 0\}$ for $x > 0$.
For $n \in \NN_+$, we write $[n] = \{1,2,\ldots,n\}$. 
For a set $A$, we write $\operatorname{conv}{A}$ for the convex hull of $A$. 
We use $\mathbf{1}_d$ to denote the $d$-dimensional vector of all ones, 
and ${\bm I}_d$ to deonote the $d-\times-d$ identity matrix. 
For two distributions $P$ and $Q$, we write $D_\KL(P\|Q)$ for their Kullback-Leibler divergence.
For two nonnegative functions $f$ and $g$, we write $f = O(g)$ (or $f \lesssim g$) 
if there exists a universal constant 
$C > 0$ such that $f \le C g$; we write $f = \Omega(g)$ (or $f \gtrsim g$) if $g = O(f)$, 
and $f = \Theta(g)$ (or $f \asymp g$) if both $f = O(g)$ and $f = \Omega(g)$ hold.
The notation $f=\tilde O(g)$, $f = \tilde \Omega (g)$, and $f = \tilde \Theta(g)$ 
is defined analogously, up to logarithmic factors.
For any $z \in \RR$, 
$\Phi(z)$ denotes the cumulative distribution function
of the standard normal distribution evaluated at $z$ and $R(z) = \Phi(z)/(1-\Phi(z))$.

\section{Preliminaries}\label{sec:prelim}
\subsection{Problem formulation}
Setting the stage, consider $K$ arms corresponding to different treatments in an A/B/K testing problem. 
At each time step 
$t \ge 1$,  a vector of potential outcomes 
$(Y_{t}(1),\ldots,Y_{t}(K))$ is i.i.d.~drawn
from some unknown joint distribution $P$. 
We impose no restrictions on the dependence structure among arms under the joint distribution 
$P$, allowing $(Y_{t}(k))_{k \in [K]}$ to be arbitrarily dependent across $k$ at each time $t$.

At time $t$, the experimenter selects an arm by sampling 
$A_t\sim \operatorname{Categorical}(\pi_t)$, 
where $\pi_t$ is a probability distribution over $K$ arms, 
and only observes the outcome of the chosen arm, 
denoted by $Y_t := Y_{t}(A_t)$.
The sampling strategy is required to be non-anticipating, 
meaning that $\pi_t$ is measurable to the $\sigma$-field $\mF_{t-1} = \sigma(A_1,Y_1,U_1,\ldots,A_{t-1},Y_{t-1},U_{t-1})$,
where $U_s$ captures any external independent randomness 
introduced by the algorithm, 
with $\mF_0 = \{\varnothing, \Omega\}$.

For each arm $k \in [K]$, we are interested in properties of 
its marginal distribution $P_k$. Formally, we consider testing 
the following hypothesis for each $k\in [K]$:
\begin{align}
H_k^{(0)}: P_k \in \mathcal{P}_k \text{ versus }  H_k^{(1)}: P_k \in \mathcal{Q}_k,
\end{align}
where $\mathcal{P}_k$ and $\mathcal{Q}_k$ are sets of distributions 
to be specified in later sections.
Our goal is to perform multiple testing over the $K$ hypotheses.
At each time $t$, the algorithm computes a statistic for each hypothesis based on the data collected so far and applies a rejection rule to determine a rejection set $\mR_t \subseteq [K]$ that contains the indices of rejected hypotheses.



Given a rejection set $\mR\subseteq [K]$, we focus on two key performance metrics: the FDR and the TPR, or power. 
Letting $\mH_0 \subseteq [K]$ be the collection of true nulls
and $\mH_1 = [K] \setminus \mH_0$ the set of nonnulls, 
we define the FDR and TPR as follows:
\begin{align*}
    &\text{FDR}(\mR) := \E[\text{FDP}(\mR)],~\text{FDP}(\mR) := \frac{\abs{\mR \cap \mH_0 } }{ \abs{\mR } \vee 1 }, \\
    & \text{TPR}(\mR) := \E[\text{TPP}(\mR)],~\text{TPP}(\mR) := \frac{\abs{\mR \cap \mH_1 } }{ \abs{\mH_1} \vee 1}.
\end{align*}
For a sequence of rejection sets $\{\mR_{t}\}_{t \ge 1},$
we seek a stronger notion of FDR control, {\em anytime-valid FDR control},
whose formal definition is given in Definition~\ref{def:anytime-fdr} below. 
\begin{definition}[Anytime-valid FDR control]\label{def:anytime-fdr}
    An algorithm is said to have the anytime-valid FDR at level $\alpha$ if for any stopping time $\tau$ with respect to the filtration $\{\mF_t\}_{t\ge 0}$, we have
    \begin{align*} 
        \operatorname{FDR}(\mR_\tau) \le \alpha.
    \end{align*}
\end{definition}
Our objective is to design an algorithm that 
produces a sequence of rejection sets $\{\mR_t\}_{t \ge 1}$, 
which achieves anytime-valid FDR control while maintaining high TPR with a minimal sampling budget.
In our proposed framework, 
we use e-values---obtained from stopped e-processes---as our evidence statistics, 
and the e-BH procedure as our rejection rule. 
We preview these ingredients next.

\subsection{E-values, e-processes, and the e-BH procedure}
A nonnegative random variable $E$ is called an e-value (or e-variable) for a null hypothesis 
if it satisfies $\E_{H_0}[E]\le 1$. The e-value can be interpreted as a measure of evidence against the null hypothesis, with larger values indicating stronger evidence.
In particular, rejecting $H_0$ when $E \ge 1/\alpha$ yields a valid 
level-$\alpha$ test by Markov's inequality.

With $K$ hypotheses and their corresponding e-values 
$E_1,\ldots,E_K$, 
the e-BH procedure~\citep{wang2022false} at level $\alpha$
outputs a rejection set 
\begin{align}    
\mR = \{k\in[K]: E_k \ge K/(\alpha k^*)\},
~\text{ where }~
k^* = \max \left\{ k \in [K] : E_{(k)} \ge \frac{K}{\alpha k} \right\},
\end{align}
with $E_{(1)} \ge \cdots \ge E_{(K)}$ denoting the e-values sorted in descending order and the convention that $\max\{\varnothing\}=0$
and $\mR = \varnothing$ when $k^* = 0$.
It is shown in~\citet{wang2022false} that the e-BH procedure controls the FDR at level $\alpha$ under arbitrary dependence among e-values.

When the data are collected sequentially, we consider a sequential analog of e-values to measure the cumulative evidence against the null, 
called e-processes,
defined as follows.
\begin{definition}[E-process]
Given a filtration $\{\mG_t\}_{t\ge 0}$, a $\{\mG_t\}$-adapted nonnegative process 
$\{E_t\}_{t \ge 1}$ is called an e-process for a null hypothesis $H_0$ on $\{\mG_t\}$ if
$E_\tau$ is an e-value for $H_0$, 
for any $\{\mG_t\}$-stopping time $\tau$.
\end{definition}
A sufficient and more commonly used condition for $\{E_t\}$ to be an e-process is 
that $\{E_t\}$ is a nonnegative supermartingale with respect to $\{\mG_t\}$ under $H_0$, 
i.e., $\E_{H_0}[E_t \given \mG_{t-1}] \le E_{t-1}$ for all $t\ge 1$ and $\E[E_0] \le 1$.
More properties of e-processes can be found in~\citet{ramdas2025hypothesis}.

If for each hypothesis $H_k^{(0)}$ we can construct an e-process $\{E_{k,t}\}_{t\ge 1}$
on $\{\mF_t\}$,
then applying e-BH to the stopped e-values---e-processes evaluated at a stopping time---$(E_{1,\tau},\ldots,E_{K,\tau})$,
for any $\{\mF_t\}$-stopping time $\tau$ yields 
a rejection set $\mR_\tau$ with valid FDR control. This gives an immediate route to anytime-valid multiple testing once valid e-processes are available.

\subsection{Additional related work}
\label{app:lit}

\paragraph{Active sampling for multiple testing.}
In addition to the aforementioned work of~\citet{jamieson2018bandit}, 
\citet{haupt2011distilled} study related problems,
with an emphasis on asymptotic FDP control and power analysis.
Although~\citet{xu2021unified} focus primarily on inference, 
they also present a generic sampling 
rule (Appendix C) that prioritizes 
the unrejected arm with the largest e-process value;
we refer to this rule as  
the greedy algorithm. Our simulation shows  
that our proposed method is more efficient than the greedy method. 
\cite{zhou2026robustness} likewise study active sampling for
multiple testing of reward means under adversarial contamination 
and establish sample-complexity guarantees for greedy-type algorithms.
In a different line of work, \citet{gandy2014mmctest,gandy2017quickmmctest,zhang2019adaptive,fischer2026multiple} study adaptive sampling and early stopping in 
Monte Carlo sample, aiming to reproduce results 
based on exact p-values with as few Monte Carlo samples as possible.
More recently, \cite{sandoval2026multi} consider global 
null testing across multiple hypotheses and propose
a UCB-type algorithm for constructing log-optimal 
e-processes. Our work is also broadly related to budget-aware approaches 
for improving computational and resource efficiency \citep{zhao2023adaptive,zrnic2024active,xu2025active,kuang2025active}. 
While these methods focus on combining inexpensive auxiliary statistics 
with selective evidence to enhance power, 
we instead study adaptive budget allocation across hypotheses to reduce sample complexity.



\paragraph{FDR control in online settings.}
Anytime-valid control in our framework builds on 
the e-BH procedure and the broader theory of e-values \citep{wang2022false,vovk2021values,ramdas2025hypothesis}. 
As noted earlier, \citet{xu2021unified} apply e-processes 
to bandit problems and achieve anytime-valid FDR control 
through the e-BH.~\citet{wang2025anytime} study anytime-valid 
FDR control based on stopped e-processes in more general settings,
while~\citet{tavyrikov2026carefree} investigate FDR control with 
the suprema of e-processes.
We also note another line of work on ``online'' FDR control that addresses 
a related but distinct problem from the one studied here: 
hypotheses (possibly infinitely many) arrive sequentially along with their associated $p$-values, and the decision maker must decide whether to reject each hypothesis immediately upon observation~\citep{javanmard2015online,javanmard2018online,ramdas2017online,yang2017framework,zrnic2021asynchronous}.
More recently, \cite{xu2024online,fischer2024online1} have examined this setting from an e-value-based perspective. Nevertheless, both the problem formulation and the inferential goals in these works differ substantially from those considered in our paper.


\section{The e-value-based posterior sampling framework}\label{sec:framework}

\subsection{Our algorithm}

We now introduce the e-PS framework, which 
consists of two components: an inference rule for constructing rejection sets 
from accumulated evidence and a sampling rule that determines which hypothesis to query next.
We describe these components in turn.

\paragraph{E-process and rejection set contruction.}
For each hypothesis $k\in[K]$, 
we construct an e-process in the product form. Specifically, 
at each time step $t \ge 1$, let $e_{k,t}$
be an $\mF_t$-measurable random variable satisfying 
\@\label{eq:e-increment}
(1)~e_{k,t} > 0 \text{ a.s. }
(2)~\E_{H_k^{(0)}}[e_{k,t} \given \mF^-_{t}] \le 1 \text{ a.s. }
(3)~e_{k,t} =1 \text{ if  }A_{t} \neq k,
\@
where $\mF^-_{t} = \sigma(A_1,Y_1,U_1,
\ldots,A_{t-1},Y_{t-1},A_t,U_t)$, 
$\forall t\ge 1$. 
We then define the corresponding e-process as 
\@\label{eq:product}
E_{k,t} = \prod^t_{s=1}e_{k,s}, \quad \forall t \ge 1.
\@
with the convention that $E_{k,0} = 1$.

It is straightforward to check that $E_{k,t} >0$ and $\forall t \ge 1$, 
\begin{align*}
\E_{H_k^{(0)}}[E_{k,t} \given \mF_{t-1}] 
& = E_{k,t-1} \cdot \E_{H_k^{(0)}}[e_{k,t} \given \mF_{t-1}]\\
& = E_{k,t-1}\cdot\bigg(\P_{H_k^{(0)}}(A_t \neq k\given \mF_{t-1}) + 
\E_{H_k^{(0)}}\Big[\ind\{A_t = k\} \cdot 
\E_{H_k^{(0)}}\big[e_{k,t} \given \mF_{t-1}, U_t\big]\Big]\bigg)\\
& \le E_{k,t-1},
\end{align*}
where the second step is because $A_t$ is measurable with respect to 
$\sigma(\mF_{t-1},U_t)$.
This certifies that $\{E_{k,t}\}_{t\ge 1}$ is a nonnegative supermartingale with respect to $\{\mF_t\}$ under $H_k^{(0)}$, 
and therefore an e-process for $H_k^{(0)}$ on $\{\mF_t\}$.
At each time $t\ge 1$, we construct the rejection set $\mR_t$ by applying 
e-BH at level $\alpha$ to $\{E_{1,t},\dots,E_{K,t}\}$, with $\mR_0 = \varnothing$.
Concrete constructions of the increment e-values $\{e_{k,t}\}$ are discussed  
in later sections.
We next describe the sampling strategy. 

\paragraph{Sampling strategy.}


As emphasized by~\citet{tavyrikov2026carefree}, nested rejection sets
are highly desirable in sequential multiple testing:
$\mR_{t} \subseteq \mR_{t+1}$, for all $t\ge 1$, 
ensuring that discoveries are never revoked as additional data is collected.
A procedure satisfying this property is called {\em carefree}.
To achieve this property in our framework, we establish a natural sufficient condition:
if evidence is collected only for hypotheses that have not yet been rejected, 
then the resulting rejection sets are nested. 
The formal statement is given in Proposition~\ref{prop:nest} below.

\begin{proposition}[Nested rejection sets]\label{prop:nest}
Let $\{\mR_t\}_{t\ge1}$ be obtained by applying e-BH 
to the product-form e-processes in~\eqref{eq:product},
with increments satisfying~\eqref{eq:e-increment}.
If $\operatorname{supp}(\pi_t) \subseteq [K]\setminus \mR_{t-1}$ a.s.,
then the rejection sets are nested: $\mR_{t-1} \subseteq \mR_t$.
\end{proposition}
The proof of Proposition~\ref{prop:nest} is given in Appendix~\ref{app:proof-nest}.
We note that several existing algorithms in bandit multiple testing literature
satisfy the above condition, and thereby produce nested rejection sets. For example,
the UCB-type algorithm proposed by~\citet{jamieson2018bandit}, as well as  
the generic sampling methods introduced by~\citet[Appendix C]{xu2021unified}.

With this condition in place, it remains to allocate samples among 
the unresolved hypotheses.
We propose a posterior sampling approach
inspired by Thompson sampling~\citep{thompson1933likelihood}. 
Fix an arm $k\in[K]$. 
At time $t\ge 1$, we denote the number of times arm $k$ has been sampled up to time $t$ as 
\begin{align}
n_{k,t} = \sum^{t}_{s=1} \ind\{A_s = k\}.
\end{align}
In addition, we define  
\$ 
\hat m_{k,t} := \frac{\log(E_{k,t})}{n_{k,t} \vee 1},
\quad \hat{\sigma}_{k,t}^2 := \frac{v_{k,t}}{n_{k,t}\vee 1},
\$
where $v_{k,t}$ are variance proxy parameters to be specified;
For $t = 1,\dots,K$, we initialize the algorithm by letting 
$A_t = t$. For $t > K$, 
we independently draw for each $k\in [K]\setminus\mR_{t-1}$
\@\label{eq:ps-update}
\tilde \mu_{k,t} \sim \mathcal{N}(\hat m_{k,t-1}, \hat{\sigma}_{k,t-1}^2).
\@
The next arm to sample is then selected as $A_t = \argmax_{k \notin \mR_{t-1}} \tilde \mu_{k,t}$,
with ties broken arbitrarily.

We summarize the complete framework in Algorithm~\ref{alg:online-e-BH-ps}
and briefly discuss the intuition behind the proposed sampling strategy.
Let $t_k(\ell)$ be the time step when arm $k$ is sampled for the $\ell$-th time. 
Then under mild conditions,  
\$ 
\frac{\log(E_{k,t})}{n_{k,t}} = \frac{1}{n_{k,t}}\sum^t_{s=1} \log(e_{k,s}) = \frac{1}{n_{k,t}}\sum^{n_{k,t}}_{\ell =1} \log e_{k, t_k(\ell)}
\$ 
has an approximate normal distribution 
when $n_{k,t}$ is large. When we posit a normal prior on the mean log increment, the sampling rule in~\eqref{eq:ps-update} can be interpreted as a posterior sampling strategy,
which allocates more samples to hypotheses with larger posterior mean $\hat m_{k,t}$ and/or larger posterior variance $\hat{\sigma}_{k,t}^2$.


\begin{algorithm}[h!]
\caption{E-value-based posterior sampling (e-PS)}
\label{alg:online-e-BH-ps}
\KwIn{null and alternative hypotheses $\{H_k^{(0)}, H_k^{(1)}\}_{k\in [K]}$, 
target FDR level $\alpha$, 
variance proxy parameters $\{v_{k,t}\}_{k\in[K], t \ge K}$.}
\KwOut{a sequence of rejection sets $\{\mR_t\}_{t\ge 1}$.}
\textbf{Initialization:}
$E_{k,0} = 1$, $\hat m_{k,0} = 0$, 
$n_{k,0} = 0$, 
$\forall k\in[K]$;
$\mR_0 \gets \varnothing$; $t \gets 1$\;
\While{not stopped}{
    \uIf{$t \le K$}{
        $A_t = t$\; 
    }
    \Else{
    draw 
    $\tilde m_{k,t} \sim \mathcal{N}(\hat m_{k,t-1}, 
    \frac{v_{k,t-1}}{n_{k,t-1} \vee 1})$, $\forall k \notin \mR_{t-1}$\;
    choose $A_{t} = \text{argmax}_{k \notin \mathcal{R}_{t-1}} \tilde m_{k,t}$,
    with ties broken arbitrarily\;
    }
    observe $Y_t = Y_t(A_t)$\;
    compute $e_{k,t}$ and update $E_{k,t} \gets E_{k,t-1} \cdot e_{k,t}$, $\forall k \in[K]$\;
    update $n_{k,t} \gets n_{k,t-1} + \ind\{A_t = k\}$, $\forall k \in[K]$\;
    update $\hat m_{k,t} \gets \log(E_{k,t})/(n_{k,t}\vee 1)$, 
    $\forall k\in[K]$\;
    apply e-BH at level $\alpha$ to obtain the rejection set:  
    $\mathcal{R}_t \leftarrow \text{e-BH}(E_{1,t},\ldots,E_{K,t};\alpha)$\;    
    \If{$\mR_t  = [K]$ }
    {break\;}
    $t \gets t + 1$.
}
\end{algorithm}

\subsection{Sample complexity guarantees}

To understand the sample complexity of the proposed framework, we 
analyze the number of samples needed to reliably reject all nonnull hypotheses.
This sample complexity  depends on the properties of the e-processes $\{E_{k,t}\}$:
how quickly they grow under the alternative hypotheses, 
as well as how quickly their empirical growth stabilizes around their asymptotic growth rates. 

To characterize these properties, we introduce the following definitions 
to quantify the tail probabilities of the e-processes.
For any $k\in [K]$, let $\{Z_{k,\ell}\}_{\ell \ge 1}$ denote its 
log-reward process, where $Z_{k,\ell}$ is the log e-value increment 
associated with the $\ell$-th pull of arm $k$. If $t_k(\ell) = \infty$, 
then $Z_{k,\ell}$ denotes the potential log e-value increment that 
would have been observed had sampling continued until 
the pull occurred.
Define 
\$ 
M_{k,\ell} = \sum^\ell_{i=1} Z_{k,i}, \quad \forall \ell \ge 1,
\$
with $M_{k,0} = 0$. Specifically, when arm $k$ is selected at time $t$, 
the algorithm observes $Z_{k,n_{k,t-1}+1} = \log e_{k,t}$.


Next, we define the $\sigma$-fields: for any $\ell \ge 1$, 
\$ 
\mG^{(k)}_{\ell-1} = \mF^-_{t_k(\ell)}. 
\$ 
$\mG^{(k)}_{\ell-1}$ is well defined since $t_k(\ell)$ is a stopping time 
with respect to $\mF_t^-$; 
in addition, $\{\mG^{(k)}_{\ell}\}_{\ell \ge 0}$ forms a filtration.
Next, we define 
\$ 
S_{k,\ell} = \sum^{\ell}_{i=1}\E\big[Z_{k,i} \given \mG^{(k)}_{i-1}\big].
\$
to be the cumulative conditional mean of the log e-value increments for arm $k$ after $\ell$ samples.

For any non-decreasing sequence $\{x_\ell\}_{\ell \ge 1}$
and any non-increasing sequence $\{y_\ell\}_{\ell \ge 1}$,
define the first time at which they become positive and non-positive, respectively:
\$
\mathsf{pos}(\{x_\ell\}) := \inf\{\ell \ge 1: x_\ell >0\}, 
\quad
\mathsf{neg}(\{y_\ell\}) := \inf\{\ell \ge 1: y_\ell \le 0\},
\$
with the convention $\inf \varnothing = +\infty$.
We further define the associated boundary-crossing probabilities
\@\label{eq:tail}
\mathsf{L}_k(\{x_\ell\}) := \P\Big(\exists \ell \ge 1 \text{ s.t. }  \frac{S_{k,\ell}}{\ell} < x_\ell\Big), \quad
\mathsf{U}_k(\{y_\ell\}) := \P\Big(\exists \ell \ge 1 \text{ s.t. } \frac{S_{k,\ell}}{\ell} > y_\ell\Big).
\@
Intuitively, we expect that under the alternative hypothesis, 
$S_{k,\ell}/\ell$ remain above a suitably chosen lower boundary 
that eventually becomes positive, 
so that $\mathsf{L}_k(\{x_\ell\})$ is small. 
Conversely, under the null hypothesis, it should remain 
below a suitably chosen upper bound that eventually 
becomes non-positive, 
making $\mathsf{U}_k(\{y_\ell\})$ small.

We introduce the following assumption on the increment e-values $\{e_{k,t}\}$, 
which is a mild condition satisfied in many common testing problems; see discussions in later sections. 
\begin{assumption}\label{assump:likelihood-e}
For every $k\in [K]$ and $\ell \ge 1$, 
there exists a nonnegative, $\mG^{(k)}_{\ell-1}$-measurable 
variance proxy $s_{k,\ell}^2$
such that 
almost surely
\$ 
\E\Bigg[\exp\Big\{\lambda \cdot \big(Z_{k,\ell} - \E[Z_{k,\ell}\given \mG_{\ell-1}^{(k)}]\big)\Big\} 
\Biggiven \mG_{\ell-1}^{(k)}\Bigg] \le \exp(\lambda^2 s_{k,\ell}^2/2), \quad \forall \lambda \in \R.
\$
\end{assumption}
Under Assumption~\ref{assump:likelihood-e}, for any $k\in [K]$, 
we further define the cumulative variance proxy
$V_{k,\ell} = \sum^\ell_{i=1} s_{k,i}^2$, $\forall \ell \ge 1$
and the probability of the cumulative variance exceeding 
a constant sequence $\{\overline{V}_{k,\ell}\}_{\ell \ge 1}$:
\$ 
\mathsf{V}_{k}(\{\overline{V}_{k,\ell}\}) 
:= \P(\exists \ell \ge 1 \text{ s.t. } V_{k,\ell} > \overline{V}_{k,\ell}).
\$
Finally, for any $v\ge 0$, $c>0$, and any $\ell \ge 1$,
we define the radius function 
\@ \label{eq:rho}
\rho_\ell(v,c) = \sqrt{\frac{2(v + c)}{\ell}\cdot 
\log\bigg(\frac{4K}{\delta}\sqrt{1+v/c}\bigg)}.
\@

\paragraph{A general result.} 
With all the definitions in place, we provide in  
the following general high-probability upper bound 
on the number of samples needed by e-PS to reject all the non-null hypotheses. In the sections to follow, we instantiate this result to concrete examples to illustrate its applications. 

The sample-complexity upper bound depends a user-chosen sequence for each $k\in [K]$: 
a non-decreasing sequence $\{\gamma_{k,\ell}\}_{\ell \ge 1}$ if $k\in \mH_1$
and a non-increasing sequence $\{\eta_{k,\ell}\}_{\ell \ge 1}$ if $k\in \mH_0$.
Roughly speaking, these sequences are chosen to track $S_{k,\ell}/\ell$ so that 
$\sfL_k(\{\gamma_{k,\ell}\})$ is small for $k\in \mH_1$ and $\sfU_k(\{\eta_{k,\ell}\})$ 
is small for $k\in \mH_0$.

Before stating the formal result, we introduce the following set of scalars, which shall determine the sample size required to reject all the non-nulls. 
For each $k\in \mH_1$, define 
\$
& \sfh_k(\sfw,\{\gamma_{k,\ell}\},\{\overline{V}_{k,\ell}\}, c_k, \{\kappa_{k,\ell}\}) 
= \inf\big\{\ell \ge \sfw: 
n \gamma_{k,n}^2 \ge 
9 \cdot (\rho_n(\overline{V}_{k,n},c_k)^2 \vee\kappa_{k,n}^2), \forall n \ge \ell
\big\},\\
& \sfl_k(\sfw,\{\gamma_{k,\ell}\}) 
= \inf\bigg\{\ell \ge  \sfw: 
n \gamma_{k,n} \ge 2\log\Big(\frac{K}{\alpha}\Big),
\forall n \ge \ell \bigg\},\\
& \mathsf{H}_k(\sfw,\{\gamma_{k,\ell}\},\{\overline{V}_{k,\ell}\}, c_k, \{\kappa_{k,\ell}\})
= 2\sum^\infty_{\ell=\sfw} 
\exp\bigg\{-\frac{\ell \gamma_{k,\ell}^2}{144\cdot 
(\rho_\ell(\overline{V}_{k,\ell},c_k)^2  
\vee \kappa_{k,\ell}^2)}\bigg\}.
\$
Here, $\sfw$ is a burn-in period, $\{\overline{V}_{k,\ell}\}_{\ell \ge 1}$ is a positive,  
non-decreasing sequence for variance upper bounds, $c_k>0$ is a constant, 
and $\{\kappa_{k,\ell}\}_{\ell \ge 1}$ is a deterministic non-negative sequence 
for potential variance inflation in the sampling algorithm.
Similarly, for each $k\in \mH_0$, we define 
\$ 
& \mathsf{m}_k(\sfw,\{\gamma_{j,\ell}\}_{j\in \mH_1}, 
\{\eta_{k,\ell}\},\{\overline{V}_{k,\ell}\}, c_k, \{\kappa_{k,\ell}\})  \\
& \hspace{5em }= \inf\bigg\{\ell \ge  \sfw: n (\gamma^* - \eta_{k,n})^2 
\ge 18\cdot (\rho_n(\overline{V}_{k,n},c_k)^2 \vee \kappa_{k,n}^2)
\cdot \log\Big(\frac{4K}{\log 2}\Big),\forall n \ge \ell\bigg\},
\$
where $\gamma^* = \min_{j\in \mH_1, \ell \ge \sfw_j} \gamma_{j,\ell}$.
For notational simplicity, we suppress the arguments in $\sfh_k, \sfl_k, \sfH_k$, 
and $\mathsf{m}_k$ when they are clear from context.
The following theorem uses the above quantities to characterize the number of 
samples needed to reject all non-null hypotheses.




\begin{theorem}
\label{thm:general-power-hp}
Fix $\alpha,\delta \in (0,1)$. 
Consider $K$ hypotheses $\{H_k^{(0)}, H_k^{(1)}\}_{k\in [K]}$,
and assume $\mH_1 \neq \varnothing$.
For each $k\in[K]$, let $\{E_{k,t}\}_{t\ge 1}$ be an e-process
for $H_k^{(0)}$, adapted to $\{\mF_t\}_{t\ge 1}$, in
the form of~\eqref{eq:product} with increment e-values $\{e_{k,t}\}$ satisfying Assumption~\ref{assump:likelihood-e}. 

For any $k\in[K]$, consider a constant $c_k > 0$, 
a deterministic non-negative sequence 
$\{\kappa_{k,\ell}\}_{\ell \ge 1}$ for potential variance inflation,
and a positive and non-decreasing sequence 
$\{\overline{V}_{k,\ell}\}_{\ell \ge 1}$ for variance upper bound.
For each $k \in \mH_1$, consider 
any non-decreasing sequence $\{\gamma_{k,\ell}\}_{\ell \ge 1}$ 
such that $\mathsf{pos}(\{\gamma_{k,\ell}\}) <\infty$ and 
a burn-in period $\sfw_k \ge \mathsf{pos}(\{\gamma_{k,\ell}\}) $; 
for each $k\in \mH_0$, 
consider any non-increasing
sequence $\{\eta_{k,\ell}\}_{\ell \ge 1}$ 
such that $\mathsf{neg}(\{\eta_{k,\ell}\}) <\infty$ and 
a burn-in period $\sfw_k \ge \mathsf{neg}(\{\eta_{k,\ell}\}) $.
The e-PS algorithm with variance parameters 
$v_{k,t} = \rho_{n_{k,t}}(V_{k,n_{k,t}},c_k)^2 \vee\kappa_{k,n_{k,t}}^2$
satisfies
\begin{align}
\P(\TPP(\mR_T) = 1) \ge 1 - \delta - \sum_{k\in \mH_1} \mathsf{L}_k(\{\gamma_{k,\ell}\}) - \sum_{k\in \mH_0} \mathsf{U}_k(\{\eta_{k,\ell}\}) - \sum_{k\in[K]} \mathsf{V}_k(\{\overline{V}_{k,\ell}\}), 
\end{align}
provided that 
\begin{align}
T\gtrsim K + \sum_{j\in \mH_1} 
\Big\{\mathsf{R}_j\cdot (\sfw_j + \log(K/\delta)) +  
\max\{\mathsf{h}_j, \mathsf{l}_j\} 
+ \mathsf{H}_j\Big\} 
+ \sum_{k\in \mH_0} \mathsf{m}_k,
\end{align}
where $\mathsf{R}_j = R(1) \vee 
\max_{1\le \ell < \sfw_j} R\Big(\frac{\sqrt{\ell}}
{\kappa_{j,\ell}}
\cdot(\gamma^* - \gamma_{j,\ell})+1\Big)$ if $\sfw_j >1$ 
and $\mathsf{R}_j = R(1)$ if $\sfw_j = 1$, 
with the convention $0/0 = 0$.
\end{theorem}

The proof of Theorem~\ref{thm:general-power-hp}, 
given in Appendix~\ref{sec:general-power-hp-proof},
centers on the first time at which all non-null hypotheses are rejected:
\begin{align}
    \tau_* = \inf\{t\ge 1: \mH_1 \subseteq \mR_t\}.
\end{align}
Since the rejection sets are nested, 
we have $\mH_1 \subseteq \mR_t$ for all $t\ge \tau_*$, and therefore $\TPP(\mR_T) = 1$ for any $T \ge \tau_*$.
The proof then boils down to upper bounding $\tau_*$ with high probability.
In Appendix~\ref{sec:upper_bound_expected_time}, we also establish a
parallel upper bound on $\E[\tau_*]$.

We remark that our algorithm does not require knowledge of 
$\{\gamma_{k,\ell}\}$ or $\{\eta_{k,\ell}\}$, 
and Theorem~\ref{thm:general-power-hp} holds for all valid choices of these sequences. 
Ideally, one chooses $\{\gamma_{k,\ell}\}$ to be as large as possible 
while keeping $\sfL(\{\gamma_{k,\ell}\})$ small, and $\{\eta_{k,\ell}\}$
to be as small as possible while keeping $\sfU(\{\eta_{k,\ell}\})$ small.
In later sections, we instantiate the above general theory in specific testing problems,
derive explicit increment e-values and growth-rate sequences, 
and bound the resulting terms in Theorem~\ref{thm:general-power-hp}.


As a useful special case, one can always set $\eta_{k,\ell} = 0$ for $k\in \mH_0,\ell\ge 1$, 
noting that 
\$ 
\E\big[\log e_{k,t_k(i)} \given \mG^{(k)}_{i-1}\big] \le 
\log \big(\E[e_{k,t_k(i)} \given \mG^{(k)}_{i-1}] \big)\le 0, \quad \forall i \ge 1,
\$
where the first inequality follows from Jensen's inequality and the second inequality follows from the definition of increment e-values in~\eqref{eq:e-increment}.
As a result, the tail probability $\mathsf{U}_k(\{0\}) = 0$. 


\section{Testing with simple alternatives}\label{sec:likelihood-e}
This section instantiates the general framework and theory of e-PS to 
testing problems with simple alternative hypotheses. We 
consider simple and composite nulls in order.

\subsection{Simple-versus-simple testing}\label{sec:simple_versus_simple_theory}
Suppose $\mP_k = \{P^\circ_k\}$ and $\mQ_k = \{Q^\circ_k\}$ are 
both singletons, 
for some given distributions $P^\circ_{k}$ and $Q^\circ_{k}$, 
with $Q^\circ_{k}$ absolutely continuous with respect to $P^\circ_{k}$.

For this testing problem, it is known that the likelihood ratio process 
achieves the optimal growth rate under the alternative hypothesis~\citep{grunwald2024safe,ramdas2025hypothesis}, 
and therefore is a natural candidate for constructing e-processes.
To be specific, we construct the increment e-values $\{e_{k,t}\}$ as 
\begin{equation}\label{eq:simple-e-val}
\begin{aligned}
    e_{k,t} =\begin{cases}
    \frac{\mathrm{d}Q^\circ_{k}}{\mathrm{d}P^\circ_{k}}(Y_t) & \text{if } A_t = k,\\
    1 & \text{if } A_t \neq k,
\end{cases}
\end{aligned}
\end{equation}
where $\frac{\mathrm{d}Q^\circ_{k}}{\mathrm{d}P^\circ_{k}}$ is the Radon-Nikodym derivative of $Q^\circ_{k}$ with respect to $P^\circ_{k}$. 
The e-process for hypothesis $k$ is then given by $E_{k,t} = \prod_{i=1}^t e_{k,i}$.



To obtain explicit sample-complexity upper bounds for the e-PS algorithm with the above e-processes, 
we invoke Theorem~\ref{thm:general-power-hp} by setting  
\$ 
\gamma_{k,\ell} \equiv D_{\KL}(Q_k^\circ||P_k^\circ) \text{ for } k\in \mH_1, \quad
\eta_{k,\ell} \equiv  -D_{\KL}(P_k^\circ||Q_k^\circ) \text{ for } k\in \mH_0, \quad 
\forall \ell \ge 1.
\$ 
We can then check that $S_{k,\ell} / \ell \equiv \gamma_{k,\ell}$ for $k\in \mH_1$ 
and $S_{k,\ell}/\ell \equiv \eta_{k,\ell}$ for $k\in \mH_0$, 
which leads all the $\mathsf{L}_k$ and $\mathsf{U}_k$ terms to be zero.

In the following theorem, we verify that the likelihood ratio processes defined 
with the increments in~\eqref{eq:simple-e-val} are valid e-processes for the null hypotheses, 
and establish a high-probability upper bound on the number of samples needed to reject all the nonnull hypotheses
for the e-PS algorithm implemented with these e-processes,
whose proof is given in Appendix~\ref{app:proof-ss-power}.

\begin{theorem}[Simple-versus-simple testing]
\label{thm:ss-power-hp}
Fix $\alpha,\delta \in (0,1)$ and $K \ge 2$.
For each $k\in[K]$, suppose $\mP_k = \{P_k^\circ\}$ and  $\mQ_k = \{Q_k^\circ\}$, 
where $Q_k^\circ$ is absolutely continuous with respect to $P_k^\circ$ 
and $Q_k^\circ \neq P_k^\circ$.
\begin{enumerate}
\item[(1)] 
For each $k\in[K]$, 
the process $\{E_{k,t}\}_{t\ge 1}$ defined in~\eqref{eq:product},
with likelihood-ratio e-values 
$\{e_{k,t}\}$ given by~\eqref{eq:simple-e-val}, is a valid e-process for $H_k^{(0)}$ on $\{\mF_t\}$.
\item[(2)] Suppose Assumption~\ref{assump:likelihood-e} holds
with $s^2_{k,\ell} \equiv \sigma_k^2$ for some constant $\sigma_k>0$, $\forall k\in[K]$.
The e-PS algorithm implemented $v_{k,t} = \rho_{n_{k,t}}(n_{k,t}\sigma_k^2, \sigma_k^2)^2$ 
satisfies 
that $\P(\TPP(\mR_T) = 1) \ge 1-\delta$, if 
\begin{align*}
    T \gtrsim K\log(K/\delta) +  & \sum_{k \in \mathcal{H}_1}
    \left\{\frac{\log(K/\alpha)}{D_{\KL}(Q_k^\circ\|P_k^\circ)} 
    + \frac{\sigma_k^2}{D_{\KL}(Q_k^\circ\|P_k^\circ)^2}\Bigg(\log\Big(\frac{K}{\delta}\Big) 
    + \log_+\bigg[\frac{\sigma_k}{D_{\KL}(Q_k^\circ\|P_k^\circ)}\bigg]\Bigg)\right\} \\
    & + \sum_{k \in \mathcal{H}_0} \frac{\sigma_k^2 \log K}{(D_\KL(P_k^\circ\|Q_{k}^\circ) + d_{\mathsf{LR}})^2}
    \cdot \Bigg\{  \log\Big(\frac{K}{\delta}\Big) + 
    \log_+\bigg[\frac{\sigma_k}{d_{\mathrm{LR}} + D_\KL(P_k^\circ \| Q_k^\circ)} \bigg]\Bigg\}.
\end{align*}
where
$d_{\mathsf{LR}} = \min_{j \in \mH_1} D_\KL(Q_j^\circ\|P_{j}^\circ)$.
\end{enumerate}
\end{theorem}
This result characterizes the number of samples required for e-PS to effectively reject
all the nonnulls in simple-versus-simple testing;
this sample complexity depends on the geometric properties of the problem
measured by an aggregated KL-divergence of all the hypotheses.

Two additional remarks are in order.

\begin{remark}[Implications for Gaussian mean detection]
Consider the Gaussian mean detection problem, where 
$P_k^\circ = \mN(0,v_k)$ and $Q_k^\circ = \mN(\theta_k,v_k)$ for 
some $\theta_k \in \mathbb{R}$, and $v_k > 0$.
In this case, the KL divergence is given by 
\$D_{\KL}(Q_k^\circ||P_k^\circ) = D_{\KL}(P_k^\circ||Q_k^\circ)= \frac{\theta_k^2}{2v_k}
\quad \forall k\in[K].
\$ 
The centered log-likelihood ratio is sub-Gaussian with variance proxy $\theta_k^2/v_k$.
Theorem~\ref{thm:ss-power-hp} then simplifies to that 
$\P(\TPP(\mR_T) = 1) \ge 1-\delta$ provided that
\begin{align}
\label{eqn:Gaussian-mean-ub}
    T \gtrsim K\log\Big(\frac{K}{\delta}\Big) + 
\sum_{k\in \mH_1} \frac{v_k}{\theta_k^2} 
\cdot 
\log\Big(\frac{K^2v_k}{\alpha\delta\theta_k^2}\Big)
+ \sum_{k\in \mH_0} \frac{v_k}{\theta_k^2}\log K\cdot 
\log\Big(\frac{Kv_k}{\delta\theta_k^2}\Big).
\end{align}
In the simple scenario, when $\theta_k = \theta$ and $v_k = v$, for every $k \in [K]$, the above inequality reduces to 
$T \gtrsim \frac{K v}{\theta^2} + K$ up to logarithmic factors. 
\end{remark}

\begin{remark}[Sub-Gaussianity of the log-likelihood ratio] 
Theorem~\ref{thm:ss-power-hp} relies on the sub-Gaussianity of the log-likelihood ratio statistics (Assumption~\ref{assump:likelihood-e}).
This assumption is satisfied in many common testing problems. 
For instance, in the Gaussian mean shift problem, where  
$P_k^\circ = \mathcal{N}(0,v)$ and $Q_k^\circ = \mathcal{N}(\theta_k, v)$,
the log-likelihood ratio is given by 
\$
\log \frac{\mathrm{d}Q^\circ_{k}}{\mathrm{d}P^\circ_{k}}(Y_t(k)) = \frac{\theta_kY_t(k)}{v} - \frac{\theta_k^2}{2v}.
\$
Under either $P_k^\circ$ or $Q_k^\circ$, the (centered) log-likelihood ratio 
is sub-Gaussian with variance proxy $\theta_k^2/v$.
More generally, for multivariate distributions 
$P_k^\circ = \mathcal{N}(\bm 0,\Sigma)$ and $Q_k^\circ = \mathcal{N}(\bm \beta_k, \Sigma)$, 
the centered log-likelihood ratio is sub-Gaussian with 
variance proxy $\bm \beta_k^T\Sigma^{-1}\bm \beta_k$. 
Similar conclusions hold for multivariate Student-\(t\) distributions 
with common degrees of freedom, and cases
where the Radon-Nikodym derivative is uniformly bounded---for example---exponential-family models with bounded sufficient statistics, which include Bernoulli, 
Binomial, and categorical distributions with finite support.
\end{remark}

\paragraph{Sample complexity lower bound.}
To complement the upper bound in Theorem~\ref{thm:ss-power-hp}, 
we establish a lower bound on the number of samples needed to 
achieve full discovery with high probability for any algorithm that controls the FDR at level $\alpha$ at all time steps.
\begin{theorem}\label{thm:ss-lower-bound}
Fix a multiple testing problem with singleton null and alternative distribution classes 
\$
\mP_k = \{P_k^\circ\}, \quad  \mQ_k = \{Q_k^\circ\}, \quad \forall k\in[K],
\$
where $Q_k^\circ$ is absolutely continuous with respect to $P_k^\circ$, the target FDR level $\alpha \in (0,1)$,
and the tolerance level $\delta \in (0,1)$.
Then for any algorithm that achieves anytime-valid FDR control at level $\alpha$, and 
$\P(\TPP(\mR_T) = 1) \ge 1-\delta$ for any problem instance with $\mH_1 \neq \varnothing$,  
we must have
\begin{align}
    T \ge 
    \sum_{k \in \mH_1} \frac{\KL_+(1-\delta - (|\mH_1|+1) \cdot \alpha \, ,\, |\mH_1| \cdot \alpha)}{D_{\KL}(Q_k^\circ\|P_k^\circ)}
    + \sum_{k \in \mH_0} \frac{\KL_+(1-\delta - (|\mH_1|+1) \cdot \alpha \, ,\, \delta)}{D_\KL(P_k^\circ\|Q_{k}^\circ)},
\end{align}
where 
\$
\KL_+(p,q) =
\begin{cases}
    p\log(p/q) + (1-p)\log((1-p)/(1-q)) & \text{if } p > q,\\
    0 & \text{if } p \le q.
\end{cases}
\$ 
\end{theorem}
The proof of Theorem~\ref{thm:ss-lower-bound} is given in Appendix~\ref{app:proof-ss-lower-bound}.
We note that in the regime of $\alpha,\delta \rightarrow 0$, 
the lower bound in Theorem~\ref{thm:ss-lower-bound} scales as
\$ 
T \gtrsim \sum_{k \in \mH_1} \frac{\log(1/\alpha)}{D_{\KL}(Q_k^\circ\|P_k^\circ)} + 
\sum_{k \in \mH_0} \frac{\log(1/\delta)}{D_\KL(P_k^\circ\|Q_{k}^\circ)}.
\$
For the Gaussian mean testing problem, this lower bound matches the upper bound in \eqref{eqn:Gaussian-mean-ub} up to logarithmic factors, which implies the optimality of the e-PS algorithm.

\subsection{Composite-versus-simple testing}
We now consider composite nulls while retaining simple alternatives. 
As in the preceding case, a natural building block is the log-optimal e-value, 
defined as the e-values maximizing $\E_{Q_k^\circ}[\log(E_{k})]$ over all valid e-values 
under $\mP_k$. 
Under suitable regularity conditions, 
the log-optimal e-values take the likelihood form 
$E_k = \di Q_k^\circ /\di P^*_k$, where $P_k^*$ is the 
{\em reverse information projection (RIPr)}
of $Q_k^\circ$ onto $\text{conv}(\mP)$.\footnote{Or more generally, 
it is the reverse information projection onto 
the effective null hypothesis of $\mP$ (see exact definition in~\citet{larsson2025numeraire})} 
For concreteness, we present two
examples adapted 
from~\cite{ramdas2025hypothesis}.

\paragraph{One-parameter exponential family.}
Suppose $Y(k) \in \mathbb{R}$ has density 
$p_{\theta}(y) = \exp(\theta \cdot T(y) - A(\theta))$ with respect 
to a common reference measure, where $T(y)$ is a scalar sufficient 
statistic and $A$ is a convex and differentiable function. 
Consider the null hypothesis $H^{(0)}: \theta \in \Theta_0$, where 
$\Theta_0$ is closed with largest element $\theta_0$, 
and a singleton alternative $H^{(1)}: \theta = \theta_1 > \theta_0$. 
Then $P_k^* = p_{\theta_0}$, and 
the log-optimal e-value is 
$E_k = p_{\theta_1}(Y(k))/p_{\theta_0}(Y(k))$.

\paragraph{Monotone likelihood ratio family.}
Suppose $Y_t(k) \in \RR$ has density $p_\theta(y)$ with respect 
to a common reference measure, where $\theta \in \Theta \subseteq \RR$.
Assume that, for any $\theta < \theta'$, 
the likelihood ratio 
$p_{\theta'}(y)/p_{\theta}(y)$ exists and is monotone 
in $y$.
Consider the null $\{\theta \in \Theta: \theta \le \theta_0\}$
and a singleton alternative $\theta_1 >\theta_0$. 
The RIPr is $P_{\theta_0}$ and the log-optimal e-value 
is $E_k = p_{\theta_1}(y)/p_{\theta_0}(y)$.

Further details and examples for constructing log-optimal e-values 
can be found in~\citet{grunwald2024safe,lardy2024reverse,larsson2025numeraire,ramdas2025hypothesis}.
In general, if we can find a distribution $P_k^* \neq Q_k^\circ$---ideally the RIPr---such
that 
\@\label{eq:least-favorable} 
\sup_{P_k \in \mP_k} \E_{P_k}\Big[\frac{\di Q_k^\circ}{\di P_k^*}\Big]\le 1,  
\@ 
then we can construct the e-process increment as
\@\label{eq:composite-e-val} 
e_{k,t} =\begin{cases}
\frac{\di Q_k^\circ}{\di P_k^*}(Y_t) & \text{if } A_t = k,\\
1 & \text{if } A_t \neq k.
\end{cases}
\@
The following theorem shows that the above construction yields 
a valid e-process and provides a characterization of 
the sample complexity of the e-PS procedure.


\begin{theorem}\label{thm:ns-power-hp} 
Fix $\alpha,\delta \in (0,1)$ and $K\ge 2$.
For each $k\in[K]$, 
consider a null distribution set $\mP_k$ and 
a single alternative $\mQ_k = \{Q_k^\circ\}$, 
where $Q_k^\circ \notin \mP_k$ is absolutely continuous with respect to $P \in \mP_k$.
The following holds:
\begin{enumerate}
\item[(1)] For each $k\in[K]$, 
the process $\{E_{k,t}\}_{t\ge 1}$ defined in~\eqref{eq:product},
with increment e-values 
$\{e_{k,t}\}$ given in~\eqref{eq:composite-e-val}, 
is a valid e-process for $H_k^{(0)}$ on $\{\mF_t\}$.
\item[(2)] Suppose Assumption~\ref{assump:likelihood-e} holds
with $s_{k,\ell}^2 \equiv \sigma_k^2$ for some constant $\sigma_k > 0$
and condition~\eqref{eq:least-favorable} holds
$\forall k \in [K]$.
The e-PS algorithm implemented 
with $v_{k,t} = \rho_{n_{k,t}}(n_{k,t}\sigma_k^2,\sigma_k^2)^2$ satisfies
that $\P(\TPP(\mR_T) = 1) \ge 1-\delta$, if 
\begin{align*}
    T \gtrsim &~ K\log\Big(\frac{K}{\delta}\Big) +  \sum_{k \in \mathcal{H}_1}
    \left\{\frac{\log(K/\alpha)}{D_{\KL}(Q_k^\circ\|P_k^*)} 
    + \frac{\sigma_k^2}{D_{\KL}(Q_k^\circ\|P_k^*)^2}\Bigg(\log\Big(\frac{K}{\delta}\Big) 
    + \log_+\bigg[\frac{\sigma_k}{D_{\KL}(Q_k^\circ\|P_k^*)}\bigg]\Bigg)\right\} \\
    & + \sum_{k \in \mathcal{H}_0} \frac{\sigma_k^2 \log(K)}{(d_{\mathsf{LR}} +D_\KL(P_k\|Q_k^\circ) -D_\KL(P_k\|P_k^*))^2}
    \cdot \Bigg\{  \log\Big(\frac{K}{\delta}\Big) + \log_+\bigg[\frac{\sigma_k}{d_{\mathrm{LR}} + 
    D_\KL(P_k \| Q_k^\circ) - D_\KL(P_k\|P_k^*) } \bigg]\Bigg\}
\end{align*}
where
$d_{\mathsf{LR}} = \min_{j \in \mH_1} D_\KL(Q_j^\circ\|P_{j}^*)$.
\end{enumerate}
\end{theorem}
The proof of Theorem~\ref{thm:ns-power-hp} can be found in Appendix~\ref{sec:proof-ns-power-hp}.
Compared to Theorem~\ref{thm:ss-power-hp}, this result relies on how close the alternative hypothesis $Q_k^\circ$ lies to the family of null hypotheses. 




\section{Testing with composite alternatives}\label{sec:sample_complextiy_theory}
We now turn to composite alternatives $\mQ_k$, while for simplicity we 
restrict the null to be simple: 
$\mP_k = \{P_k^\circ\}$ and assume $Q_k \ll P_{k}^\circ$, 
for $\forall Q_k \in \mQ_k$, $\forall k\in [K]$.
Because the alternative distribution is no longer known, 
the log-optimal e-value increments are not directly available.
We instead update the e-value increments $e_{k,t}$  adaptively 
using the observed data.

Concretely, for each $t \ge 1$, we define the plug-in e-value increment  
\begin{align} \label{eq:plugin_eval}
e_{k,t} = \begin{cases}
\frac{\di \hat Q_{k,n_{k,t-1}}}{\di P_k^\circ}(Y_t) & \text{ if }A_t = k\\
1 & \text{ if }A_t \neq k,
\end{cases}
\end{align}
where $\frac{\di \hat Q_{k,n_{k,t-1}}}{\di P_k^\circ}$ is 
an $\mF_{t-1}$-measurable estimate of $\frac{\di P_k}{\di P_k^\circ}$
with $\hat Q_{k,n_{k,t-1}} \ll P_k^\circ$, 
and $\hat Q_{k,0} = P_k^\circ$.
Intuitively, if the estimated likelihood ratio approaches the 
ground truth quickly, then the sample complexity will be close to the likelihood ratio e-process case. 

The following theorem establishes the validity of the e-process, 
as well as the number of samples required to reject all the non-nulls with high probability.
To ease the notation, we define the first time a non-increasing sequence 
$\{x_{\ell}\}_{\ell \ge 1}$ is below a constant $\gamma$:
\$ 
\mathsf{HT}(\{x_{\ell}\}, \gamma) :=
\inf\{\ell \ge 1: x_{\ell} \le \gamma \}.
\$

\begin{theorem}\label{thm:comp-power-hp} 
Fix $\alpha,\delta \in (0,1)$.
For each $k\in[K]$, consider testing the simple null 
$\mP_k = \{P_k^\circ\}$ against a composite alternative $\mQ_k$, where 
$Q_k \ll P_k^\circ$, $\forall k\in[K]$. 
\begin{enumerate}
\item[(1)] For each $k\in[K]$, 
the process $\{E_{k,t}\}_{t\ge 1}$ defined in~\eqref{eq:product},
with increment e-values 
$\{e_{k,t}\}$ given in~\eqref{eq:plugin_eval}, 
is a valid e-process for $H_k^{(0)}$ on $\{\mF_t\}$.
\item[(2)] For any $k\in[K]$, 
suppose Assumption~\ref{assump:likelihood-e} holds 
and that there exists 
a positive and non-decreasing deterministic sequence $\{\overline{V}_{k,\ell}\}_{\ell \ge 1}$, such that 
\$
\P(V_{k,\ell} \le \overline{V}_{k,\ell},~\forall \ell \ge 1) \ge 1-\frac{\delta}{2K},
\$
and that $\rho(\overline{V}_{k,\ell},c_k)/\sqrt{\ell}$ 
is non-increasing in $\ell$.
In addition, we assume that there exists a non-negative 
and non-increasing known sequence $\{\mathsf{r}_{k,\ell}\}_{\ell \ge 1}$ such that 
\$ 
\P\bigg(\frac{1}{\ell}\sum^\ell_{i=2} 
D_{\KL}(P_k \| \hat Q_{k,i-1}) \le \mathsf{r}_{k,\ell}, 
\forall \ell \ge 2\bigg)
\ge 1-\frac{\delta}{2K}.
\$
The e-PS algorithm implemented 
with $v_{k,t} = \rho_{n_{k,t}}(V_{k,n_{k,t}},c_k)^2 \vee 
n_{k,t} \cdot \mathsf{r}_{k,n_{k,t}}^2$ 
satisfies
that $\P(\TPP(\mR_T) = 1) \ge 1-2\delta$, if 
\begin{align*}
T &\gtrsim K\log(K/\delta) + 
\sum_{k\in \mH_0} \mathsf{HT} \bigg(
\big\{\mathsf{r}_{k,\ell}\vee \rho(\overline{V}_{k,\ell},c_k)/\sqrt{\ell}\big\},
\frac{d_\KL}{30\sqrt{\log(4K)}}\bigg)+\\
& \quad \quad \quad \quad  
+\sum_{k \in \mH_1}\Bigg\{ 
\mathsf{HT}\bigg(
\big\{\mathsf{r}_{k,\ell} \vee \rho(\overline{V}_{k,\ell},c_k)/\sqrt{\ell}\big\},
\frac{D_\KL(P_k \| P_k^\circ)}{9}\bigg) 
 + \frac{\log(K/\alpha)}{D_{\KL}(P_k \| P_k^\circ)} + \sfH_k\Bigg\},
\end{align*}
where
\$ 
d_{\KL} = \inf_{j\in \mH_1} D_{\KL}(P_j \| P_j^\circ),\quad 
\sfH_k = \sum^\infty_{\ell = \tilde \sfw_k} 
\exp\bigg\{-\frac{D_\KL(P_k \| P_k^\circ)^2}{1296\cdot (\rho(\overline{V}_{k,\ell},c_k)^2/\ell
\vee \mathsf{r}_{k,\ell}^2)}\bigg\},
\$
with $\tilde w_k = \mathsf{HT}(\{\mathsf{r}_{k,\ell}\}, D_\KL(P_k \| P_k^\circ)/6)$.
\end{enumerate}
\end{theorem}
The proof of Theorem~\ref{thm:comp-power-hp} is provided 
in Appendix~\ref{app:proof-comp-power-hp}. 
The theorem shows that the 
sample complexity of e-PS depends explicitly on 
how closely the plug-in e-values approximate the 
log-optimal e-values and on how well the 
cumulative variance is controlled.

\paragraph{Multivariate Gaussian mean testing.}
To illustrate Theorem~\ref{thm:comp-power-hp}, 
we specialize it to the multivariate Gaussian class. 
For each $k\in[K]$,  $P_k = \mN(\theta_k, \sigma^2_k  I_d)$, where $I_d$ denotes 
the $d\times d$ identity matrix. We consider  
the following hypothesis for each mean: 
\@ \label{eq:gaussian_test}
H_{k}^{(0)}: \theta_k = \theta_k^\circ, 
\quad 
H_{k}^{(1)}: \theta_k \neq \theta_k^\circ,
\@
where the variance parameter $\sigma_k$ is assumed to be known. 
To construct the plug-in e-values, for $\ell \ge 1$, we define 
\$
\hat \theta_{k,\ell} = \frac{1}{\ell }\sum^\ell_{i=1} 
Y_{t_k(i)},\text{ and } \hat \theta_{k,0} = \theta_k^\circ.
\$ 
Next, we define the e-value increment as 
\$ 
e_{k,t} 
& = \ind\{A_t \neq k\} + \ind\{A_t = k\}\cdot 
\exp\Bigg\{-\frac{\|Y_t - \hat \theta_{k,n_{k,t-1}}\|^2}{2\sigma_k^2} + 
\frac{\|Y_t - \theta_k^\circ\|^2}{2\sigma^2_k}\Bigg\}.
\$
In this case, we can explicitly calculate the quantities in Theorem~\ref{thm:comp-power-hp}; the sample complexity bound is summarized in the following corollary.
\begin{corollary}\label{cor:gaussian}
Fix $\alpha,\delta \in (0,1)$ and $K \ge 2$. For $k\in[K]$, suppose 
$P_k = \mN(\theta_k,\sigma_k^2  I_d)$ with the variance parameter $\sigma$ 
assumed to be known and consider the hypothesis in~\eqref{eq:gaussian_test}. 
For $\ell \ge 1$, define 
\$
V_{k,\ell} = \sum_{i=2}^\ell \frac{\|\hat \theta_{k,i-1} - \theta_k^\circ\|_2^2}{\sigma_k^2}
\text{ and }
\mathsf{r}_{k,\ell} = \big(3d + 10\log(K/\delta)\big)\cdot \frac{3+\log \ell}{\ell}.
\$
The e-PS algorithm implemented with 
$v_{k,t} = \rho_{n_{k,t}}(V_{k,n_{k,t}},1)^2 \vee n_{k,t} \mathsf{r}_{k,n_{k,t}}^2$ satisfies that $\P(\TPP(\mR_T) = 1) \ge 1-\delta$
if 
\$ 
T \gtrsim K\log\Big(\frac{K}{\delta}\Big) & + 
\sum_{k \in \mH_1} \frac{(d+\log(K/\delta))\sigma_k^2}
{\|\theta_k - \theta_k^\circ\|^2}
\cdot \Bigg\{\log\Big(\frac{Kd}{\delta}\Big)
+\log_+\bigg(\frac{\sigma_k^2}{\|\theta_k-\theta_k^\circ\|^2}\bigg)\Bigg\}
+ \frac{\log(K/\alpha)\sigma_k^2}{\|\theta_k - \theta_k^\circ\|^2} \\
& + |\mH_0|\cdot \frac{(d+\log(K/\delta))\cdot \sqrt{\log K}}{d_{\mathsf{LR}}}
\cdot \Bigg\{ \log\Big(\frac{Kd}{\delta}\Big)
+\log_+ \bigg(\frac{1}{d_{\mathsf{LR}}}\bigg)\Bigg\},
\$
where $d_{\mathsf{LR}} = \min_{k \in \mH_1} \|\theta_k - \theta_k^\circ\|_2^2/\sigma_k^2$.
\end{corollary}
The proof of Corollary~\ref{cor:gaussian} is provided in Appendix~\ref{appd:proof-gaussian}.

\section{Simulation studies}\label{sec:simulation}
This section collects simulation results comparing our method against 
other baselines. Due to space constraint, we defer results under other settings 
to Appendix~\ref{sec:additional_simulation}.

\paragraph{Simulation setup.}
We consider a testing problem with $K=50$ hypotheses. 
For each $k\in[K]$, 
we consider four parametric distribution families for $P_k$:
\begin{enumerate}[topsep=0pt, itemsep=0pt, parsep=0pt, partopsep=0pt]
\item[(1)] {\em Univariate Gaussian:} $P_{k} = \mN(\theta_k,1)$, 
where $\theta_k \in \mathbb{R}$.
\item[(2)] {\em Multivariate Gaussian:} 
$P_k = \mN(\theta_k \mathbf{1}_5, {\bm I}_5)$, 
where $\theta_k \in \mathbb{R}$.
\item[(3)] {\em Univariate Student-t distribution:} 
$P_{k} = t_{\theta_k,5}$,
where $\theta_k \in \mathbb{R}$ is the noncentrality parameter and 
the distribution has five degrees of freedom.
\item[(4)] {\em Multivariate Student-t distribution:} 
$P_{k} = t_{\theta_k \mathbf{1}_5,5}$,
where $\theta_k \in \mathbb{R}$, 
$\theta_k \mathbf{1}_5$ is the noncentrality parameter
and the distribution has five degrees of freedom. 
\end{enumerate}
For each distribution family, the parameter of interest is $\theta_k$,
and we consider two testing settings.
For the simple-versus-simple setting, 
for each $k\in [K]$, we test
\$
H_k^{(0)}: \theta_k = 0 \text{ versus } H_k^{(1)}: \theta_k = 0.02 \cdot k.
\$
In the simple-versus-composite setting, we test 
\$ 
H_k^{(0)}: \theta_k = 0 \text{ versus } H_k^{(1)}: \theta_k \neq 0.
\$
For each setting, we independently draw $B_k \sim \text{Bern}(0.1)$, 
$\forall k \in [K]$ 
and set $\theta_k = 0.02 k\cdot B_k$.
The $B_k$'s, and hence the parameters $\{\theta_k\}_{k=1}^K$, 
are generated once and 
are then held fixed across all the simulation repetitions.
  
In each setting, we compare e-PS with two baseline methods:
(1) {\em Uniform sampling,} which samples an arm uniformly
at random from the set of unrejected hypotheses at each round; 
and (2) {\em Greedy sampling}~\citep{xu2021unified}, which samples 
the unrejected arm with the largest current e-value.

\paragraph{Implementation details.}
For each sampling method, we use the likelihood-ratio e-process (referred to as 
\texttt{Oracle}) for the simple-versus-simple test and the plug-in likelihood-ratio
e-process (referred to as \texttt{Adaptive}) for the simple-versus-composite test.
The e-PS algorithm is implemented with the variance proxy $v_{k,t} = 1.1 \hat \sigma_{k,t}^2$,
where $\hat \sigma_{k,t}^2$ is the sample variance computed from observations 
collected from arm $k$ up to time $t-1$.
We set the target FDR level to be $\alpha =0.05$ and, at each time point, 
obtain the rejection sets by applying e-BH at level $\alpha$ 
to the current e-values.
For each setting, we repeat the experiment over $500$ independent trials 
with different random seeds, and report the averaged TPP and FDP.

\paragraph{Results.}
The simulation results are shown in Figure~\ref{fig:svs}. 
Across all four distribution families and two e-value contructions, 
e-PS controls the FDR throughout the time horizon.
In terms of TPR, e-PS substantially outperforms the uniform sampling approach
in all settings, and generally achieves a higher TPR than greedy sampling,
with its advantage over the greedy method growing over time.

\begin{figure}[htbp!]
    \centering
\includegraphics[width = \textwidth]{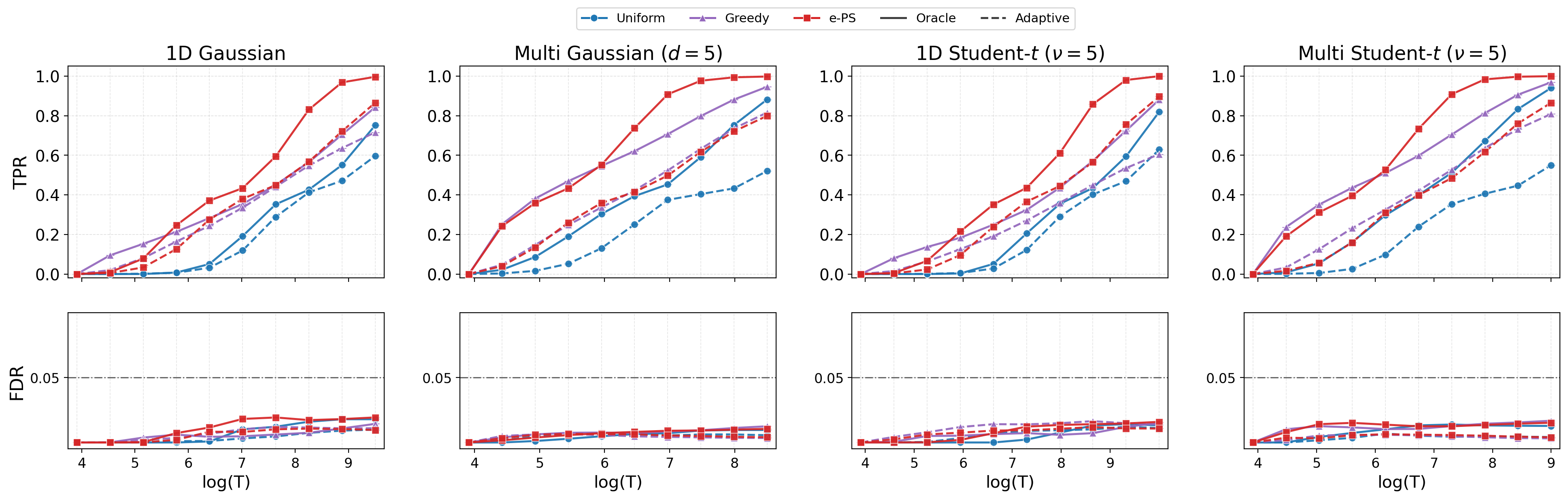}
    \caption{
    Realized TPR and FDR from the simple-versus-simple test (solid)
    and the simple-versus-composite test (dashed)  
    as described in Section \ref{sec:simulation}: (a) 1-dimensional Gaussian, (b) multivariate Gaussian, (c) 1-dimensional student-t, (d) multivariate student-t.
    The results are under an average of 500 independent trials.
    }
\label{fig:svs}
\end{figure}



\section{Real data experiments}\label{sec:realdata}
\subsection{Joke rating mean testing}\label{sec:joke}
We consider a joke rating 
dataset,\footnote{\url{https://www.kaggle.com/datasets/vikashrajluhaniwal/jester-17m-jokes-ratings-dataset}.}
which contains user ratings in $[-10,10]$ for $K=140$ jokes~\citep{goldberg2001eigentaste}.
For each joke $k\in[K]$, we test whether its mean rating is negative:
\begin{align}\label{eq:joke-hypothesis}
    H_k^{(0)}:\E[Y_1(k)] \geq 0
    \qquad \textnormal{versus} \qquad
    H_k^{(1)}:\E[Y_1(k)] < 0.
\end{align}
Although a user may rate multiple jokes, we treat the observed user--rating pairs as independent and view the available ratings for each joke as an empirical sampling distribution. Whenever a joke is queried, we sample one of its available ratings with replacement.




Specifically, within each replication we independently split the available ratings of each joke into two equal parts. The first half is used to determine the ground-truth label: joke $k$ is treated as a non-null if its first-$50\%$ empirical mean is negative. The second half is then used exclusively as the sampling pool during the sequential experiment.
To test~\eqref{eq:joke-hypothesis}, we construct the e-process 
following the betting-based framework of \citet{xu2021unified,waudby2024estimating}:
for each joke $k\in[K]$, at each round $t \ge 1$, we construct
the e-value increment as 
\$
e_{k,t} =\exp\big\{\ind\{A_t = k\}\cdot 
(-\lambda_{k,n_{k,t}}Y_t/10 -\lambda_{k,n_{k,t}}^2/2) \big\},
\$
where $\lambda_{k, \ell} = \sqrt{\frac{2\log (2/\alpha)}{\ell\log(\ell+1)}}$, 
$\forall \ell \ge 1$.
This construction satisfies the conditions in~\eqref{eq:e-increment}
since $-Y_t(k)/10  \in [-1, 1]$ and is therefore $1$-sub-Gaussian.
The e-PS algorithm is implemented with $v_{k,t} = (\sum_{i=1}^{n_{k,t}} \lambda_{k,i}^2)/n_{k,t}$.
We compare the performance of e-PS with uniform sampling and 
greedy sampling described in Section~\ref{sec:simulation}.


We set the FDR target to be $\alpha = 0.1$ 
and run $500$ repetitions, independently reshuffling the data in each replication.
Figure \ref{fig:joke_compare} reports the FDP and TPP averaged over the repetitions as 
a function of the sampling budget.
Across budgets, e-PS maintains the empirical FDR below the target level 
while achieving a higher TPR than the two baselines, thereby identifying more 
non-null hypotheses for the same sampling budget.


\begin{figure}[h!]
    \centering
      \includegraphics[width = 0.6\textwidth]{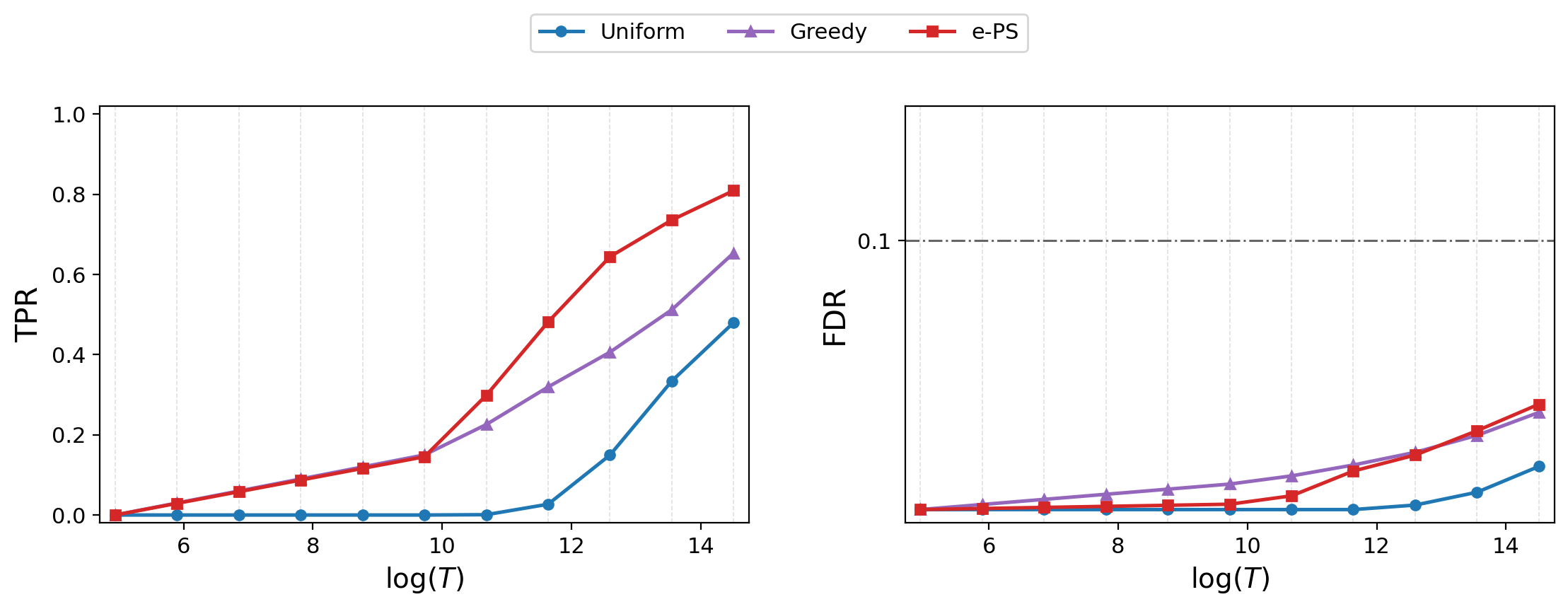}
    \caption{Realized TPR (left) and FDR (right) in the joke rating experiment, 
    averaged across 500 repetitions, as a function of the sampling budget.}
\label{fig:joke_compare}
\end{figure}



\subsection{Watermark detection with FDR control}\label{sec:watermark}
We now consider online watermark detection in large language model (LLM)-generated text, 
following the setup of~\citet{su2026online}. 
Suppose we observe $K$ token streams.
For each $k\in[K]$, let $(w_{k,t})_{t\ge 1}$ denote the $k$-th stream and
let $(\zeta_{k,t})_{t\ge 1}$ denote its associated sequence of pseudo random  vectors. 
Define the stream-specific filtration $\mF_{k,t} = \sigma(w_{k,1},\zeta_{k,1},\dots,w_{k,t},\zeta_{k,t})$, for $t \ge 1$.
To detect whether the $k$-th stream is generated by a watermarked LLM, we test
\$
H_k^{(0)}: w_{k,t} \perp\!\!\!\!\perp \zeta_{k,t} \given \mF_{k,t-1}\quad \text{for all }t \ge 1.
\$
Under the null, the current token is conditionally independent of the pseudorandom variable given the preceding history, 
whereas watermarking induces conditional dependence between them. 
In particular, under $H_k^{(0)}$, $\zeta_{k,t,w_{k,t}} \given \mF_{k,t-1} \sim \operatorname{Unif}(0,1)$ by design.

In our experiment, we generate five independent pools of 
 1000 text-generation streams using OPT-1.3B. 
 Each pool contains 200 watermarked streams generated using the Gumbel-max watermarking mechanism~\citep{Aaronson2023watermark}, along with 800 non-watermarked streams. 
In each pool, texts are generated using 20 prompts, and we consider 5 regularly-spaced temperatures in $[0.5, 0.8]$ for different pools.
The generation lengths are constrained to be between 200 and 1200 tokens. 
In other words, for each prompt-temperature pair, we generate 2 watermarked and 8 non-watermarked streams.

For each token stream $k\in[K]$, 
we construct the e-process  as in Section 4.2 of \citet{su2026online}, 
where the e-increment is 
\$
e_{k,t} = 1- \ind\{A_t = k\} \cdot 
\big(\lambda + \lambda \log(1-\zeta_{k,s,w_{k,s}})\big), \quad \forall t \ge 1
\$
where $\lambda = 0.25$.
To implement e-PS, we use a fixed variance proxy $v_{k,t} = 1.1 \cdot \hat \sigma_{k,t}^2$,
where $\hat \sigma_{k,t}^2$ is the empirical variance of the log e-value increments estimated in
a hold-out set.
We compare e-PS with uniform sampling and greedy sampling. 
At each repetition, we sample $K = 500$ streams without replacement from one of the pools, and all methods are evaluated on the same selected streams.
We set the FDR target level to be $\alpha = 0.1$, performing $100$ repetitions for each of the $5$ independently generated pools, giving a total of  $500$ experiments.

The FDR and TPR, averaged across all repetitions, 
are reported in Figure \ref{fig:watermark_exp_}
as functions of the number of token reads.
We can see that e-PS substantially outperforms uniform sampling. 
Greedy sampling performs better initially, as it quickly rejects the more 
obvious hypotheses; however, the advantage of e-PS increases over time.
This behavior is consistent with our theoretical analysis, 
which focuses on the sample complexity required to reject all nonnull hypotheses: 
while greedy sampling prioritizes hypotheses that are easiest to reject, 
e-PS more effectively allocates samples toward the remaining nonnull hypotheses 
as the experiment progresses.

\begin{figure}[H]
    \centering
      \includegraphics[width = 0.6\textwidth]{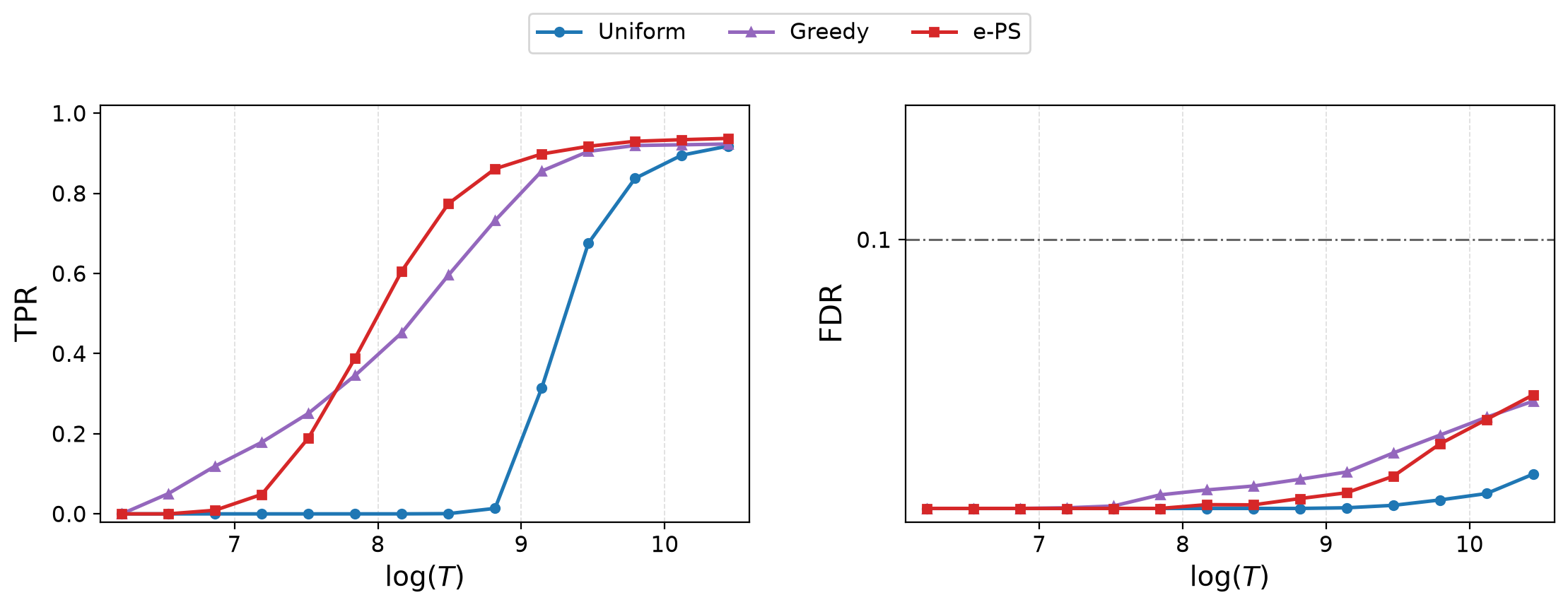}
    \caption{Realized TPR (left) and FDR (right) 
    in the watermark detection experiment, 
    averaged over $500$ repetitions, as a function of number of token-reads. }
\label{fig:watermark_exp_}
\end{figure}

\section{Discussion}
In this work, we propose a general framework for adaptive sampling 
in multiple testing that achieves anytime-valid FDR control and 
sample efficiency at the same time. 
For general testing problems, our algorithm
constructs pseudo reward processes based on e-values and 
allocates samples using posterior sampling. Our e-value-based 
reward processes can also be combined with UCB-type algorithms: 
we evaluate
the empirical performance of such methods in Section~\ref{app:sim-ucb}
while leaving the theoretical analysis to future work.
Our framework uses the e-BH procedure for constructing the rejection set, 
it would be interesting to consider boosted versions of e-BH~\citep{lee2024boosting,xu2025bringing}
for further power improvement.

While our sample-complexity guarantees apply broadly to multiple testing settings with sub-Gaussian log-e-value increments, extending the analysis to heavier-tailed or misspecified models remains an open direction for future work.
Our theoretical guarantees also depend on the quality of the underlying e-process and on tuning choices such as the variance proxy parameter; developing fully data-adaptive calibration procedures is an important avenue for further study. These challenges are particularly salient in human-facing applications, where model misspecification or poorly calibrated adaptive sampling may lead to misleading conclusions. Accordingly, such deployments should be accompanied by careful domain oversight.

\subsection*{Acknowledgment}
Z.~Ren is supported by the National Science Foundation (NSF) 
under grant DMS-2413135 and Wharton Analytics.
Z.~Ren~is thankful to Jelle Goeman, Rianne de Heide, and Ricardo J. Sandoval
for helpful feedback on the project.
Y.~Wei is supported in part by Wharton Dean's Research Fund, the NSF grants CCF-2106778, CCF-2418156 and CAREER award DMS-2143215.  
This work is also supported by the NSF under Cooperative Agreement No. 2433450.
\bibliographystyle{apalike}
\bibliography{refs}

\newpage
\appendix 

\section{Proofs of main results}
\subsection{Proof of Proposition~\ref{prop:nest}}
\label{app:proof-nest}
For any $t \ge 1$, we let $\sigma_t$ denote a permutation on $[K]$ such that
\[
E_{\sigma_t(1),t} \ge \cdots \ge E_{\sigma_t(K),t},
\]
and let $r_t(j) = \sigma_t^{-1}(j)$ denote the rank of $E_{j,t}$
among $\{E_{1,t},\dots,E_{K,t}\}$.
Recall that the e-BH procedure rejects the null hypotheses 
corresponding to the $k_t^*$ largest e-values, where
\[
k_t^* = \max\left\{k: E_{\sigma_t(k),t} \ge \dfrac{K}{\alpha k} \right\}.
\]
Equivalently, the rejection set 
$\mathcal{R}_t = \{k\in[K]: E_{k,t} \ge K/(\alpha k_t^*)\}$. 
When $\mathcal{R}_t = \emptyset$ or $\mathcal{R}_t = [K]$, 
the result trivially holds.
We assume otherwise in the following. 

Since $A_{t+1} \in [K]\backslash \mathcal{R}_t$ and 
only the corresponding e-value $E_{A_{t+1},t}$ 
is updated, we have $r_t(A_{t+1})>k_t^*$
and $E_{k,t+1} = E_{k,t}$, for $\forall k \neq A_{t+1}$. 
We consider the following two cases.
\begin{itemize}
\item If $E_{A_{t+1},t+1} < E_{\sigma_t(k_t^*),t}$, 
then the $k_t^*$ largest e-values at time $t$
remain the $k_t^*$ largest at time $t+1$. As a result, 
\[
E_{\sigma_{t+1}(k_t^*),t+1} = E_{\sigma_t(k_t^*),t} \ge \frac{K}{\alpha k_t^*}, 
\]
where the last step is by the choice of $k_t^*$.
By the e-BH rejection rule, we conclude that $k_{t+1}^* \ge k_t^*$.
For any $k \in \mathcal{R}_t$, $E_{k,t+1} = E_{k,t} \ge \frac{K}{\alpha k_t^*} 
\ge \frac{K}{\alpha k^*_{t+1}}$,
and therefore $\mathcal{R}_t \subseteq \mathcal{R}_{t+1}$. 
\item If  $E_{A_{t+1},t+1} \ge E_{\sigma_t(k_t^*),t}$, then 
the $(k_t^*+1)$-th largest e-value at time $t+1$ coincides with the  
$k_t^*$-th largest at time $t$. To be precise,  we have
\[ 
E_{\sigma_{t+1}(k_t^*+1),t+1} = E_{\sigma_{t}(k_t^*),t} \ge \frac{K}{\alpha k_t^*} 
\ge \frac{K}{\alpha (k_t^*+1)}.
\]
Again by the e-BH rejection rule, $k_{t+1}^* \ge k_t^*+1$.
For any $k\in \mathcal{R}_t$, $E_{k,t+1} = E_{k,t} \ge \dfrac{K}{\alpha k_t^*} 
> \dfrac{K}{\alpha k_{t+1}^*}$.
As a result, $\mathcal{R}_t \subseteq \mathcal{R}_{t+1}$.
\end{itemize}
Combining the two cases, we complete the proof.

\subsection{Proof of Theorem~\ref{thm:general-power-hp}}
\label{sec:general-power-hp-proof}

Define the event 
\$ 
\mathcal{A} = \bigcap_{k \in [K]} \bigcap_{\ell \ge 1} 
\big\{V_{k,\ell} \le \overline{V}_{k,\ell}\big\}, 
\$
and 
$\mathcal{B} = \mathcal{B}_0 \cap \mathcal{B}_1$, where 
\$
\mathcal{B}_0 = 
\bigcap_{k \in \mH_0}\bigcap_{\ell \ge 1}
\Big\{S_{k,\ell}/\ell \le \eta_{k,\ell}\Big\},
\quad 
\mathcal{B}_1 = \bigcap_{k \in \mH_1}\bigcap_{\ell \ge 1}
\Big\{S_{k,\ell}/\ell \ge \gamma_{k,\ell}\Big\}. 
\$
In view of this definition, one has 
\$ 
\P(\cA^c \cup \mathcal{B}^c) \le  
\sum_{k\in \mH_1} \mathsf{L}_k(\{\gamma_{k,\ell}\}) 
+ \sum_{k\in \mH_0} \mathsf{U}_k(\{\eta_{k,\ell}\})
+ \sum_{k \in [K]} \mathsf{V}_k(\{\overline{V}_{k,\ell}\}).
\$
Next, we characterize the concentration of $M_{k,\ell}$ that will be useful for 
the rest of the proof.
For any $k\in [K]$, define for any $\lambda \in \RR$ that  
\$
\mathsf{\Gamma}_{k,\ell}(\lambda) :=  
\exp\big(\lambda(M_{k,\ell}- S_{k,\ell}) - \lambda^2 V_{k,\ell}/2\big).
\$
By Assumption~\ref{assump:likelihood-e}, $\{\mathsf{\Gamma}_{k,\ell}(\lambda)\}_{\ell \ge 1}$
is a supermartingale with respect to $\{\mathcal{G}_\ell^{(k)}\}_{\ell \ge 1}$. 

Mixing $\lambda$ over $\mathcal{N}(0,1/c_k)$ 
for some $c_k>0$ yields
\$
\mathsf{\Gamma}_{k,\ell} = \int \frac{\sqrt{c_k}}{\sqrt{2\pi}} 
\exp\Big(-\frac{c_k \lambda^2}{2}\Big) \cdot \Gamma_{k,\ell}(\lambda) \, \di \lambda 
= \sqrt{\frac{c_k}{V_{k,\ell} + c_k}} 
\cdot \exp\bigg(\frac{(M_{k,\ell} - S_{k,\ell})^2}{2(V_{k,\ell} + c_k)}\bigg),
\$
which is also a supermartingale with respect to $\{\mG_{\ell}^{(k)}\}_{\ell \ge 1}$.
Applying Ville's inequality, we have  
\$
\P\Big(\sup_{\ell \ge 1}~\mathsf{\Gamma}_{k,\ell} \ge 4K/\delta\Big) \le \delta/(4K).
\$
Taking a union bound over $k\in[K]$, 
we have with probability at least $1-\delta/4$ 
that 
\$ 
\big|M_{k,\ell} - S_{k,\ell}\big| & \le \sqrt{2(V_{k,\ell} + c_k) 
\cdot \log\Big(\frac{4K}{\delta}\sqrt{1+V_{k,\ell}/c_k}} \Big)
= \sqrt{\ell}\rho_\ell(V_{k,\ell},c_k) \le \sqrt{\ell}\rho_\ell(\overline{V}_{k,\ell},c_k), 
\quad \forall k\in [K], ~\forall \ell \ge 1.
\$
We use $\mathcal{C}$ to denote the above event.

Since the e-PS algorithm outputs nested rejection sets, it must be the case that  
\$ 
\P(\TPP(\mR_T) < 1) \le \P(\mH_1 \not\subseteq \mR_T) \le \P(\tau^* > T).
\$ 
It therefore suffices to upper bound $\P(\tau^* > T)$ for any $T$ satisfying the condition in the theorem.
We decompose the stopping time by the 
contribution of the non-nulls and the nulls:
\$ 
\tau_* = \sum_{k\in \mH_1}n_{k,\tau_*} + \sum_{k\in \mH_0}n_{k,\tau_*}, 
\$
and establish an upper bound for each term separately.

\paragraph{The contribution of non-nulls.}
Fix $k \in \mH_1$. Define 
\$ 
T_k = \inf\big\{\ell \ge 1:  M_{k,\ell} \ge \log(K/\alpha)\big\}.
\$
By the sampling strategy, we have $n_{k,\tau_*} \le T_k$. 
For any $\ell \ge \max\{\mathsf{h}_k, \mathsf{l}_k\}$, 
we have by definition that $\gamma_{k,\ell}>0$, 
\begin{align}\label{eq:ell_cond}
\sqrt{\ell} \gamma_{k,\ell} \ge 3 \cdot 
(\rho_\ell(\overline{V}_{k,\ell},c_k) \vee \kappa_{k,\ell}),
\text{ and }
\ell \gamma_{k,\ell} \ge 2\log(K/\alpha).
\end{align} 
As a result, for $\ell \ge \max\{\sfh_k,\sfl_k\}$, 
\$ 
\P(n_{k,\tau_*} > \ell, \cA\cap \mathcal{B}\cap \cC) 
& \le \P(T_k > \ell, \cA\cap\mathcal{B}\cap \cC)\\
& \le \P\big(M_{k,\ell} < \log(K/\alpha),\cA\cap \cB\cap \cC\big)\\
& \le \P\big(M_{k,\ell} - S_{k,\ell} < \log ({K}/{\alpha}) -  \ell\gamma_{k,\ell}, \cA\cap \cB \cap \cC\big)\\
& \le \P\big(M_{k,\ell} - S_{k,\ell} < -  \ell\gamma_{k,\ell}/2, \cA\cap \cB \cap \cC\big)=0.
\$
where the last two steps are due to the choice of $\ell$. 
As a result, on the event $\cA \cap \cB \cap \cC$, we have 
\$ 
n_{k,\tau^*} \le \max\{\sfh_k,\sfl_k\}, ~\forall k \in \mH_1.
\$



\paragraph{The contribution of nulls.}
We now turn to the contribution of nulls.  
Recall that $\gamma_{*} = \min_{j\in \mH_1,\ell \ge \sfw_j}  \gamma_{j,\ell}$. 
For any $k\in \mH_0$ and $\ell \ge 1$, we let  
$\Delta_{k,\ell} = \gamma_{*} - \eta_{k,\ell}$ 
and define the following two (sequences of) thresholds  
\$ 
\sfa_{k,\ell} = \eta_{k,\ell} 
+ \frac{\Delta_{k,\ell}}{3},\quad \sfb = \frac{2\gamma_{*}}{3}.
\$
We can check that $\eta_{k,\ell} < \sfa_{k,\ell} < \sfb < \gamma_{*}$,
for $\ell \ge \sfw_k$.

Using these thresholds, we decompose the  number of  
pulls of $k \in \mH_0$ as
\$ 
\sum_{t \le \tau_*} \ind\{A_t = k\}
& \le 1 + \sum_{K< t \le \tau_*} \ind\{A_t = k, \hat m_{k,t-1} > \sfa_{k,n_{k,t-1}}\}
+ \sum_{K< t \le \tau_*} \ind\{A_t = k, \hat m_{k,t-1} \le \sfa_{k,n_{k,t-1}}, 
\tilde m_{k,t} > \sfb\}\\
& \qquad \qquad + \sum_{K< t \le \tau_*} \ind\{A_t = k, 
\hat m_{k,t-1} \le \sfa_{k,n_{k,t-1}}, \tilde m_{k,t} \le \sfb\}\\
& =:1 + T_{1,k} + T_{2,k} + T_{3,k}.
\$
We establish an upper bound for each term separately.

\paragraph{Term $T_{1,k}$:} 

For $T_{1,k}$, 
\$
T_{1,k} & \le \sum^\infty_{t=K+1} 
\ind\{A_t = k, \hat m_{k,t-1} > \sfa_{k,n_{k,t-1}}\} \\
& =\sum^\infty_{t=K+1}\sum^\infty_{\ell=1} 
\ind\{A_t = k,n_{k,t-1} = \ell, M_{k,\ell}/\ell > \sfa_{k,\ell}\}\\
& = \sum_{\ell = 1}^\infty  \ind\{M_{k,\ell}/\ell > \sfa_{k,\ell}\}
\sum_{t = K+1}^\infty \ind\{A_t = k, n_{k,t-1} = \ell\}\\
& \le \sum^\infty_{\ell=1} 
\ind\{M_{k,\ell}/\ell > \sfa_{k,\ell}\}.
\$
On the event $\cA \cap \mathcal{B}\cap \mathcal{C}$,   
\$
T_{1,k} \le \sum^\infty_{\ell = 1} \ind\{M_{k,\ell} > \ell \sfa_{k,\ell}\}
& \le \sum^\infty_{\ell = 1} \ind\{M_{k,\ell} - S_{k,\ell}> \ell \Delta_{k,\ell}/3\}
\le \mathsf{m}_k.
\$
The last step is because when $\ell \ge \mathsf{m}_k$, 
\$ 
M_{k,\ell} - S_{k,\ell} \le \sqrt{\ell}\rho_\ell(V_{k,\ell})
\le \sqrt{\ell}\rho_\ell(\overline{V}_{k,\ell},c_k)
\le \ell \Delta_{k,\ell}/3.
\$

\paragraph{Term $T_{2,k}$:} 
For any $k\in \mH_0$ and $t > K$, when $n_{k,t-1} = \ell \ge \sfw_k$
and $\hat m_{k,t-1} \le \sfa_{k,n_{k,t-1}}$, 
\@\label{eq:t2}
\P(\tilde m_{k,t} > \sfb\given \mF_{t-1}) 
& \le \P\Bigg(\frac{\rho_\ell(V_{k,\ell},c_k) \vee \kappa_{k,\ell}}{\sqrt{\ell}} 
\cdot \xi_{k,t} > \sfb - \sfa_{k,\ell}\bigggiven \mF_{t-1}\Bigg) \qquad \xi_{k,t} 
\stackrel{\text{iid}}{\sim} \mN(0,1) \notag\\
& \le \P\Bigg(\frac{\rho_\ell(V_{k,\ell},c_k) \vee \kappa_{k,\ell}}{\sqrt{\ell}} \cdot \xi_{k,t} 
> \Delta_{k,\ell}/3\bigggiven \mF_{t-1}\Bigg)\notag \\
& \le \exp\Bigg\{-\frac{\ell \Delta_{k,\ell}^2}
{18 \cdot \rho_\ell(V_{k,\ell},c_k)^2 \vee \kappa_{k,\ell}^2}\Bigg\}
\@
where the second-to-last step is because for $\ell \ge \sfw_k$,  
\$ 
\sfb - \sfa_{k,\ell} = 
\frac{2 \gamma_{*}}{3} - \eta_{k,\ell} - \frac{\Delta_{k,\ell}}{3} =
\frac{\Delta_{k,\ell} - \eta_{k,\ell}}{3} \ge \frac{\Delta_{k,\ell}}{3}. 
\$

We now consider 
\$
T_2 &:= \sum_{k\in \mH_0} T_{2,k}\\ 
& \le \sum_{k\in \mH_0} \mathsf{m}_k +   
\sum_{k\in \mH_0} \sum_{t = K+1}^{\tau_*} 
\ind\{A_t = k, n_{k,t-1}\ge \mathsf{m}_k,  \hat m_{k,t-1} \le \sfa_{k,n_{k,t-1}}, \tilde m_{k,t} > \sfb\}
\\
& = \sum_{k\in \mH_0} \mathsf{m}_k +  
\sum_{t = K+1}^{\tau_*} \ind\{A_t \in \mH_0, n_{A_t,t-1} \ge \mathsf{m}_{A_t},
\hat m_{A_t,t-1} \le \sfa_{A_t,n_{A_t,t-1}}, \tilde m_{A_t,t} > \sfb\}. 
\$
Next, we consider the event of $\mathcal{A}$ restricted to 
$\mF_{t-1}$: 
\$
\mathcal{A}_t = \bigcap_{k\in[K]}\bigcap_{\ell \le n_{k,t-1}} 
\{V_{k,\ell} \le \overline{V}_{k,\ell}\}.
\$
Define $\overline{\sfB}_{t} =\ind\{A_t \in \mH_0, n_{A_t,t-1} \ge \mathsf{m}_{A_t}, 
\hat m_{A_t,t-1} \le \sfa_{A_t,n_{A_t,t-1}}, \tilde m_{A_t,t} > \sfb\}$
and $\sfB_t = \overline{\sfB}_t \cdot \ind\{\mathcal{A}_t\}$.
For any $\lambda > 0$ and $t > K$ 
\$
\Xi_t(\lambda) = \exp\Bigg(\lambda \sum^t_{s=K+1} \sfB_s - \epsilon_0 \cdot (e^\lambda-1)(t-K) \Bigg),
\$
with $\Xi_K(\lambda)=1$ and $\epsilon_0 = \log 2/ 4$. 
We can check for any $t>K$ that 
\@\label{eq:t2_mg}
\E[\Xi_t(\lambda) \given \mF_{t-1}]
& = \Xi_{t-1}(\lambda) \cdot \E\Big[\exp\{\lambda \sfB_t - \epsilon_0 (e^\lambda-1)\} \given \mF_{t-1}\Big] \notag \\
& =\Xi_{t-1}(\lambda) \cdot \exp\big(-\epsilon_0(e^\lambda-1)\big) \cdot 
\Big(e^\lambda \cdot \P(\sfB_t = 1 \given \mF_{t-1}) + \P(\sfB_t = 0 \given \mF_{t-1})\Big) \notag \\
& = \Xi_{t-1}(\lambda) \cdot \exp\big(-\epsilon_0(e^\lambda-1)\big) \cdot 
\Big(\big(e^\lambda-1\big) \cdot \P(\sfB_t = 1 \given \mF_{t-1}) + 1)\Big).
\@
Above, 
\$
\P(\sfB_t = 1 \given \mF_{t-1}) & = 
\sum_{k \in \mH_0}\ind\{n_{k,t-1} \ge \mathsf{m}_k, \hat m_{k,t-1} \le \sfa_{k,n_{k,t-1}},\mathcal{A}_t\} 
\cdot  \P(A_t = k, \tilde m_{k,t} > \sfb \given \mF_{t-1}) \\
& \le \sum_{k \in \mH_0} \frac{\epsilon_0}{K} \le \epsilon_0,
\$
where the second-to-last step uses~\eqref{eq:t2} and the choice of $\mathsf{m}_k$ and that $V_{k,\ell} \le \overline{V}_{k,\ell}$
on $\mA_t$.
As a result, there is 
\$
\eqref{eq:t2_mg} & \le \Xi_{t-1}(\lambda) \cdot \exp\big(-\epsilon_0(e^\lambda-1)\big) 
\cdot \big(\big(e^\lambda-1\big) \cdot \epsilon_0 + 1\big)\\
& \le \Xi_{t-1}(\lambda) \cdot \exp\big(-\epsilon_0(e^\lambda-1)\big)
\cdot \exp\big(\epsilon_0(e^\lambda-1)\big)
= \Xi_{t-1}(\lambda),
\$
where we use the inequality $1+x \le e^x$ for any $x\in \RR$ in the second inequality.
We have therefore verified that $(\Xi_t(\lambda))_{t \ge K}$ is a supermartingale adapted to the filtration $\{\mF_t\}_{t \ge K}$.
Applying Ville's inequality, we have for $\lambda > 0$ ,
\$
\P\Bigg(
\exp\Bigg(\lambda \sum^{\tau^*}_{s=K+1} \overline{\sfB}_s - 
\epsilon_0 \cdot (e^\lambda-1)(\tau^*-K) \Bigg) > 4/\delta, \mathcal{A}\Bigg)
\le \P\bigg(\sup_{t \ge 1}\Xi_{t}(\lambda) > 4/\delta \bigg) 
\le \E[\Xi_K(\lambda)] \cdot \frac{\delta}{4} = \frac{\delta}{4},
\$
where we use the fact that $\overline{B}_t = B_t$ for any $t\ge 1$
on $\mA$.
Rearranging the inequality 
gives that 
\$
\sum_{s=K+1}^{\tau_* } \overline{\sfB}_s \le \frac{1}{\lambda}\cdot \Big\{ 
\log(4/\delta) + \epsilon_0 (e^\lambda-1) \cdot \tau_* \Big\},  
\$
with probability at least $1-\delta/4$.
Putting everything together and taking $\lambda = \log 2$, 
we have with probability at least $1-\delta/4$,
\$
T_2 
\le \sum_{k \in \mH_0} \mathsf{m}_k + \frac{\log(4/\delta)}{\log 2} + \tau_*/4. 
\$

\paragraph{Term $T_{3,k}$:}  We now turn to $T_{3,k}$. Instead of bounding $T_{3,k}$ for 
each $k\in \mH_0$ separately, we directly target the 
total contribution of nulls $T_3 := \sum_{k\in \mH_0} T_{3,k}$, 
for which we have the following decomposition: 
\$ 
T_3 & =  \sum_{k\in \mH_0} \sum_{t=K+1}^\infty \ind\{t \le \tau_*, 
A_t = k, \hat m_{k,t-1} \le \sfa_{k,n_{k,t-1}}, \tilde m_{k,t} \le \sfb\}\\
& = \sum_{t=K+1}^\infty \ind\{t \le \tau_*, A_t \in \mH_0, 
\hat m_{A_t,t-1} \le \sfa_{A_t,n_{A_t,t-1}}, \tilde m_{A_t,t} \le \sfb\}.
\$
On the event $\{t \le \tau_*\}$, we have $\mH_1 \not\subseteq \mR_{t-1}$, and therefore 
there exists some $j\in \mH_1$ such that $j \in \mathcal{C}_{t-1} :=  [K]\backslash \mR_{t-1}$. 
Define $J_t = \min\{j\in \mH_1: j \in \mathcal{C}_{t-1}\}$, i.e., the non-null hypothesis with the smallest index 
that has not been rejected yet. 

For any $t > K$, 
\$ 
& \ind\{t \le \tau_*, A_t \in \mH_0, 
\hat m_{A_t,t-1} \le \sfa_{A_t,n_{A_t,t-1}}, 
\tilde m_{A_t,t} \le \sfb\}\\
& =\sum_{j\in \mH_1} \ind\{J_t = j, t \le \tau_*, A_t \in \mH_0, 
\hat m_{A_t,t-1} \le \sfa_{A_t,n_{A_t,t-1}}, 
\tilde m_{A_t,t} \le \sfb\}\\
& =: \sum_{j\in \mH_1} \sfC_{j,t}.
\$ 
Using the above notation, we have for any $x_j >0$ that 
\@\label{eq:t3-hp-decomp}
\P\Bigg(T_3 > \sum_{j\in \mH_1} x_j, \cA \cap \mathcal{B} \cap \cC\Bigg) 
& = \P\Bigg(\sum_{t=K+1}^\infty \sum_{j\in \mH_1} \sfC_{j,t} > \sum_{j\in \mH_1} x_j, \cA \cap \mathcal{B} \cap \cC\Bigg) \notag \\
& \le\sum_{j\in \mH_1} \P\Bigg(\sum_{t=K+1}^\infty \sfC_{j,t} > x_j, \cA \cap \mathcal{B}\cap \cC\Bigg).
\@
On the event $\cA \cap \cB \cap \cC$, we have 
\$ 
\big(M_{j,\ell}/\ell - \gamma_{j,\ell}\big)
\ge -\frac{\rho_\ell(V_{j,\ell},c_j)}{\sqrt{\ell}} \ge 
-\frac{\rho_\ell(V_{j,\ell},c_j) \vee \kappa_{j,\ell}}{\sqrt{\ell}}, 
\quad \forall j\in \mH_1, \forall \ell \ge 1.
\$
Next, we let $\sfD_{j,t} = \ind\{A_t = j\}$ and 
\$
\sfO_{j,t} = \frac{\P(\tilde m_{j,t} \le \sfb \given \mF_{t-1})}
{\P(\tilde m_{j,t} > \sfb \given \mF_{t-1})}
= R\bigg(\sqrt{\frac{n_{j,t-1}}{v_{j,t-1}}} \cdot 
\big(\sfb - \hat m_{j,t-1}\big)\bigg)
= R\bigg(\frac{\sqrt{n_{j,t-1}}}{\rho_{n_{j,t-1}}(V_{j,n_{j,t-1}},c_j) \vee \kappa_{j,n_{j,t-1}}} \cdot 
\big(\sfb - \hat m_{j,t-1}\big)\bigg).
\$
On $\cA \cap \cB \cap \cC$, we have 
\$
& \frac{\sqrt{n_{j,t-1}}}{\rho_{n_{j,t-1}}(V_{j,n_{j,t-1}},c_j) \vee \kappa_{j,n_{j,t-1}}}
\cdot (\sfb - \hat m_{j,t-1})\\
=~&   
\frac{\sqrt{n_{j,t-1}}}{\rho_{n_{j,t-1}}(V_{j,n_{j,t-1}},c_j) \vee \kappa_{j,n_{j,t-1}}}
\cdot (\sfb - \gamma_{j,n_{j,t-1}}) + 
\frac{\sqrt{n_{j,t-1}}}{\rho_{n_{j,t-1}}(V_{j,n_{j,t-1}},c_j) \vee \kappa_{j,n_{j,t-1}}}
\cdot(\gamma_{j,n_{j,t-1}} - \hat m_{j,t-1})\\
\le~& 
\frac{\sqrt{n_{j,t-1}}}{\kappa_{j,n_{j,t-1}} }
\cdot (\sfb - \gamma_{j,n_{j,t-1}})_+ +  1. 
\$ 
We therefore have for any $t > K$,  
\$ 
\sfO_{j,t} \le \begin{cases}
\mathsf{R}_j& n_{j,t-1} < \sfw_j,\\
R(1) & n_{j,t-1} \ge \sfw_j.
\end{cases}
\$
Letting $\lambda^- = 
\log\Big(1+\frac{1}{2\mathsf{R}_j}\Big)$ and 
$\lambda^+ =  \log\Big(1+\frac{1}{2R(1)}\Big)$,  
we define 
\$
& L_{j,t}^- = \exp\Bigg\{\sum_{s=K+1}^t
\ind\{n_{j,s-1} <  \sfw_j, \sfO_{j,s} \le \sfR_j\}\cdot 
\bigg(\lambda^- \cdot \sfC_{j,s} + 
\log \big(1 - \sfO_{j,s}(e^{\lambda^-}-1)\big) \cdot \sfD_{j,s}\bigg)\Bigg\},\\
& L_{j,t}^+ = \exp\Bigg\{\sum_{s=K+1}^t
\ind\{n_{j,s-1} \ge \sfw_j,\sfO_{j,s} \le R(1)\}\cdot 
\bigg(\lambda^+ \cdot \sfC_{j,s} + 
\log \big(1 - \sfO_{j,s}(e^{\lambda^+}-1)\big) \cdot \sfD_{j,s}\bigg)\Bigg\},
\$ 
where we note that $\sfO_{j,t}$ is adapted to $\mF_{t-1}$.
Next, since $\sfC_{j,t} \sfD_{j,t} = 0$, we can check 
\@ \label{eq:t3-hp}
\E[L_{j,t}^- \given \mF_{t-1}] 
& = L_{j,t-1}^- \cdot 
\E\Bigg[\exp\bigg\{\ind\{n_{j,t-1} <  \sfw_j, \sfO_{j,t} \le \sfR_j\}
\cdot \Big(\lambda^- \cdot \sfC_{j,t} + 
\log\big(1 - \sfO_{j,t}(e^{\lambda^-} - 1)\big) \cdot \sfD_{j,t}\Big)\bigg\} \bigggiven \mF_{t-1}\Bigg]\notag\\ 
& = L_{j,t-1}^- \cdot \Bigg\{\ind\{n_{j,t-1} \ge  \sfw_j\text{ or }\sfO_{j,t}>\sfR_j\} 
+ \ind\{n_{j,t-1} < \sfw_j, \sfO_{j,t} \le \sfR_j\}\times \notag \\
& \hspace{7em}\Big(1 + \P(\sfC_{j,t}=1 \mid \mF_{t-1}) \cdot (e^{\lambda^{-}} - 1) 
- \P(\sfD_{j,t}=1 \mid \mF_{t-1}) \cdot \sfO_{j,t}(e^{\lambda^{-}}-1)\Big)\Bigg\}.
\@
On the event $\{t \le \tau_*, J_t = j\}$, there is 
\$ 
\P(A_t \in \cH_0, \hat m_{A_t,t-1} \le \sfa_{A_t,n_{A_t,t-1}}, 
\tilde m_{A_t,t} \le \sfb  \mid \mF_{t-1})
& \le \P\Big(\max_{\ell \in \cC_{t-1}} \tilde m_{\ell,t} \le \sfb \mid \mF_{t-1}\Big) \\
& =  \P(\tilde m_{j,t} \le \sfb \mid \mF_{t-1}) \prod_{\ell \in \cC_{t-1}\backslash\{j\}} 
\P(\tilde m_{\ell,t} \le \sfb \mid \mF_{t-1}).
\$
On the same event, we also have 
\$ 
\P( A_t = j \mid \mF_{t-1})
& \ge 
\P(\tilde m_{j,t} > \sfb , \max_{\ell \in \cC_{t-1}\backslash\{j\}} \tilde m_{\ell,t} \le \sfb \mid \mF_{t-1})\notag\\
& =\P(\tilde m_{j,t} > \sfb \mid \mF_{t-1}) \prod_{\ell \in \cC_{t-1}\backslash\{j\}} \P(\tilde m_{\ell,t} \le \sfb \mid \mF_{t-1}).
\$
Combining the above, we have
\$ 
\P(\sfC_{j,t} = 1 \mid \mF_{t-1}) & =
\ind\{t \le \tau_*, J_t = j\} \cdot \P(A_t \in \cH_0, 
\hat m_{A_t,t-1} \le \sfa_{A_t,n_{A_t,t-1}}, 
\tilde m_{A_t,t} \le \sfb \mid \mF_{t-1})\\
& \le \sfO_{j,t} \cdot \P(A_t = j \given \mF_{t-1}) = \sfO_{j,t} 
\cdot \P(\sfD_{j,t} = 1\given \mF_{t-1}).
\$
As a result, 
\$ 
\eqref{eq:t3-hp} &\le L^-_{j,t-1} \cdot 
\Bigg\{\ind\{n_{j,s-1} \ge \sfw_j \text{ or }\sfO_{j,s} > \sfR_j\} 
+ \ind\{n_{j,s-1} <  \sfw_j, \sfO_{j,s} \le \sfR_j\}
\times \\
& \hspace{10em} \Big(1 + \P(\sfD_{j,t}=1\given \mF_{t-1}) \cdot \sfO_{j,t} 
\cdot (e^{\lambda^-} - 1)
- \P(\sfD_{j,t}=1 \mid \mF_{t-1}) \cdot \sfO_{j,t} \cdot (e^{\lambda^-}-1)\bigg)\Bigg\}\\
&= L_{j,t-1}^-.
\$
The above verifies that $\{L_{j,t}^-\}_{t \ge 1}$ is a supermartingale adapted to the filtration $\{\mF_t\}_{t \ge 1}$.
Similarly, we can verify that $\{L_{j,t}^+\}_{t \ge 1}$ is a supermartingale 
adapted to $\{\mF_t\}_{t \ge 1}$.

On $\cA \cap \cB \cap \cC$, 
we have the following decomposition
\$ 
\sum^\infty_{t = K+1} \sfC_{j,t} = 
\sum^\infty_{t = K+1} \sfC_{j,t}\cdot\ind\big\{n_{j,t-1} <\sfw_j, \sfO_{j,t} \le \sfR_j\big\} + 
 \sum^\infty_{t = K+1} \sfC_{j,t}\cdot\ind\big\{n_{j,t-1} \ge \sfw_j, \sfO_{j,t} \le R(1)\big\}.
\$
Returning to~\eqref{eq:t3-hp-decomp}, we have for any $x_j^->0$ 
and any $N \in \mathbb{N}_+$ that 
\begin{align}\label{eq:t3-hp-odds-25}    
& \P\Bigg(\sum_{t=K+1}^{\tau^* \wedge N} 
\sfC_{j,t} \cdot \ind\{n_{j,t-1} <  \sfw_j, \sfO_{j,t} \le \sfR_{j}\} > x_j^-, 
\cA \cap \cB\cap\cC\Bigg)\notag\\
& = \P\Bigg(\exp\bigg\{\lambda^- \sum_{t=K+1}^{\tau^* \wedge N} 
\sfC_{j,t} \cdot 
\ind\{n_{j,t-1} <  \sfw_j,\sfO_{j,t} \le \sfR_j\}\bigg\} > \exp(\lambda^- x_j^-), \cA \cap \cB\cap\cC
\Bigg)\notag\\
& = \P\Bigg(L_{j,\tau^* \wedge N}^- 
\exp\bigg\{-\sum_{s=K+1}^{\tau^* \wedge N} 
\ind\big\{n_{j,s-1} < \sfw_j, \sfO_{j,t} \le \sfR_j\big\}
\cdot \log\big(1 - \sfO_{j,s}(e^{\lambda^-}-1)\big) \cdot \sfD_{j,s}\bigg\} > \exp(\lambda^- x_j^-), 
 \cA \cap \cB\cap\cC\Bigg)\notag\\
& \le \P\Bigg(L_{j,\tau^* \wedge N}^- \exp\bigg\{-\sum_{t=K+1}^{\infty}
\ind\big\{n_{j,t-1} <  \sfw_j, \sfO_{j,t} \le \sfR_j \big\} \cdot
\log\big(1 - \sfO_{j,s}(e^{\lambda^-}-1)\big) \cdot \sfD_{j,s}\bigg\} > \exp(\lambda^- x_j^-),
\cA \cap \cB\cap\cC \Bigg).
\end{align}
On $\cA \cap \cB \cap \cC$, there is 
\@\label{eq:t3-hp-odds}
& \sum_{t=K+1}^\infty -\log\big(1 - \sfO_{j,t}(e^{\lambda^-}-1)\big) \cdot \sfD_{j,t}
\cdot \ind\big\{n_{j,t-1} <  \sfw_j, \sfO_{j,t} \le \sfR_j\}\notag \\
\le ~&  \frac{\log 2}{\mathsf{R}_j}
\sum_{t=K+1}^\infty \sfO_{j,t} \cdot \sfD_{j,t}\cdot 
\ind\big\{n_{j,t-1} < \sfw_j, \sfO_{j,t} \le \sfR_j\}\notag \\
\le ~& \frac{\log 2}{\mathsf{R}_j}\sum_{\ell < \sfw_j} \mathsf{R}_j
= \log 2 \cdot ({\mathsf{w}}_j-1),
\@
where the first step uses the fact that $-\log(1-x) \le 2\log 2 \cdot x$ for any $x \in [0,1/2]$. 
Using the above, we have 
\$ 
\eqref{eq:t3-hp-odds-25} 
\le \P\Big(L_{j,\tau^*\wedge N}^- \cdot \exp(\log 2 \cdot ({\mathsf{w}}_j-1)) > \exp(\lambda^- x_j^-)\Big)
\le \exp\big(-\lambda^- x_j^- + \log 2 \cdot ({\mathsf{w}}_j-1)\big). 
\$
On the other hand, we have for any $x_j^+ >0$ and any 
$N \in \mathbb{N}_+$ that 
\$ 
& \P\Bigg(\sum^{\tau^* \wedge N}_{t=K+1} \sfC_{j,t} 
\ind\{n_{j,t-1} \ge  \sfw_j, \sfO_{j,t} \le R(1)\} > x_j^+ ,
\cA \cap \cB \cap \cC\Bigg)\\
= ~& \P\Bigg(\exp\bigg(\lambda^+\sum^{\tau^* \wedge N}_{t=K+1} \sfC_{j,t} 
\ind\{n_{j,t-1} \ge \sfw_j,\sfO_{j,t}\le R(1)\}\bigg) 
> \exp(\lambda^+ x_j^+), \cA \cap \cB \cap \cC\Bigg) \\
= ~& \P\Bigg(L_{j,\tau^*\wedge N}^+ 
\cdot \exp\bigg\{ -\sum^{\tau^* \wedge N}_{s=K+1}  
\ind\{n_{j,s-1} \ge  \sfw_j, \sfO_{j,s} \le R(1)\}\cdot 
\log\big(1-\sfO_{j,s}(e^{\lambda^+}-1)\big) \cdot \sfD_{j,s}\bigg\} > \exp(\lambda^+ x_j^+), 
\cA \cap \cB \cap \cC\Bigg)\\
\le ~& \P\Bigg(L_{j,\tau^* \wedge N}^+\cdot \exp\bigg\{-\sum^\infty_{t=K+1} 
\ind\{n_{j,t-1} \ge  \sfw_j,\sfO_{j,t}\le R(1)\}\cdot 
\log\big(1-\sfO_{j,t}(e^{\lambda^+}-1)\big) \cdot \sfD_{j,t}\bigg\} > \exp(\lambda^+ x_j^+), 
\cA \cap \cB \cap \cC\Bigg).
\$
On $\cA \cap \cB \cap \cC$, we have 
\@\label{eq:t3-hp-odds-3}
& \sum^\infty_{t = K+1} -\log\big(1 - \sfO_{j,t}(e^{\lambda^+}-1)\big) 
\cdot \sfD_{j,t} \cdot \ind\{n_{j,t-1} \ge  \sfw_j,\sfO_{j,t} \le R(1)\}\notag\\
\le~& \frac{\log 2}{R(1)} \sum^\infty_{t=K+1} 
\sfO_{j,t}\sfD_{j,t} \cdot  \ind\{n_{j,t-1} \ge  \sfw_j, \sfO_{j,t}\le R(1)\}\notag\\
\le~& \frac{\log 2}{R(1)} 
\sum_{\ell = {\mathsf{w}}_j}^\infty 
R\Bigg(\frac{\sqrt{\ell}}{\rho_{\ell}(V_{j,\ell},c_j) \vee \kappa_{j,\ell}}
\cdot\big (\sfb - \gamma_{j,\ell}\big) + 1 \Bigg)\notag\\
\le~& \frac{\log 2}{R(1)}\sum_{\ell = {\mathsf{w}}_j}^\infty 
R\Bigg(-\frac{\sqrt{\ell}}{\rho_{\ell}(\overline{V}_{j,\ell},c_j) \vee \kappa_{j,\ell}}
\cdot\frac{\gamma_{j,\ell}}{3} + 1 \Bigg),
\@
where the last step uses that 
$\sfb \le 2\gamma_{j,\ell}/3$, for $\ell \ge  \sfw_j$.
Recalling the definition of $\sfh_k$, we have 
\$ 
\eqref{eq:t3-hp-odds-3} & \le \log 2 \cdot 
(\mathsf{h}_j - 1)
+\frac{2\log 2}{R(1)} 
\sum_{\ell = \mathsf{h}_j}^\infty 
\exp\Bigg\{-\frac{1}{2}\Bigg(1-\frac{\sqrt{\ell}\gamma_{j,\ell}}
{3(\rho_{\ell}(\overline{V}_{j,\ell},c_j) \vee \kappa_{j,\ell})}\Bigg)^2\Bigg\}\\
& \le \log 2  \cdot (\mathsf{h}_{j}-1)
+ \frac{2\log 2}{R(1)}\sum_{\ell = \mathsf{h}_{j}}^\infty 
\exp\Bigg\{-\frac{\ell \gamma_{j,\ell}^2}{36(\rho_{\ell}(\overline{V}_{j,\ell},c_j) \vee \kappa_{j,\ell})^2} + \frac{1}{2}\Bigg\} \\
& \le \log 2 \cdot ( \mathsf{h}_j -1) +  
\mathsf{H}_j. 
\$
Above, the first step follows from $R(z) \le 2e^{-z^2/2}$ for $z< 0$,
the second step uses $(a-1)^2 \ge a^2/2 - 1$ for $a \in \mathbb{R}$. 
Combining the above and letting $N\rightarrow \infty$ gives 
\$ 
& \P\Bigg(\sum_{t=K+1}^{\infty} \sfC_{j,t} > x_j, \cA \cap \cB \cap \cC\Bigg)
\le \exp\Big( -\lambda^- x_j^- +   \log 2 \cdot ({\mathsf{w}}_j-1)\Big)
+ \exp\Big(-\lambda^+ x_j^+ + \log2 \cdot (\mathsf{h}_j-1) + \mathsf{H}_j\Big).
\$
Taking 
$x_j^- = (\log 2\cdot ({\mathsf{w}}_j-1) + \log(4K/\delta))/\lambda^-$
and $x_j^+ = (\log 2 \cdot (\mathsf{h}_j-1)+\mathsf{H}_j +\log(4K/\delta))/\lambda^+$, 
we have
\$ 
\P\Bigg(\sum^\infty_{t = K+1} \mathsf{C}_{j,t} > x_j^++x_j^-, \cA \cap \cB \cap \cC\Bigg) \le \frac{\delta}{2K}.
\$
Taking a union over $j\in \mH_1$, we have 
\$ 
\P\bigg(T_3 > \sum_{j\in \mH_1}
\frac{\log 2 \cdot ({\mathsf{w}}_j-1) +\log(4K/\delta)}{\log(1+1/(2\mathsf{R}_j))} + 
\frac{\log 2 \cdot (\mathsf{h}_j-1)+\mathsf{H}_j +\log(4K/\delta)}{\log(1+1/(2R(1)))},
\cA \cap \cB\cap \cC \bigg) \le \frac{\delta}{2}
\$
Putting everything together, we have 
with probability at least 
$1-\delta - \sum_{j\in \mH_1}\sfL_j(\{\gamma_{j,\ell}\}) - \sum_{k\in \mH_0} 
\sfU_k(\{\eta_{k,\ell}\}) - \sum_{k\in[K]} \mathsf{V}_k(\{\overline{V}_{k,\ell}\})$ that 
\$
\tau^* & \lesssim  K + 
\sum_{j \in \mH_1}\Big\{ 
\mathsf{R}_j \cdot \big( \sfw_j + \log(K/\delta)\big) + 
\max\{\mathsf{h}_j, \sfl_j\}+ \mathsf{H}_j\Big\} + 
\sum_{j \in \mH_0} \mathsf{m}_j.
\$
The proof is complete.

\subsection{Proof of Theorem~\ref{thm:ss-power-hp}}
\label{app:proof-ss-power}
\paragraph{Proof of (1).}
Fix $k \in [K]$. By construction, $E_{k,t} \ge 0$. 
For any $t\ge 1$, we have
\$ 
\E_{H_k^{(0)}}[E_{k,t} \given \mF_{t-1}] & = E_{k,t-1} \cdot \E_{H_k^{(0)}}[e_{k,t}\given \mF_{t-1}]\\
& = E_{k,t-1} \cdot \E_{H_k^{(0)}}\Big[\ind\{A_t = k\} \cdot \frac{\mathrm{d}Q^\circ_{k}}{\mathrm{d}P^\circ_{k}}(Y_t(k)) + \ind\{A_t \neq k\} \cdot 1 \given \mF_{t-1}\Big]\\
& = E_{k,t-1} \cdot \Bigg(\P(A_t = k \given \mF_{t-1}) \cdot \E_{P^\circ_k}\bigg[\frac{\mathrm{d}Q^\circ_{k}}{\mathrm{d}P^\circ_{k}}(Y_t(k))\bigg] + \P(A_t \neq k \given \mF_{t-1})\Bigg)\\
& = E_{k,t-1}.
\$
This certifies that $\{E_{k,t}\}_{t\ge 1}$ is a nonnegative supermartingale with respect to $\{\mF_t\}$ under $H_k^{(0)}$ with $E_{k,0}=1$. 
Therefore, $\{E_{k,t}\}_{t \ge 1}$ is an e-process for $H_k^{(0)}$ on $\{\mF_t\}$.    

\paragraph{Proof of (2).} 
To instantiate Theorem~\ref{thm:general-power-hp}, for any $\ell \ge 1$, 
we take 
\$
\gamma_{k,\ell} = D_\KL(Q_k^\circ \| P_k^\circ) > 0, k \in \cH_1, \quad
\eta_{k,\ell} = -D_\KL(P_k^\circ \| Q_k^\circ) < 0, k \in \cH_0.
\$
Since $\eta_{k,\ell}$ and $\gamma_{k,\ell}$ are both constants, 
we drop the dependence on $\ell$ for simplicity.
For each $k \in [K]$, we let $c_k = \sigma_k^2$, $\overline{V}_{k,\ell} = V_{k,\ell} = \sigma_k^2 \ell$, 
and $\kappa_{k,\ell} = 0$, $\forall \ell \ge 1$. With this choice, we also have 
\$
\rho_\ell(V_{k,\ell},c_k) = \sqrt{2\Big(1+\frac{1}{\ell}\Big)\sigma_k^2 
\cdot \log\Big(\frac{4K}{\delta} \sqrt{1+\ell}\Big)}.
\$
We next focus on the non-nulls.
For $k\in \cH_1$, we have 
\$ 
\E_{Q_k^\circ}\Big[\log \frac{\di Q_k^\circ}{\di P_k^\circ}\big(Y_t(k)\big) \Biggiven \mF_{t-1}, A_t =k\Big]
= D_{\KL}(Q_k^\circ \| P_k^\circ) > 0. 
\$ 
As a result, for any $\ell \ge 1$, 
\$ 
S_{k,\ell} = \sum_{i=1}^\ell
\E[\log e_{k, t_k(i)} \given \mG^{(k)}_{i-1}] \equiv 
\ell \cdot D_{\KL}(Q_k^\circ\|P_k^\circ) = \ell \cdot \gamma_{k}.
\$
By definition, $\sfL_{k}(\{\gamma_{k,\ell}\}) = 0$ and $\sfw_k = 1$.
We then take $\sfw_k =1$, and as a result, $\sfR_j = R(1)$
and $\gamma^* = d_{\mathsf{LR}}$.
In addition, we have 
\$
\sfh_k \lesssim 
\frac{\sigma_k^2}{D_\KL(Q_k^\circ \| P_k^\circ)^2}\cdot 
\bigg\{\log\Big(\frac{K}{\delta}\Big)
+ \log\bigg[\frac{\sigma_k}{D_\KL(Q_k^\circ \| P_k^\circ)} \bigg]\bigg\}
\quad 
\text{ and }
\quad 
\sfl_k = \frac{2\log(K/\alpha)}{D_\KL(Q_k^\circ\|P_k^\circ)}.
\$
For $\mathsf{H}_k$, letting $c_+ = 1 \vee 288 \sigma_k^2/\gamma_k^2$, we have  
\$ 
\mathsf{H}_k 
= 2 \sum_{\ell=1}^\infty 
\exp\Big(-\frac{\ell \gamma_k^2}{144\sigma^2_k \log(8K\ell/\delta)}\Big) 
& \le 2 \cdot \Big\{c_+ +
\int_{c_+}^\infty e^{-\frac{\gamma_k^2 x}{144\sigma_k^2\log(8Kx/\delta) }} \di x\Big\}\\
& \stackrel{\mathrm{(a)}}{\le} 2 \cdot \Bigg\{e^{-\frac{\gamma_k^2}{144\sigma_k^2\log(8K/\delta)}}\vee\frac{288 \sigma_k^2}{\gamma_k^2} +
\frac{576 \sigma_k^2 }{\gamma_k^2}\bigg(\log\Big(\frac{8K}{\delta}\Big) + \log_+
\Big(\frac{288\sigma_k^2}{\gamma_k^2}\Big)\bigg)\Bigg\}\\
& \stackrel{\mathrm{(b)}}{\lesssim} \frac{\sigma_k^2}{\gamma_k^2}
\cdot\bigg(\log\Big(\frac{K}{\delta}\Big) + \log_+\Big(\frac{\sigma_k^2}{\gamma_k^2}\Big)\bigg)\\
& = \frac{\sigma_k^2}{D_\KL(Q_k^\circ \| P_k^\circ)^2} \cdot 
\bigg(\log\Big(\frac{K}{\delta}\Big) + \log_+\Big(\frac{\sigma_k}{D_{\mathsf{KL}}(Q_k^\circ\|P_k^\circ) }\Big) \bigg), 
\$
where step (a) uses Lemma~\ref{lem:integral}
and step (b) follows from the elementary inequality $e^{-x} \le 1/x$ for any $x > 0$.

We now turn to the nulls. 
For $k\in \cH_0$, we have
\$ 
\E_{P_k^\circ}\Big[\log \frac{\di Q_k^\circ}{\di P_k^\circ}\big(Y_t(k)\big) \Biggiven \mF_{t-1}, A_t = k\Big] = -D_{\KL}(P_k^\circ\|Q_k^\circ) < 0,
\$
As a result, for any $\ell \ge 1$,
\$
S_{k,\ell} = \sum_{i=1}^\ell
\E[\log e_{k, t_k(i)} \given \mG^{(k)}_{i-1}] \equiv \ell \cdot \eta_{k,\ell},
\$
which leads to $\sfU_k(\eta_k) = 0$ and $\mathsf{pos}(\eta_k)= 1$. We take $ \sfw_k = 1$.
and have
\$
\mathsf{m}_k \lesssim 
\frac{\log K  \cdot \sigma_k^2}{(d_{\mathrm{LR}} + D_\KL(P_k^\circ \| Q_k^\circ))^2} 
\cdot \Bigg\{  \log\Big(\frac{K}{\delta}\Big) + \log_+\Big(\frac{\sigma_k}{d_{\mathrm{LR}} + D_\KL(P_k^\circ \| Q_k^\circ)} \Big)\Bigg\}
\$
Taking these together and invoking Theorem~\ref{thm:general-power-hp}
gives the desired result.

\subsection{Proof of Theorem~\ref{thm:ss-lower-bound}}
\label{app:proof-ss-lower-bound}
Fix $T\ge 1$ and any algorithm that achieves anytime-valid FDR control at level $\alpha$, and
$\P(\TPP(\mR_T) = 1) \ge 1-\delta$ for any problem instance.
By definition, we have 
\$ 
\P(\mH_1 \subseteq \mR_T) \ge 1-\delta.
\$
On the event  $\{\mR_T \cap \mH_0 \neq \varnothing, \mH_1 \subseteq \mR_T\}$, 
we have at least one false discovery, 
and therefore 
\$ 
\FDP(\mR_T) = \frac{|\mR_T \cap \mH_0|}{|\mR_T \cap \mH_1| + |\mR_T \cap \mH_0|}
\ge \frac{1}{|\mR_T \cap \mH_1| + 1} = \frac{1}{|\mH_1|+1},
\$ 
where the last equality follows from the fact that $\mH_1 \subseteq \mR_T$.
As a result, 
\$ 
\P\big(\mR_T \cap \mH_0 \neq \varnothing, \mH_1 \subseteq \mR_T\big) 
\le \big(|\mH_1|+1\big) \cdot \E[\FDP(\mR_T)] \le (|\mH_1|+1) \cdot \alpha,
\$
where we have used the anytime-valid FDR control property in the last inequality.
Combining the above two displays, we have
\$ 
\P(\mR_T = \mH_1) = \P(\mH_1 \subseteq \mR_T) - \P(\mR_T \cap \mH_0 \neq \varnothing, \mH_1 \subseteq \mR_T) 
\ge 1-\delta - (|\mH_1|+1) \cdot \alpha.
\$
Note that the algorithm satisfies the FDR control and power guarantee for any problem instance.
In particular, fixing $j\in \mH_1$, 
we can consider the problem instance with $\mH_1^{(j)} := \mH_1 \backslash \{j\}$.
We use $\P^{(j)}$ to denote the probability measure under the problem instance with $\mH_1^{(j)}$.
Under $\P^{(j)}$, the null hypothesis $H_j^{(0)}$ is true, and therefore on the event $\{\mR_T = \mH_1\}$, 
\$ 
\FDP(\mR_T) = \frac{1}{|\mH_1|}.
\$
As a result, 
\$ 
\P^{(j)}(\mR_T = \mH_1) \le |\mH_1| \cdot \E^{(j)}[\FDP(\mR_T)] \le |\mH_1| \cdot \alpha,
\$
where we have used the anytime-valid FDR control property under $\P^{(j)}$ in the last inequality.

On the other hand, we have 
\$ 
D_{\KL}(\P^T \,\|\, \P^{(j),T})
= \E\bigg[\log \frac{\di \P}{\di \P^{(j)}}(A_1,Y_1,\ldots,A_T,Y_T)\bigg]
= \E[n_{j,T}] \cdot D_{\KL}(Q_j^\circ \,\|\, P_j^\circ).
\$
By the data processing inequality for KL divergence,
\$ 
D_{\KL}\big(\P(\mR_T = \mH_1) \,\|\, \P^{(j)}(\mR_T = \mH_1)\big) \le D_{\KL}(\P^T \,\|\, \P^{(j),T})
= \E[n_{j,T}] \cdot D_{\KL}(Q_j^\circ \,\|\, P_j^\circ),
\$
Rearranging the above inequality, we have
\$ 
\E[n_{j,T}] \ge \frac{D_{\KL}(\P(\mR_T = \mH_1) \,\|\, \P^{(j)}(\mR_T = \mH_1))}{D_{\KL}(Q_j^\circ \,\|\, P_j^\circ)}
\ge \frac{\KL_+( 1-\delta - (|\mH_1|+1) \cdot \alpha, |\mH_1| \cdot \alpha)}{D_{\KL}(Q_j^\circ \, \| \,P_j^\circ)}.
\$

Similarly, for any $k\in \mH_0$, we can consider the problem instance with $\mH_1^{(k)} := \mH_1 \cup \{k\}$, 
and use $\P^{(+k)}$ to denote the probability measure under this problem instance.
Under $\P^{(+k)}$, the alternative hypothesis $H_k^{(1)}$ is true,  
\$ 
\P^{(+k)}(\mR_T = \mH_1) \le \P^{(+k)}(\mH_1\cup \{k\} \not\subseteq \mR_T) \le \delta,
\$
while $\P(\mR_T = \mH_1) \ge 1-\delta - (|\mH_1|+1) \cdot \alpha$.
We then have 
\$ 
D_{\KL}(\P^T \,\|\, \P^{(+k),T})
= \E\bigg[\log \frac{\di \P}{\di \P^{(+k)}}(A_1,Y_1,\ldots,A_T,Y_T)\bigg] 
= \E[n_{k,T}] \cdot D_{\KL}(P_k^\circ \,\|\, Q_k^\circ).
\$
By the data processing inequality for KL divergence, we have
\$
D_{\KL}\big(\P(\mR_T = \mH_1) \,\|\, \P^{(+k)}(\mR_T = \mH_1)\big) \le D_{\KL}(\P^T \,\|\, \P^{(+k),T})
= \E[n_{k,T}] \cdot D_{\KL}(P_k^\circ \,\|\, Q_k^\circ),
\$
which leads to
\$ 
\E[n_{k,T}] \ge \frac{D_{\KL}(\P(\mR_T = \mH_1) \,\|\, \P^{(+k)}(\mR_T = \mH_1))}{D_{\KL}(P_k^\circ \,\|\, Q_k^\circ)} 
\ge \frac{\KL_+(1-\delta - (|\mH_1|+1) \cdot \alpha, \delta)}{D_{\KL}(P_k^\circ \,\|\, Q_k^\circ)}.
\$
Summing over $j\in \mH_1$ and $k\in \mH_0$, we have
\begin{align*}
T & = \sum_{j\in \mH_1} \E[n_{j,T}] + \sum_{k\in \mH_0} \E[n_{k,T}]\\
& \ge \sum_{j\in \mH_1} \frac{\KL_+( 1-\delta - (|\mH_1|+1) \cdot \alpha, |\mH_1| \cdot \alpha)}{D_{\KL}(Q_j^\circ \, \| \,P_j^\circ)} + \sum_{k\in \mH_0} \frac{\KL_+(1-\delta - (|\mH_1|+1) \cdot \alpha, \delta)}{D_{\KL}(P_k^\circ \,\|\, Q_k^\circ)}.
\end{align*}

\subsection{Proof of Theorem~\ref{thm:ns-power-hp}}\label{sec:proof-ns-power-hp}
\paragraph{Proof of (1).}
Fix $k\in[K]$. By construction, $E_{k,t} \ge 0$ for all $t\ge 1$.
For any $t\ge 1$, we have for any $P \in \mP_k$ that 
\$
\E_{P}\big[E_{k,t} \given \mF_{t-1}, A_t\big]
& = E_{k,t-1} \cdot \E_P[e_{k,t} \given \mF_{t-1}, A_t\big] \\ 
& = E_{k,t-1} \cdot \bigg(\ind\{A_t \neq k\}  + 
\ind\{A_t = k\} \cdot \E_P\Big[\frac{\di Q_k^\circ}{\di P_k^*}(Y_t(k)) \biggiven \mF_{t-1},A_t\Big] \bigg)\\
& = E_{k,t-1} \cdot \bigg(\ind\{A_t \neq k\}  + 
\ind\{A_t = k\} \cdot \E_P\Big[\frac{\di Q_k^\circ}{\di P_k^*}(Y_t(k))\Big] \bigg)\\
& \le E_{k,t-1},
\$
where the last step is by construction.
This shows that $\{E_{k,t}\}_{t\ge 0}$ is a nonnegative supermartingale under any $P \in \mP_k$
with $\E_{k,0} = 1$, and therefore an e-process.

\paragraph{Proof of (2).}
The proof is similar to that of 
Theorem~\ref{thm:ss-power-hp} (2), with 
$P_k^\circ$ replaced by $P_k^*$ and noting that 
\$
0 \ge \E_{P_k}\Big[\log \frac{\di Q_k^\circ}{\di P_k^*}\Big]
= \E_{P_k}\Big[\log \frac{\di Q_k^\circ}{\di P_k}\Big]
+\E_{P_k}\Big[\log \frac{\di P_k}{\di P_k^*}\Big]
= - D_{\KL}(P_k \| Q_k^\circ) + D_{\KL}(P_k \| P_k^*). 
\$
The rest
is omitted for brevity.

\subsection{Proof of Theorem~\ref{thm:comp-power-hp}}
\label{app:proof-comp-power-hp}
\paragraph{Proof of (1).}
Fix any $k\in[K]$.
By construction, $E_{k,t} > 0$ and $E_{k,0}=1$.
For any $t \ge 1$, 
\$
\E_{P_k^\circ}[E_{k,t} \given \mF_{t-1}] & = 
E_{k,t-1}\cdot 
\E_{P_k^\circ}\bigg[\ind\{A_t \neq k\} + \ind\{A_t = k\} 
\frac{\di \hat Q_{k,n_{k,t-1}}}{\di P_{k}^\circ} \Biggiven \mF_{t-1} \bigg]\\
& =E_{k,t-1} \cdot \Bigg\{\P_{P_k^\circ}(A_t \neq k\given \mF_{t-1}) + 
\E_{P_k^\circ}\bigg[\ind\{A_t = k\} \cdot \E_{P_k^\circ}
\bigg[ \frac{\di \hat Q_{k,n_{k,t-1}}(Y_t(A_t))}{\di P_{k}^\circ} \Biggiven \mF_{t-1},A_{t}\bigg]
\bigggiven \mF_{t-1}\bigg]\Bigg\}\\
& = E_{k,t-1}.
\$ 
In the last step, we use the fact that $\hat Q_{k,t}$ is predictable.
Together, it is verified that $\{E_{k,t}\}_{t\ge 1}$ is an e-process.

\paragraph{Proof of (2).}
For any $k\in \mH_1$, let 
\$ 
\gamma_k^* := 
D_{\KL} \big(P_k \| P_k^\circ\big) > 0
\$
For any $\ell \ge 2$,  we have 
\$ 
\E\big[Z_{k,\ell} \given \mG^{(k)}_{\ell-1}\big] 
& = \E\Bigg[\log \frac{\di \hat Q_{k,\ell-1}}{\di P_k^\circ}(Y_{t_k(\ell)}(k)) \bigggiven \mG^{(k)}_{\ell-1} \Bigg] \\
& = \E\Bigg[\log \frac{\di \hat Q_{k,\ell-1}}{\di P_k}(Y_{t_k(\ell)}(k)) \bigggiven \mG^{(k)}_{\ell-1} \Bigg] 
+ \E\Bigg[\log \frac{\di P_{k}}{\di P_k^\circ}(Y_{t_k(\ell)}(k)) \bigggiven \mG^{(k)}_{\ell-1} \Bigg]\\ 
& = \gamma_k^* - D_{\KL}(P_k \| \hat Q_{k,\ell-1}), 
\$
and $\E[Z_{k,1}] =0$.
By assumption, we have with probability at least $1-\delta/K$ that
$V_{k,\ell} \le \overline{V}_{k,\ell}$, $\forall \ell \ge 1$ and 
\$ 
\frac{1}{\ell}\sum^{\ell}_{i=2} D_{\KL}(P_k \| \hat Q_{k,i-1}) \le \mathsf{r}_{k,\ell},
\quad \forall \ell \ge 2.
\$
We work on the above event in the following.
To instantiate Theorem~\ref{thm:general-power-hp}, 
we take 
\$
\gamma_{k,\ell} = \begin{cases}
0 \wedge (\gamma_k^*/2 -  \mathsf{r}_{k,2}) & \ell  = 1\\ 
\gamma_k^*/2 -  \mathsf{r}_{k,\ell} & \ell > 1,
\end{cases} 
\$
and define 
\$
{\sfw}_k = \mathsf{HT}(\{\mathsf{r}_{k,\ell}\}, \gamma_k^*/9) \vee 2.
\$
It is straightforward to see that $\sfw_k \ge \mathsf{pos}(\{\gamma_{k,\ell}\})$;
since $\kappa_{k,\ell} = \sqrt{\ell}\mathsf{r}_{k,\ell}$, 
we have for all $\ell \le \tilde{\sfw}_j$, 
\$
R\Bigg(\frac{\sqrt{\ell}}{\kappa_{k,\ell}} \cdot 
(\gamma^* - \gamma_{k,\ell}) + 1\Bigg)
\le R\Bigg(\frac{\sqrt{\ell}}{\kappa_{k,\ell}} \cdot 
\mathsf{r}_{k,\ell} + 1\Bigg)
\le R(2),
\$
and as a result $\sfR_j = O(1)$.

For $\ell \ge \sfw_k$, $\gamma_{k,\ell} \ge \gamma_k^*/3$;
in addition, since $\rho_{\ell}(\overline{V}_{k,\ell},c_k)/\sqrt{\ell}$ is non-increasing, we have
\$ 
\sfh_k \le \inf
\Bigg\{\ell \ge \sfw_k: 
 (\gamma^*_k)^2 \ge 81\cdot (\rho_\ell(\overline{V}_{k,\ell}, c_k)^2 /\ell)
\vee \mathsf{r}_{k,\ell}^2\Bigg\}
\lesssim \sfw_k \vee 
\mathsf{HT}\big(\{\rho_\ell(\overline{V}_{k,\ell},c_k)/\sqrt{\ell}\},\gamma_k^*/9\big).
\$
Similarly, 
\$ 
\sfl_k \le 
\inf\big\{\ell \ge \tilde \sfw_k: \ell \gamma_k^* \ge 6\log(K/\alpha)\big\}
\lesssim  \sfw_k \vee 
\frac{\log(K/\alpha)}{\gamma_k^*}.
\$
For $\sfH_k$, there is 
\$ 
\sfH_k \le 
2 \sum^\infty_{\ell =  \sfw_k} 
\exp\Bigg\{-\frac{\ell (\gamma_k^*)^2}{1296 
(\rho_\ell(\overline{V}_{k,\ell},c_k)^2 \vee \ell\mathsf{r}_{k,\ell}^2})\Bigg\}
& \le 
2 \sum^\infty_{\ell = \sfw_k} 
\exp\Bigg\{-\frac{(\gamma_k^*)^2}
{1296 \cdot(\rho_\ell(\overline{V}_{k,\ell},c_k)^2/\ell \vee \mathsf{r}_{k,\ell}^2})\Bigg\}.
\$
For $k\in\mH_0$, we take 
$\eta_{k,\ell} \equiv 0$, and therefore $\mathsf{neg}(\{\eta_{k,\ell}\}) = 1$. 
We also let  
\$
\sfw_k = \mathsf{HT}(\{\mathsf{r}_{k,\ell}\},\gamma^*/(10\sqrt{\log(4K)})) \vee 2.
\$ 
Recalling the monotonicity of $\rho_\ell(\overline{V}_{k,\ell},c_k)/\sqrt{\ell}$
yields
\$ 
\mathsf{m}_k & \le \inf\Bigg\{\ell \ge  \sfw_k: 
(\gamma^*)^2 \ge 18 \cdot \log\Big(\frac{4K}{\log 2}\Big) 
\cdot (\rho_\ell(\overline{V}_{k,\ell},c_k)^2/\ell)
\vee \mathsf{r}_{k,\ell}^2\Bigg\} \\
& \lesssim  \sfw_k \vee \mathsf{HT}
\bigg(\Big\{\frac{\rho_\ell(\overline{V}_{k,\ell},c_k)}{\sqrt{\ell}}\Big\},
\frac{\gamma^*}{10\sqrt{\log(4K)}}\bigg)\\
& \le  \sfw_k \vee \mathsf{HT}
\bigg(
\Big\{\frac{\rho_\ell(\overline{V}_{k,\ell},c_k)}{\sqrt{\ell}}\Big\}, 
\frac{d_\KL}{30\sqrt{\log(4K)}}\bigg),
\$
where the last step is because 
\$ 
\gamma^* = \min_{j \in \mH_1} \min_{\ell \ge \tilde \sfw_j}
\gamma_{j,\ell}
\ge \min_{j \in \mH_1} \frac{1}{3}\gamma_j^* = \frac{1}{3}d_\KL.
\$
Taking a union bound over $k \in [K]$
completes the proof.

\subsection{Proof of Corollary~\ref{cor:gaussian}}
\label{appd:proof-gaussian}
Fix $\delta,\alpha  \in (0,1)$ and $K \ge 2$. For $k\in [K]$, 
we write for any $\ell \ge 1$
\$ 
X_\ell = \frac{Y_{t_k(\ell)} - \theta_k}{\sigma_k}
\text{ and }
\overline{X}_{1:\ell} = \frac{1}{\ell}\sum^\ell_{i=1}X_i.
\$
With the above notation, we can write for any $\ell \ge 1$, 
\$ 
\sum^\ell_{i=1} \frac{\|\hat \theta_{k,i} - \theta_k\|_2^2}{\sigma_k^2}
= \sum_{j=1}^d\sum^\ell_{i=1} \frac{\|\hat \theta_{k,i,j}-\theta_{k,j}\|^2_2}{\sigma_k^2},
\$ 
where $\hat \theta_{k,i,j}$ and $\theta_{k,j}$ denotes the $j$-th 
coordinate of $\hat \theta_{k,i}$ and $\theta_k$, respectively.
Letting $\overline{X}_{1:\ell,j}$ denote the $j$-th coordinate 
of $\overline{X}_{1:\ell}$, we have 
\$
\sum_{j=1}^d\sum^\ell_{i=1} \frac{\|\hat \theta_{k,i,j}-\theta_{k,j}\|^2_2}{\sigma_k^2}
= \sum^d_{j=1} \sum^\ell_{i=1} (\overline{X}_{1:i,j})^2
& = \sum^d_{j=1}\sum^\ell_{i=1} 
\frac{\bmX_j^\top{\bm 1}_i {\bm 1}_i^\top\bmX_j}{i^2}\\
& = \sum_{j=1}^d\bmX_j^\top\Big(\sum^\ell_{i=1}\frac{{\bm 1}_i {\bm 1}_i^\top}{i^2}\Big)\bmX_j
:= \sum_{j=1}^d\bmX_j^\top \Sigma \bmX_j,
\$
where $\bmX = (X_1,\dots,X_\ell)^\top$, 
$\bmX_j$ is the $j$-th column of $\bmX$, 
and ${\bm 1}_i$ is a $\ell$-dimensional vector 
with its first $i$ elements being one and the 
remaining elements being zero.
We can further write 
\$ 
\sum_{j=1}^d \bmX_j^\top \Sigma \bmX_j 
= (\bmX_1^\top, \bmX_2^\top,\dots,\bmX_d^\top)
\mathsf{diag}(\Sigma,\Sigma,\dots,\Sigma) 
(\bmX_1^\top, \bmX_2^\top,\dots,\bmX_d^\top)^\top 
= : \tilde\bmX^\top \tilde \Sigma \tilde \bmX.
\$
By definition, $\tilde{X} \sim \mN(0, I_{d\times \ell})$ and 
\$ 
\mathsf{tr}(\tilde\Sigma) = \sum_{j=1}^d\sum^\ell_{i=1} \frac{1}{i} \le 
d\cdot(1 + \log \ell), 
\quad 
\|\tilde \Sigma\|_{\mathsf{op}} \le 4, 
\quad 
\|\tilde \Sigma\|_{\text{F}}^2 \le 4d(1+\log \ell).
\$
By Proposition 1 of~\citet{hsu2012tail}, we have for 
any $x>0$ that 
\$ 
\P\big(\tilde\bmX^\top \tilde \Sigma \tilde \bmX 
> d\cdot(1+\log \ell) + 4 \sqrt{d\log(1+\ell)x} + 8x\big) \le e^{-x}.
\$
For any $m\ge 0$, and $\ell = 2^m$, we 
let $x_m = \log(2K/\delta) + (m+1)\log 2$;
taking a union bound over $m\ge 0$, we have with probability
at least $1 - \delta/(2K)$ that 
\$ 
\sum^{2^m}_{i=1} \frac{\|\hat \theta_{k,i}-\theta_k\|_2^2}{\sigma_k^2}
\le (3d+10\log(K/\delta))\cdot (m+2), \quad \forall m\ge 0.
\$
On the above event, for any $\ell \ge 2$, there exists 
$m_0 \ge 0$ such that $2^{m_0} \le \ell \le 2^{m_0+1}$.
As a result, we have 
\@\label{eq:concentration}
\sum^\ell_{i=1} \frac{\|\hat \theta_{k,i}-\theta_k\|_2^2}{\sigma_k^2}
\le
\sum^{2^{m_0+1}}_{i=1} \frac{\|\hat \theta_{k,i}-\theta_k\|_2^2}{\sigma_k^2}
& \le (3d+10\log(K/\delta))(m_0+3) \notag\\
& \le (3d+10\log(K/\delta))\cdot (2\log \ell + 3).
\@
We work on the event in~\eqref{eq:concentration} hereafter.

Note that 
\$ 
\frac{1}{\ell} \sum^\ell_{i=2}
D_\KL(P_{\theta_k} \| P_{\hat \theta_{k,i-1}})
= \frac{1}{2\sigma_k^2\ell}
\sum^\ell_{i=2} \|\hat \theta_{k,i-1} - \theta_k\|_2^2 
\le \frac{(3d+10\log(K/\delta))\cdot(\log \ell +3) }{2\ell} 
=: \mathsf{r}_{k,\ell}. 
\$
As a result, we have 
\$ 
& \mathsf{HT}\Big(\{\mathsf{r}_{k,\ell}\},\frac{D_\KL(P_k\|P_k^\circ)}{9}\Big)
\lesssim \frac{A_\delta}{\eta_k^2}
\cdot \log\Big(\frac{A_\delta}{\eta_k^2}\Big), \quad \forall k \in\mH_1\\
& \mathsf{HT}\Big(\{\mathsf{r}_{k,\ell}\},\frac{d_{\mathsf{LR}}}{30\sqrt{\log(4K)}}\Big)
\lesssim \frac{A_\delta\cdot \sqrt{\log K}}{d_{\text{LR}}}
\cdot \log\Bigg(\frac{A_\delta\cdot \sqrt{\log K}}{d_{\mathsf{LR}}}\Bigg), \quad \forall k \in\mH_0,
\$
where $A_\delta = d+\log(K/\delta)$
and $\eta_k = \|\theta_k - \theta_k^\circ\|_2/\sigma_k$.

Next, it is straightforward to check 
that for any $k\in[K]$ and $\ell \ge 2$,  
\$ 
V_{k,\ell} = \sum^\ell_{i=2}
\frac{\|\hat \theta_{k,i-1} - \theta_k^\circ\|_2^2}{\sigma_k^2}
&\le \sum^\ell_{i=2} \frac{2\|\hat \theta_{k,i-1} - \theta_k\|_2^2 
+ 2\|\theta_k - \theta_k^\circ\|_2^2}{\sigma_k^2}\\
&\lesssim A_\delta\log \ell
+ \eta_k^2\ell =: \overline{V}_{k,\ell}.
\$
We can then check that 
\$ 
\rho_\ell(\overline{V}_{k,\ell},1)/\sqrt{\ell} 
\lesssim \frac{\log(Kd\ell/\delta)}{\ell}  
\sqrt{A_\delta}
\vee 
\eta_k \sqrt{\frac{\log(\eta_k K\ell/\delta)}{\ell}}
\$
With the above, we can verify that 
\$
& \mathsf{HT}\Bigg(\bigg\{\frac{\rho_{\ell}(\overline{V}_{k,\ell},1)}{\sqrt{\ell}}\bigg\}, \frac{D_\KL(P_k\|P_k^\circ)}{9}\Bigg)
\lesssim 
\frac{\sqrt{A_\delta}}{\eta_k^2}
\cdot\log\Bigg(\frac{Kd }
{\delta \eta_k}\Bigg),\quad \forall k \in \mH_1\\
& \mathsf{HT}\Bigg(\bigg\{\frac{\rho_{\ell}(\overline{V}_{k,\ell},1)}{\sqrt{\ell}}\bigg\}, \frac{d_{\mathsf{LR}}}{30\sqrt{\log(4K)}}\Bigg)
\lesssim 
\frac{\sqrt{A_\delta\log K}}{d_{\mathsf{LR}}}
\cdot\log\Bigg(\frac{Kd }
{\delta d_{\mathsf{LR}}}\Bigg),\quad \forall k \in \mH_0.
\$
Finally, for any $j\in \mH_1$, we have 
\$ 
\sfH_k \lesssim \frac{A_\delta}{\eta_k^2}
\cdot\log\Big(\frac{Kd}{\delta \eta_k^2}\Big),
\$
where the inequality uses Lemma~\ref{lem:integral}.
Putting everything together and letting $\delta = \delta'/2$, 
we have with probability at least $\P(\TPP(\mR_T))\ge 1-\delta'$, 
if 
\$ 
T \gtrsim K\log\Big(\frac{K}{\delta'}\Big) & + 
\sum_{j \in \mH_1} \Bigg\{\frac{A_{\delta'}}{\eta_k^2}
\cdot \log\Big(\frac{Kd}{\delta'\eta_k^2}\Big)
+ \frac{\log(K/\alpha)}{\eta_k^2}\Bigg\}  + |\mH_0|\cdot \frac{A_{\delta'}\cdot \sqrt{\log K}}{d_{\mathsf{LR}}}
\cdot \log\Big(\frac{Kd}{\delta' d_{\mathsf{LR}}}\Big).
\$

\section{Upper bound on the expected stopping time}\label{sec:upper_bound_expected_time}
This section collects parallel results to the high-probability bounds 
in the main text, but for the expected stopping time $\E[\tau_*]$ of 
the e-PS algorithm. 
For notational simplicity, we let
$v_{k,t} = \nu_{k,n_{k,t}}$, for any $t \ge 1$.

Given any sequence $\{x_\ell\}$, we define the following key quantities.
\begin{equation}\label{eq:general-power-terms}
\begin{aligned} 
& \widetilde{\mathsf{L}}_k(\{x_\ell\}) = 
\sum^\infty_{\ell = 1} \P\Big(\frac{S_{k,\ell}}{\ell} < x_\ell\Big), 
\quad \widetilde{\mathsf{U}}_k(\{x_\ell\}) 
= \sum^\infty_{\ell = 1} \P\Big(\frac{S_{k,\ell}}{\ell} > x_\ell\Big),\\ 
& \widetilde{\mathsf{V}}_k(\{x_{k,\ell}\}) 
= \sum^\infty_{\ell = 1} \P(V_{k,\ell} > x_{k,\ell}),
\quad 
\widetilde{\mathsf{S}}_k(\{x_\ell\}) = \sum^\infty_{\ell = 1}
\P(\nu_{k,\ell} > x_{\ell})
\end{aligned}
\end{equation}

The following theorem provides an upper bound on $\E[\tau_*]$ for the e-PS algorithm 
in terms of the above quantities.
\begin{theorem}
\label{thm:general-power}
Fix $\alpha,\delta \in (0,1)$.
Consider $K$ hypotheses $\{H_k^{(0)}, H_k^{(1)}\}_{k\in [K]}$ and assume $\mH_1 \neq \varnothing$.
For each $k\in[K]$, let $\{E_{k,t}\}_{t\ge 1}$ be a valid e-process 
for $H_k^{(0)}$, adapted to $\{\mF_t\}$, in 
the form of~\eqref{eq:product} with increment e-values $\{e_{k,t}\}$ satisfying Assumption~\ref{assump:likelihood-e}.

For any $k\in[K]$, consider a deterministic, positive and non-decreasing 
sequence $\{\overline{V}_{k,\ell}\}_{\ell \ge 1}$ for variance upper bound, 
a predictable sequence $\{\nu_{k,\ell}\}_{\ell \ge 1}$ for the sampling variance,
and a deterministic sequence $\{\overline{\nu}_{k,\ell}\}_{\ell \ge 1}$ 
for the upper bound on $\{\nu_{k,\ell}\}$. Define for any sequence $\{x_\ell\}$
and $w \in \mathbb{N}_+$
\$ 
\widetilde{\sfH}_k(\{x_\ell\},w) = 
\sum^\infty_{\ell = w} \exp\Big(-\frac{\ell^2 x_{\ell}^2}{2 \overline{V}_{k,\ell}}\Big).
\$
For each $k\in \mH_1$, consider any non-decreasing sequence $\{\gamma_{k,\ell}\}$
with $\mathsf{pos}(\{\gamma_{k,\ell}\}) < \infty$. 
Fixing any burn-in period $\sfw_k \ge \mathsf{pos}(\{\gamma_{k,\ell}\})$, 
define $\gamma^* = \min_{k\in \mH_1, \ell \ge \sfw_k} \gamma_{k,\ell}$
and
\$ 
& \tilde \sfl_k = \inf\bigg\{\ell \ge \sfw_k: n \gamma_{k,n} \ge 2 \log\Big(\frac{K}{\alpha}\Big), 
\forall  n \ge \ell \bigg\},\quad 
\widetilde{\sfR}_k = \sum^\infty_{\ell = 1} 
\E\bigg[R\bigg(\sqrt{\frac{\ell}{\nu_{k,\ell }}}\Big(\frac{2\gamma^*}{3}
- \frac{M_{k,\ell}}{\ell} \Big) \bigg)\bigg]. 
\$
For each $k\in \mH_0$, consider any non-increasing sequence $\{\eta_{k,\ell}\}_{\ell \ge 1}$
with $\mathsf{neg}(\{\eta_{k,\ell}\}) < \infty$. Fixing any burn-in period 
$\sfw_k \ge \mathsf{neg}(\{\eta_{k,\ell}\})$, define 
\$ 
\widetilde{\mathsf{m}}_k = \inf\Big\{\ell \ge \sfw_k: n (\gamma^* - \eta_{k,\ell})^2 \ge 18 \overline{\nu}_{k,n} \log(4K), \forall n \ge \ell\Big\}.
\$
The e-PS algorithm with variance parameters 
$v_{k,t} = \nu_{k,n_{k,t}}$ satisfies
\$ 
\E[\tau^*] & \lesssim 
 K + \sum_{k\in \mH_1}\big(\tilde{\sfl}_k + \widetilde{\sfL}_k(\{\gamma_{k,\ell}\}) + 
\widetilde{\mathsf{H}}_k(\{\gamma_{k,\ell}/2\}, \sfw_k) + \widetilde{\sfR}_k\big)\\
& \qquad \qquad + \sum_{k\in\mH_0} \big(\widetilde{\mathsf{m}}_k + \widetilde{\sfU}_k(\{\eta_{k,\ell}\} + \widetilde{\mathsf{S}}_k(\{\overline{\nu}_{k,\ell}\})) + 
\widetilde{\mathsf{H}}_k(\{(\gamma^* - \eta_{k,\ell})/3\}, \sfw_k)  \big)
+ \sum_{k\in [K]} \widetilde{\mathsf{V}}_k(\{\overline{V}_{k,\ell}\})
\$
\end{theorem}

\begin{proof}
For any $k\in \mH_1$, fix a non-decreasing sequence $\{\gamma_{k,\ell}\}_{\ell \ge 1}$ with $\mathsf{pos}(\{\gamma_{k,\ell}\}) < +\infty$,
and for any $k\in \mH_0$, fix a non-increasing sequence $\{\eta_{k,\ell}\}_{\ell \ge 1}$ with $\mathsf{neg}(\{\eta_{k,\ell}\}) < +\infty$.
We decompose the expected stopping time by the 
contribution of the non-nulls and the nulls:
\$ 
\E[\tau_*] = \sum_{k\in \mH_1}\E[n_{k,\tau_*}] + \sum_{k\in \mH_0}\E[n_{k,\tau_*}], 
\$
and establish an upper bound for each term separately.

\paragraph{The contribution of non-nulls.}
Fix $k \in \mH_1$. Define 
\$ 
T_k = \inf\{\ell \ge 1:  M_{k,\ell} \ge \log(K/\alpha)\}.
\$
By the sampling strategy, we have $n_{k,\tau_*} \le T_k$, and therefore $\E[n_{k,\tau_*}] \le \E[T_k]$.
Next, notice that 
\$ 
\{T_k > \ell \} \subseteq \{M_{k,\ell} < \log(K/\alpha)\},
\$ 
and therefore
\$ 
\E[T_k] = \sum^\infty_{\ell=0} \P(T_k > \ell) 
& \le 1 + \sum^\infty_{\ell=1} \P(M_{k,\ell} < \log (K/\alpha)).
\$
For $\ell \ge \tilde{\sfl}_k$, 
\$ 
& \P(M_{k,\ell} < \log(K/\alpha)) \\
= ~ & \P\Big(M_{k,\ell} < \log ({K}/{\alpha}), 
S_{k,\ell} <\ell \gamma_{k,\ell} \text{ or } V_{k,\ell} > \overline{V}_{k,\ell}\Big) 
+ \P\Big(M_{k,\ell} < \log ({K}/{\alpha}), S_{k,\ell} \ge \ell \gamma_{k,\ell}, 
V_{k,\ell} \le \overline{V}_{k,\ell}\Big)\\
\le ~& \P(S_{k,\ell} < \ell \gamma_{k,\ell}) 
+ \P(V_{k,\ell} > \overline{V}_{k,\ell})   
+ \P\Big(M_{k,\ell} - S_{k,\ell} < \log ({K}/{\alpha}) -  \ell\gamma_{k,\ell}, 
V_{k,\ell} \le \overline{V}_{k,\ell}\Big)\\
\le~& 
\P(S_{k,\ell} < \ell \gamma_{k,\ell}) + \P(V_{k,\ell} > \overline{V}_{k,\ell}) +
\exp\bigg(-\frac{\ell^2 \gamma_{k,\ell}^2}{8 \overline{V}_{k,\ell}}\bigg),
\$
where the last step follows from 
the sub-Gaussian property of $Z_{k,i}$ and the choice of $\tilde{\sfl}_k$.
Summing over $\ell \ge 1$, we have 
\$ 
\E[T_k] & \le 1 + \sfl_k + 
\sum^\infty_{\ell=\sfl_k} 
\bigg\{\P(S_{k,\ell} \le \ell \gamma_{k,\ell}) 
+ \P(V_{k,\ell} > \overline{V}_{k,\ell}) + 
\exp\Big(-\frac{\ell^2 \gamma_{k,\ell}^2}{8 \overline{V}_{k,\ell}}\Big)\bigg\}\\
& \lesssim 1 + \sfl_k + \widetilde{\mathsf{L}}_k(\{\gamma_{k,\ell}\}) +
\widetilde{\mathsf{V}}_k(\{\overline{V}_{k,\ell}\}) + \widetilde{\sfH}_k(\{\gamma_{k,\ell}/2\},\sfw_k),
\$
where the last step follows from the definition of 
$\widetilde{\mathsf{L}}_k$, 
$\widetilde{\mathsf{V}}_k$, and $\widetilde{\sfH}_k$.
Summing over $k\in \mH_1$, we have
\$ 
\sum_{k\in \mH_1}\E[n_{k,\tau_*}] 
\lesssim \sum_{k\in \mH_1} \Big\{1 + \tilde{\sfl}_k + 
\widetilde{\mathsf{L}}_k(\{\gamma_{k,\ell}\}) + 
\widetilde{\mathsf{V}}_k(\{\overline{V}_{k,\ell}\}) + 
\widetilde{\sfH}_k(\{\gamma_{k,\ell}/2\},\sfw_k)
\Big\}.
\$
\paragraph{The contribution of nulls.}
We now turn to the contribution of nulls.  Define 
\$
\gamma_* = \min_{j\in \mH_1}\min_{\ell \ge \sfw_j} \gamma_{j,\ell}.
\$
For any $k\in \mH_0$, we let  
$\Delta_{k,\ell} = \gamma_* - \eta_{k,\ell}$ and define the following two thresholds 
\$ 
\sfa_{k,\ell} = \eta_{k,\ell} + \frac{\Delta_{k,\ell}}{3},\quad 
\sfb = \frac{2 \gamma_*}{3},\quad \forall \ell \ge 1.
\$
We can check that $\eta_{k,\ell} < \sfa_{k,\ell} < \sfb < \gamma_*$, for 
every $\ell \ge \sfw_k$.

Next, we decompose the  number of  
pulls of $k \in \mH_0$ as
\$ 
\sum_{t \le \tau_*} \ind\{A_t = k\}
& = 1 + \sum_{K< t \le \tau_*} \ind\{A_t = k, \hat m_{k,t-1} > \sfa_{k,n_{k,t-1}}\}
+ \sum_{K< t \le \tau_*} \ind\{A_t = k, \hat m_{k,t-1} \le \sfa_{k,n_{k,t-1}}, 
\tilde m_{k,t} > \sfb\}\\
& \qquad \qquad + \sum_{K< t \le \tau_*} \ind\{A_t = k, \hat m_{k,t-1} \le \sfa_{k,n_{k,t-1}}, \tilde m_{k,t} \le \sfb\}\\
& =:1 + T_{1,k} + T_{2,k} + T_{3,k}.
\$
We establish an upper bound for each term separately.

For $T_{1,k}$, we have for any $t > K$,
\$ 
\sum_{K<t\le \tau^*}\ind\{A_t = k, \hat m_{k,t-1} > \sfa_{k,n_{k,t-1}}\} 
& =\sum_{K<t\le \tau^*}\sum^\infty_{\ell=1}
\ind\{A_t = k, M_{k,\ell}/\ell > \sfa_{k,\ell}, n_{k,t-1} = \ell\} \\
& \le \sum^\infty_{\ell = 1} 
\ind\{M_{k,\ell}/\ell > \sfa_{k,\ell}\}.
\$
And therefore, 
\$ 
\E[T_{1,k}] & \le \sfw_k + \sum^\infty_{\ell=\sfw_k} \P(M_{k,\ell} > \ell \sfa_{k,\ell})\\
& \le \sfw_k + \sum^\infty_{\ell=\sfw_k} \P(M_{k,\ell} > \ell \sfa_{k,\ell}, 
S_{k,\ell} \le \ell \eta_{k,\ell}, V_{k,\ell} \le \overline{V}_{k,\ell}) 
+ \sum^\infty_{\ell=\sfw_k} 
\big\{\P(S_{k,\ell} > \ell \eta_{k,\ell}) + \P(V_{k,\ell} > \overline{V}_{k,\ell})\big\}\\
& \le  \sfw_k +  
\sum^\infty_{\ell=\sfw_k} \P(M_{k,\ell} -S_{k,\ell} > \ell \Delta_{k,\ell}/3, 
S_{k,\ell} \le \ell \eta_{k,\ell}, V_{k,\ell} \le \overline{V}_{k,\ell}) 
+ \widetilde{\mathsf{U}}_k(\{\eta_{k,\ell}\}) + \widetilde{\mathsf{V}}_k(\{\overline{V}_{k,\ell}\})\\
& \stackrel{\mathrm{(a)}}{\le} \sfw_k +  
\sum^\infty_{\ell=\sfw_k} \exp\bigg(-\frac{(\ell \Delta_{k,\ell}/3)^2}{2 \overline{V}_{k,\ell}}\bigg) 
+  \widetilde{\mathsf{U}}_k(\{\eta_{k,\ell}\}) 
+ \widetilde{\mathsf{V}}_k(\{\overline{V}_{k,\ell}\})\\
& =\sfw_k + \widetilde{\mathsf{H}}_k(\{\Delta_{k,\ell}/3\},\sfw_k) +  \widetilde{\mathsf{U}}_k(\{\eta_{k,\ell}\}) 
+ \widetilde{\mathsf{V}}_k(\{\overline{V}_{k,\ell}\}), 
\$
where 
step (a) follows from the sub-Gaussian property of $Z_{k,i}$.

For $T_{2,k}$, we consider 
\$ 
\E[T_2] & := \sum_{k\in \mH_0} \E[T_{2,k}]\\
& = \sum_{k\in \mH_0} \sum_{t>K} \P(t \le \tau^*, A_t = k, 
\hat m_{k,t-1} \le \sfa_{k,n_{k,t-1}},\tilde m_{k,t} \ge \sfb)\\
& \le \sum_{k\in \mH_0} \sum_{t>K} \P(t \le \tau^*, A_t = k, 
\hat m_{k,t-1} \le \sfa_{k,n_{k,t-1}},\tilde m_{k,t} \ge \sfb, 
V_{k,n_{k,t-1}} \le \overline{V}_{k,n_{k,t-1}}, \nu_{k,n_{k,t-1}} \le \overline{\nu}_{k,n_{k,t-1}})\\ 
& \hspace{5em}  
+ \sum_{k\in\mH_0} \sum_{t > K} \P(A_t = k, V_{k,n_{k,t-1}} > \overline{V}_{k,n_{k,t-1}})
+  \P(A_t = k, \nu_{k,n_{k,t-1}} > \overline{\nu}_{k,n_{k,t-1}})
\\
& \le \sum_{k\in \mH_0}\Bigg\{ \sum_{t>K} \P(t \le \tau^*, A_t = k, n_{k,t-1} > \mathsf{m}_k, 
\hat m_{k,t-1} \le \sfa_{k,n_{k,t-1}},\tilde m_{k,t} \ge \sfb, 
V_{k,n_{k,t-1}} \le \overline{V}_{k,n_{k,t-1}}, \nu_{k,n_{k,t-1}} \le \overline{\nu}_{k,n_{k,t-1}}) \\
& \hspace{6em} + \mathsf{m}_k +  \widetilde{\mathsf{V}}_k(\{\overline{V}_{k,\ell}\})
+ \widetilde{\mathsf{S}}_k(\{\overline{\nu}_{k,\ell}\})\Bigg\}.
\$ 
For any $K\in\mH_0$ and $t > K$, 
on the event $\{t \le \tau^*, n_{k,t-1} > \mathsf{m}_k, 
\hat m_{k,t-1} \le \sfa_{k,n_{k,t-1}}, V_{k,n_{k,t-1}}\le \overline{V}_{k,n_{k,t-1}},
\nu_{k,n_{k,t-1}}\le \overline{\nu}_{k,n_{k,t-1}}\}$, there is 
\$ 
\P(A_t = k, \tilde m_{k,t} \ge \sfb\given \mF_{t-1})
\le ~& 
\P\bigg(\sqrt{\frac{v_{k,t-1}}{n_{k,t-1}}}\cdot \xi_{k,t} > \sfb - \sfa_{k,n_{k,t-1}}\Biggiven \mF_{t-1}\bigg)\\ 
\stackrel{\mathrm{(a)}}{\le}~& 
\P\bigg(\sqrt{\frac{v_{k,t-1}}{n_{k,t-1}}}\cdot \xi_{k,t} > \Delta_{k,n_{k,t-1}}/3
\Biggiven \mF_{t-1}\bigg) \\
\le ~& 
\exp\bigg(-\frac{\Delta^2_{k,n_{k,t-1}}n_{k,t-1}}{18 v_{k,t-1}}\bigg)
\le \exp\bigg(-\frac{\Delta^2_{k,n_{k,t-1}}n_{k,t-1}}{18 \overline{\nu}_{k,n_{k,t-1}}}\bigg)
\stackrel{\mathrm{(b)}}{\le}  \frac{1}{4K} \ind\{t \le \tau^*\}. 
\$
Above, step (a) is due to 
\$
\sfb - \sfa_{k,n_{k,t-1}} = \frac{2\gamma^*}{3} - \frac{\Delta_{k,n_{k,t-1}}}{3}
- \eta_{k,n_{k,t-1}} = \frac{\Delta_{k,n_{k,t-1}} - \eta_{k,n_{k,t-1}}}{3}
\ge \frac{\Delta_{k,n_{k,t-1}}}{3}, \text{ for } n_{k,t-1} \ge \sfw_k,
\$
and step (b) follows from the choice of $\mathsf{m}_k$.
Putting everything together, we arrive at
\$ 
\E[T_2] \le \frac{\E[\tau^*]}{4} + \sum_{k\in \mH_0}\big(\mathsf{m}_k + 
\widetilde{\mathsf{V}}_k(\{\overline{V}_{k,\ell}\})  +\widetilde{\mathsf{S}}_k(\{\overline{\nu}_{k,\ell}\})\big)
\$
We now turn to $T_{3,k}$, 
where we directly target the 
total contribution of nulls 
\$ 
T_3 := \sum_{k\in \mH_0} T_{3,k}
& =  \sum_{k\in \mH_0} \sum_{t=K+1}^\infty \ind\{t \le \tau_*, 
A_t = k, \hat m_{k,t-1} \le \sfa_{k,n_{k,t-1}}, \tilde m_{k,t} \le \sfb\}\\
& = \sum_{t=K+1}^\infty \ind\{t \le \tau_*, A_t \in \mH_0, 
\hat m_{A_t,t-1} \le \sfa_{A_t,n_{A_t,t-1}},
\tilde m_{A_t,t} \le \sfb\}
\$
On the event $\{t \le \tau_*\}$, we have $\mH_1 \not\subseteq \mR_{t-1}$, and therefore 
there exists some $j\in \mH_1$ such that $j \in \mathcal{C}_{t-1} :=  [K]\backslash \mR_{t-1}$. 
Define $J_t = \min\{j\in \mH_1: j \in \mathcal{C}_{t-1}\}$, i.e., the non-null hypothesis with the smallest index 
that has not been rejected yet. 
By the tower property of conditional expectation and that $\{t \le \tau_*\}$ is $\mF_{t-1}$-measurable, we
have for any $t > K$, 
\$ 
& \P\big(t \le \tau_*, A_t \in \mH_0, \hat m_{A_t,t-1} \le 
\sfa_{A_t,n_{A_t,t-1}},\tilde m_{A_t,t} \le \sfb\big)\\
=~&\E\Big[\ind\big\{t \le \tau_*\big\}
\cdot \P(A_t \in \mH_0, \hat m_{A_t,t-1} \le \sfa_{A_t,n_{A_t,t-1}}, 
\tilde m_{A_t,t} \le \sfb \mid \mF_{t-1})\Big]\\
=~&\E\Bigg[\ind\big\{t \le \tau_*\big\}\sum_{j\in \cH_1}\ind\{J_t = j\}
\cdot \P(A_t \in \mH_0, \hat m_{A_t,t-1} \le \sfa_{A_t,n_{A_t,t-1}}, 
\tilde m_{A_t,t} \le \sfb \mid \mF_{t-1})\Bigg],
\$
where the last step is because $\sum_{j\in \mH_1} \ind\{J_t = j\} = 1$ on the event $\{t \le \tau_*\}$.
On the event $\{t \le \tau_*, J_t = j\}$, 
there is 
\$ 
\P(A_t \in \cH_0, \hat m_{A_t,t-1} \le \sfa_{A_t,n_{A_t,t-1}}, 
\tilde m_{A_t,t} \le \sfb  \mid \mF_{t-1})
& \le \P\Big(\max_{\ell \in \cC_{t-1}} \tilde m_{\ell,t} \le \sfb \mid \mF_{t-1}\Big) \\
& =  \P(\tilde m_{j,t} \le \sfb \mid \mF_{t-1}) \prod_{\ell \in \cC_{t-1}\backslash\{j\}} 
\P(\tilde m_{\ell,t} \le \sfb \mid \mF_{t-1}).
\$
On the same event, we also have 
\$ 
\P( A_t = j \mid \mF_{t-1})
& \ge 
\P(\tilde m_{j,t} > \sfb , \max_{\ell \in \cC_{t-1}\backslash\{j\}} \tilde m_{\ell,t} \le \sfb \mid \mF_{t-1})\notag\\
& =\P(\tilde m_{j,t} > \sfb \mid \mF_{t-1}) 
\prod_{\ell \in \cC_{t-1}\backslash\{j\}} \P(\tilde m_{\ell,t} \le \sfb \mid \mF_{t-1}).
\$
Combining the above, we have
\$ 
\P(A_t \in \cH_0, \hat m_{A_t,t-1} \le \sfa_{A_t,n_{A_t,t-1}}, 
\tilde m_{A_t,t} \le \sfb \mid \mF_{t-1})
\le \frac{\P(\tilde m_{j,t} \le \sfb \mid \mF_{t-1})}{\P(\tilde m_{j,t} > \sfb \mid \mF_{t-1})} \P(A_t = j \given \mF_{t-1}). 
\$
As a result, 
\$ 
\E[T_3] & \le \sum_{j\in \cH_1}\sum_{t=K+1}^\infty \E\Bigg[\ind\{t \le \tau_*, J_t = j\} 
\frac{\P(\tilde m_{j,t} \le \sfb \mid \mF_{t-1})}{\P(\tilde m_{j,t} > \sfb \mid \mF_{t-1})} \P(A_t = j \given \mF_{t-1}) \Bigg] \\
& \le \sum_{j\in \cH_1}\sum^\infty_{t=K+1}\E\Bigg[\ind\{A_t = j\} \cdot 
\frac{\P(\tilde m_{j,t} \le \sfb \mid \mF_{t-1})}{\P(\tilde m_{j,t} > \sfb \mid \mF_{t-1})}\Bigg] \\
& = \sum_{j \in \cH_1} \sum_{\ell = 1}^\infty\E\Bigg[R\bigg(\sqrt{\frac{\ell}{\nu_{j,\ell}}}
\Big(\sfb - \frac{M_{j,\ell}}{\ell}\Big)\bigg)\Bigg].
\$
Combining everything, we have 
\$ 
\E[\tau^*] & \lesssim 
 K + \sum_{k\in \mH_1}\big(\tilde{\sfl}_k + \widetilde{\sfL}_k(\{\gamma_{k,\ell}\}) + 
\widetilde{\mathsf{H}}_k(\{\gamma_{k,\ell}/2\}, \sfw_k) + \sfR_k \big)\\
& \qquad \qquad + \sum_{k\in\mH_0} \big(\widetilde{\mathsf{m}}_k + \widetilde{\sfU}_k(\{\eta_{k,\ell}\}) + 
\widetilde{\mathsf{H}}_k(\{\Delta_{k,\ell}/3\}, \sfw_k) +  \widetilde{\mathsf{S}}_k(\{\overline{\nu}_{k,\ell}\}) 
\big)
+ \sum_{k\in [K]} \widetilde{\mathsf{V}}_k(\{V_{k,\ell}\})
\$
\end{proof}

As previously, we can then specify the general bound to specific testing problems 
given the e-process construction. For example, for the simple-versus-simple testing problem and the 
likelihood ratio e-process given in~\eqref{eq:simple-e-val}, we have the following specialization.
\begin{theorem}\label{thm:ss-power} 
Consider a multiple testing problem with singleton null and alternative distribution classes 
\$
\mP_k = \{P_k^\circ\}, \quad  \mQ_k = \{Q_k^\circ\}, \quad \forall k\in[K],
\$
where $Q_k^\circ$ is absolutely continuous with respect to $P_k^\circ$, and let the target FDR level be $\alpha \in (0,1)$. 
Suppose Assumption~\ref{assump:likelihood-e} holds with $s^2_{k,\ell} \equiv \sigma_k^2$.
The e-PS algorithm implemented 
with the likelihood-ratio e-values defined in~\eqref{eq:simple-e-val} 
and $v_{k,t} = (1+\epsilon)\,\sigma^2_k$ for some $\epsilon > 0$ satisfies
\begin{align*}
    \E[\tau_*] \lesssim K +  \sum_{k \in \mathcal{H}_1}
    \left(\frac{\log(K/\alpha)}{D_{\KL}(Q_k^\circ\|P_k^\circ)} 
    + \frac{\sigma_k^2}{D_{\KL}(Q_k^\circ\|P_k^\circ)^2}\right) + 
    \sum_{k \in \mathcal{H}_0} \frac{\sigma_k^2 \log K}{(D_\KL(P_k^\circ\|Q_{k}^\circ) + d_{\mathsf{LR}})^2},  
\end{align*}
where
$d_{\mathsf{LR}} = \min_{j \in \mH_1} D_\KL(Q_j^\circ\|P_{j}^\circ)$.
\end{theorem}

\begin{proof}
For any $k\in [K]$, we let $\overline{V}_{k,\ell} = \ell \sigma_k^2$ and $\overline{\nu}_{k,\ell} = 
\nu_{k,\ell}$. Then $\widetilde{\mathsf{S}}_k(\{\overline{\nu}_{k,\ell}\}) = \widetilde{\mathsf{V}}_{k}(\{\overline{V}_{k,\ell}\}) = 0$.

For each $k\in \mH_1$, we take $\gamma_{k,\ell} \equiv D_{\KL}(Q_k^\circ \| P_k^\circ)$ and $\sfw_k = 1$.
By construction, we have $S_{k,\ell}/\ell \equiv \gamma_{k,\ell}$ and
therefore $\widetilde{\mathsf{L}}_k(\{\gamma_{k,\ell}\}) = 0$.
In addition, 
\$
\gamma^* = \min_{k\in \mH_1} D_{\KL}(Q_k^\circ \| P_k^\circ) = d_{\mathsf{LR}}.
\$
Next, there is
\$ 
\tilde {\sfl}_k \lesssim \frac{\log(K/\alpha)}{D_{\KL}(Q_k^\circ \| P_k^\circ)}
\$
and 
\$ 
\tilde{\sfH}_k(\{\gamma_{k,\ell}/2\},\sfw_k)
= \sum^\infty_{\ell = 1} \exp\bigg(-\frac{D_{\KL}(Q_k^\circ \| P_k^\circ)\ell}{8 \sigma_k^2}\bigg)
\lesssim \frac{\sigma_k^2}{D_{\KL}(Q_k^\circ \| P_k^\circ)^2}. 
\$
For $\sfR_k$, we have
\@\label{eq:exp-svs-t2}
& \sum_{\ell = 1}^\infty\E\Bigg[R\bigg(\sqrt{\frac{\ell}{(1+\epsilon) \sigma_k^2}}
\Big(\sfb - \frac{M_{k,\ell}}{\ell}\Big)\bigg)\Bigg] \notag \\
= ~& \sum_{\ell = 1}^\infty\E\Bigg[R\bigg(\sqrt{\frac{\ell}{(1+\epsilon)\sigma_k^2}}
\Big(\sfb - \frac{M_{k,\ell}}{\ell}\Big)\bigg) \ind\Big\{\frac{S_{k,\ell}}{\ell} 
<  \gamma_{k,\ell}\Big\}\Bigg]
+ \sum_{\ell = 1}^\infty\E\Bigg[R\bigg(\sqrt{\frac{\ell}{(1+\epsilon)\sigma_k^2}}\Big(\sfb - \frac{M_{j,\ell}}{\ell}\Big)\bigg)\ind\Big\{\frac{S_{j,\ell}}{\ell} \ge \gamma_{j,\ell}\Big\}\Bigg] \notag\\
\le ~& \sum_{\ell = 1}^\infty 
\E\Bigg[R\bigg(\sqrt{\frac{\ell}{(1+\epsilon)\sigma_k^2}}\cdot \Big\{
-\frac{\gamma_{k,\ell}}{3} + \frac{S_{k,\ell}-M_{k,\ell}}{\ell}\Big\}\bigg)\Bigg],
\@
where we again use $S_{k,\ell}/\ell \equiv \gamma_{k,\ell} $
For $\ell \ge 1$, we apply Lemma~\ref{lem:exp-odds} with 
\$ 
W = \sqrt{\frac{1}{(1+\epsilon)\sigma_k^2\ell}}(S_{k,\ell} - M_{k,\ell}) \text{ and }
a = \frac{\gamma_{k,\ell}\sqrt{\ell}}{3\sqrt{1+\epsilon}\sigma_k},
\$ 
where $W$ is centered and sub-Gaussian with variance proxy $\frac{1}{1+\epsilon}<1$.
Then 
\$ 
\sfR_j \le 
\sum_{\ell = 1}^\infty C(\epsilon)\exp\bigg(-\frac{c(\epsilon)}{9(1+\epsilon)}
\cdot \frac{\gamma_{j,\ell}^2\ell}{\sigma_k^2} \Big)
\lesssim \frac{\sigma_k^2}{D_{\KL}(Q_k^\circ \| P_k^\circ)^2}. 
\$
For each $k\in \mH_0$, we take $\eta_{k,\ell} \equiv - D_{\KL}(P_k^\circ \| Q_k^\circ)$
and $\sfw_k = 1$. As a result, $S_{k,\ell} / \ell \equiv \eta_{k,\ell}$
and $\widetilde{\sfU}_k(\{\eta_{k,\ell}\})=0$.
We can check that 
\$ 
\widetilde{\mathsf{m}}_k \lesssim \frac{\sigma_k^2 \log K}{(d_{\mathsf{LR}} 
+ D_{\KL}(P_k^\circ \| Q_k^\circ))^2},  
\$
and 
\$ 
\widetilde{\sfH}_k(\{\Delta_{k,\ell}/3\}, \sfw_k)
= \sum_{\ell = 1}^\infty \exp\bigg(-\frac{\ell \Delta_{k,\ell}^2}{18\sigma_k^2}\bigg)
\lesssim \frac{\sigma_k^2}{(d_{\mathsf{LR}} + D_{\KL}(P_k^\circ \| Q_k^\circ))^2}.
\$
Putting everything together, we have
\$ 
\E[\tau_*] \lesssim K + \sum_{k\in \mH_1} \bigg(\frac{\log(K/\alpha)}{D_{\KL}(Q_k^\circ \| P_k^\circ)} 
+ \frac{\sigma_k^2}{D_{\KL}(Q_k^\circ \| P_k^\circ )^2}\bigg) + 
\sum_{k\in \mH_0} \frac{\sigma_k^2 \log K}{(d_{\mathsf{LR}} +
D_{\KL}(P_k^\circ \| Q_k^\circ))^2}. 
\$
\end{proof}

\section{Technical lemmas}
\begin{lemma}\label{lem:integral}
For any $a,b>0$, and $c \ge e^{2-b} \vee 2/a$, there is 
\$ 
\int_c^\infty \exp\Big(-\frac{ax}{b+ \log x}\Big) \, dx \le \frac{4(b+\log c)}{a} \exp\Big(-\frac{ac}{b+\log c}\Big).
\$
\end{lemma}
\begin{proof}
Let $\phi(x) =\frac{ax}{b+\log x}$ and $f(x) = e^{-\phi(x)}$.
Direct calculation leads to 
\$ 
\phi'(x) = \frac{a(b+\log x -1)}{(b+\log x)^2}.
\$  
We additionally define $F(x) = \frac{-2}{\phi'(x)}e^{-\phi(x)}$ and 
\$ 
F'(x) = 2e^{-\phi(x)}\bigg(1 - \Big(\frac{1}{\phi(x)}\Big)'\bigg)
= 2e^{-\phi(x)}\bigg(1 - \frac{(b+\log x)(b+\log x-2)}{ax(b+\log x -1)^2}\bigg)
\ge e^{-\phi(x)} = f(x),
\$
where the last inequality holds when $c \ge e^{2-b} \vee 2/a$.
Then we have
\$
\int_c^\infty f(x) dx \le \int_c^\infty F'(x) dx = F(\infty) - F(c) 
= \frac{2(b+\log c)^2}{a(b+\log c -1)} e^{-\phi(c)} 
\le \frac{4(b+\log c)}{a} \exp\Big(-\frac{ac}{b+\log c}\Big).
\$

\end{proof}

\begin{lemma}
\label{lem:exp-odds}
Let $W$ be a centered $\rho^2$-sub-Gaussian random variable, 
for some $\rho^2<1$.
Then there exist constants $C(\rho),c(\rho)>0$ depending only on $\rho$, such that for all $a\ge 0$,
\[
\E\big[R(W-a)\big]\le C(\rho) \cdot e^{-c(\rho) a^2},
\]
where $R(z) = \Phi(z)/(1-\Phi(z))$ and $\Phi(z)$ is the cumulative distribution function (CDF) of standard normal distribution.
\end{lemma}

\begin{proof}
Fix $W$ to be a centered $\rho^2$-sub-Gaussian random variable and $a\ge 0$.
Let $Z:=W-a$.
We can decompose $\E[R(Z)]$ according to the sign of $Z$:
\[
\E[R(Z)]
=
\E[R(Z)\cdot \ind\{Z\le 0\}]
+
\E[R(Z)\cdot\ind\{Z>0\}].
\]
We will bound the two terms separately.
\paragraph{Bounding $\E[R(Z)\cdot \ind\{Z\le 0\}]$.}
On the event $\{Z\le 0\}$, we have $1-\Phi(Z)\ge 1/2$ and  as a result, 
\[
R(Z)=\frac{\Phi(Z)}{1-\Phi(Z)}\le 2\Phi(Z).
\]
On the other hand, the standard Gaussian tail bound gives
\[
\Phi(Z)\le e^{-Z^2/2}.
\]
Combining the above two bounds, we have
\[
R(Z)\cdot \ind\{Z\le 0\}\le 2e^{-Z^2/2}.
\]
Taking expectation on both sides yields
\[
\E[R(Z)\cdot \ind\{Z\le 0\}]
\le
2\cdot \E\big[e^{-Z^2/2}\big] = 2\cdot \E\big[e^{-(W-a)^2/2}\big].
\]
Note that 
\$
\E\big[e^{-(W-a)^2/2}\big]
=
e^{-a^2/2}\E[e^{aW}e^{-W^2/2}]
\le
e^{-a^2/2}\E[e^{aW}] \le e^{-a^2(1-\rho^2)/2},
\$
where the last step follows from the sub-Gaussian property of $W$.
Combining the above, we conclude that
\begin{equation}
\label{eq:neg-part}
\E[R(Z)\cdot \ind\{Z\le 0\}]
\le
2\cdot e^{-a^2(1-\rho^2)/2}.
\end{equation}

\paragraph{Bounding $\E[R(Z)\cdot \ind\{Z>0\}]$.} 
On the event $\{0< Z\le 1\}$, we have 
\[
R(Z)=\frac{\Phi(Z)}{1-\Phi(Z)} \le \frac{1}{1 - \Phi(1)}.
\]
On the other hand, on the event $\{Z>1\}$, we have
\$ 
R(Z) =\frac{\Phi(Z)}{1-\Phi(Z)} \le \frac{1}{\frac{Z}{1+Z^2}\phi(Z)}
= \frac{1+Z^2}{Z}e^{Z^2/2} 
\le (1+Z) \cdot e^{Z^2/2}. 
\$
Combining the above two cases and taking the expectation, 
we have 
\$ 
\E\big[R(Z)\cdot \ind\{Z>0\}\big]
& \le \frac{1}{1-\Phi(1)} \cdot \E\Big[(1+Z)e^{Z^2/2}\cdot \ind\{Z\ge0\}\Big].
\$
Define $g(u):=(1+u) \cdot e^{u^2/2}$, for $u\ge 0$.
Since $g$ is increasing and absolutely continuous,
\$
\E\big[g(Z) \cdot \ind\{Z \ge 0  \}\big]
& = g(0) \cdot \P(Z\ge 0)+\int_0^\infty g'(u) \cdot \P(Z\ge u)\,\di u\\ 
& = \P(W\ge a)+\int_0^\infty (1+u+u^2)e^{u^2/2} \cdot \P(W\ge a+u)\,\di u\\
& \le e^{-a^2/(2\rho^2)} + \int_0^\infty (1+u+u^2)e^{u^2/2} \cdot e^{-(a+u)^2/(2\rho^2)}\, \di u,\\
& = e^{-a^2/(2\rho^2)}\cdot 
\bigg\{ 1 + \int_0^\infty (1+u+u^2)\cdot e^{-au/\rho^2 - (1-\rho^2)u^2/(2\rho^2)}\,\di u \bigg\}
\$
where the inequality follows from the sub-Gaussian tail bound for $W$.
Because $\rho^2<1$, the coefficient
\[
\beta:=\frac{1-\rho^2}{2\rho^2}>0.
\]
As a result, 
\$
1 + \int_0^\infty
(1+u+u^2)e^{-(a/\rho^2)u-\beta u^2}\,\di u
& \le
1 + \int_0^\infty
(1+u+u^2)e^{-\beta u^2}\,\di u\\
& = 1 + \frac{\sqrt{\pi}}{2}\beta^{-1/2} + \frac{1}{2}\beta^{-1/2} + \frac{\sqrt{\pi}}{4}\beta^{-3/2}
=:C_0(\rho).
\$
Combining the above, we conclude that
\begin{equation}
\label{eq:pos-part}
\E\big[R(Z)\cdot \ind\{Z>0\}\big]
\le
\frac{C_0(\rho)}{1-\Phi(1)} \exp\!\left(-\frac{a^2}{2\rho^2}\right).
\end{equation}

Finally, we combine the bounds 
in~\eqref{eq:neg-part}~and \eqref{eq:pos-part}
and define $C_1(\rho) := C_0(\rho)/(1-\Phi(1))$ to obtain
\[
\E[R(Z)]
\le
2\exp\!\left(-\frac{(1-\rho^2)a^2}{2}\right)
+
C_1(\rho)\exp\!\left(-\frac{a^2}{2\rho^2}\right).
\]
Letting 
\[
c(\rho):=\min\left\{\frac{1-\rho^2}{2},\frac{1}{2\rho^2}\right\}>0,
\text{ and }
C(\rho):=2+C_1(\rho),
\]
we conclude that
\[
\E[R(Z)]\le C(\rho) e^{-c(\rho) a^2}.
\]
The proof is complete.

\end{proof}

\section{Additional simulation results}\label{sec:additional_simulation}
This section presents additional numerical results evaluating the effectiveness of e-PS.
Section~\ref{app:sim-ucb} evaluates UCB-type sampling based on e-process-based rewards.
Section~\ref{sec:twogroup_simu} considers a simulation setup where the outcome distributions 
follow a two-group model. Finally, Section~\ref{sec:general_simu} evaluates the sampling methods 
with e-process constructions beyond those based on likelihood ratios.

\subsection{UCB-type sampling}\label{app:sim-ucb}
Given the e-process, we implement a UCB-style sampling rule using 
the e-process-based reward. To be specific, at each time point 
$t\ge 1$, we select the hypothesis 
\$ 
A_t := \underset{k \notin \mR_{t-1}}{\argmax} 
~\frac{\log E_{k,t}}{n_{k,t}} + \sqrt{\frac{2v_{k,t}\log n_{k,t} }{n_{k,t}}},
\$
where $v_{k,t}$ is chosen the same way as in  e-PS.
The other implementation details are identical to those in e-PS.

Figure~\ref{fig:svs-ucb} compares 
e-PS and UCB across all configurations considered in Section~\ref{sec:simulation}.
For both e-process constructions and across all distribution families, 
e-PS is always as powerful as UCB, with the advantage becoming more pronunced over time.
We also implement UCB-style sampling in the joke rating and watermark detection 
problems considered in Section~\ref{sec:realdata}.
In the joke rating experiment, e-PS achieves a higher TPR than UCB throughout the time 
horizon, whereas in the watermark detection experiment, 
the two methods yield nearly identical TPR curves.

\begin{figure}[h!]
    \centering
\includegraphics[width = 1\textwidth]{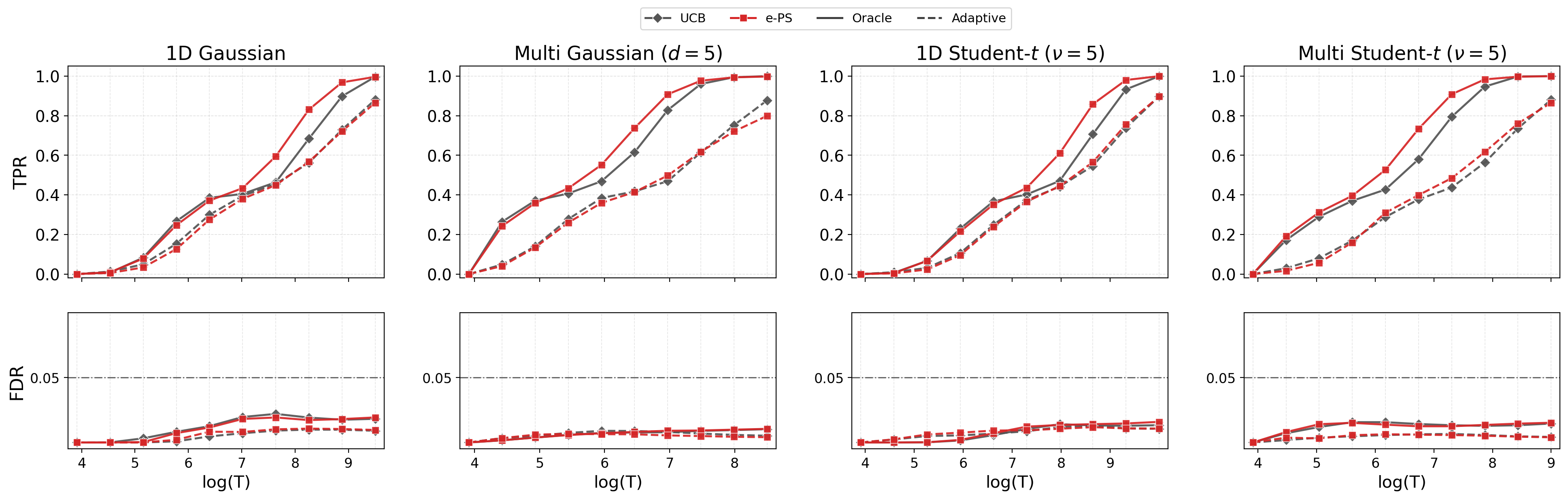}
    \caption{Realized TPR and FDR of UCB sampling and e-PS, averaged over $200$ independent trials.
    The settings 
    are otherwise the same as that in Figure~\ref{fig:svs}.}
\label{fig:svs-ucb}
\end{figure}

\begin{figure}[h!]
\centering
\begin{minipage}{0.48\textwidth}
\centering
\includegraphics[width = \textwidth]{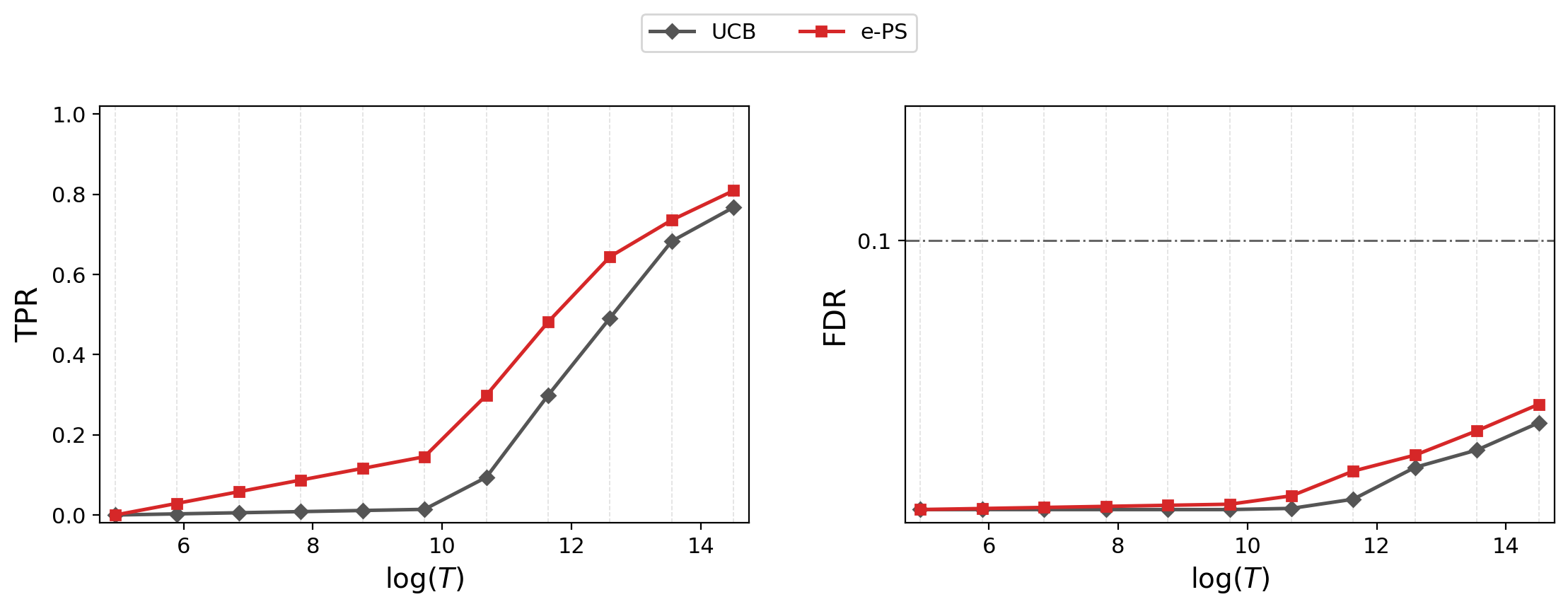}
(a) joke rating
\end{minipage}
\hfill
\begin{minipage}{0.48\textwidth}
\centering
\includegraphics[width = \textwidth]{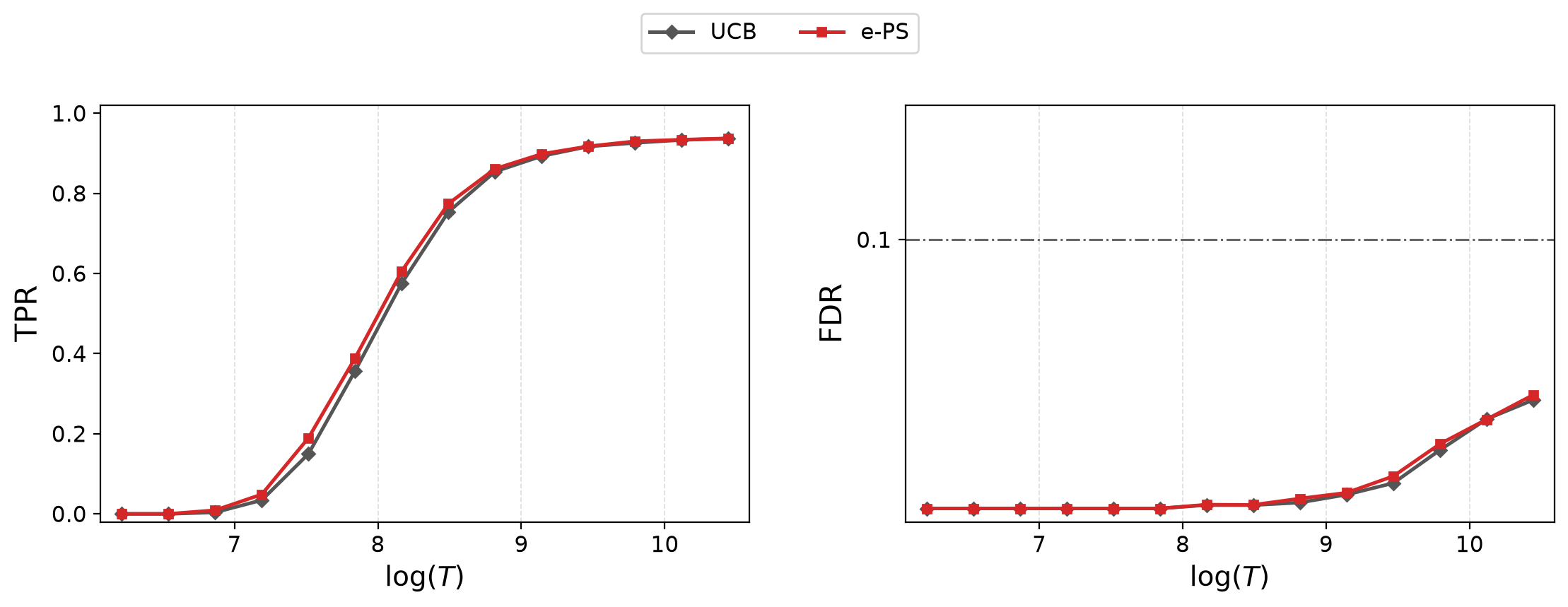}
(b) watermark detection
\end{minipage}
\caption{Realized TPR and FDR of e-PS and UCB-style sampling in the 
joke rating (left) and watermark detection (right) experiment as functions 
of the sample size, averaged over 500 repetitions.}
\label{fig:real-ucb}
\end{figure}

\subsection{Two-group model}\label{sec:twogroup_simu}
We now consider a two-group model with 
$Y_{t}(k) \sim p F_0 + (1-p) F_1$ for some null proportion $p\in [0,1]$ and 
distributions $F_0$ and $F_1$.
For each $k\in[K]$, we independently draw $B_k \sim \text{Bern}(1-p)$, 
and, for $t\ge 1$, sample $Y_t(k)$ from $F_0$ if $B_k = 0$ and 
$Y_t(k)$ from $F_1$ otherwise.
The task is to test for each $k\in[K]$:
\@\label{eq:two-group}
H_{k}^{(0)}: Y_t(k) \sim F_0 \text{ versus } H_{k}^{(1)}: Y_t(k) \sim F_1.
\@
    
In our experiment, we consider $K = 200$ arms and 
set $p=0.8$,
$F_0={N}(\bm{0},\bm{I}_d)$
and $F_1={N}( 0.5 \cdot \bm{1}_d,\bm{I}_d)$, 
with $d$ ranging in $\{1,5,10\}$.
For each setting, $200$ independent trials are carried out, 
where the $B_k$'s are drawn in each trial but fixed throughout the time horizon.
The FDR target is set to be $\alpha = 0.05$. 

To test~\eqref{eq:two-group}, we use the likelihood ratio e-process
and implemement e-PS with $v_{k,t} = 1.1 \hat \sigma_{k,t}^2$,
where $\hat \sigma_{k,t}^2$ is the sample variance computed from observations 
collected from arm $k$ up to time $t-1$.
In addition to the e-process-based methods, we evaluate a fixed-horizon baseline that  
applies the BH procedure to fixied-horizon p-values, with arms uniformly sampled from $[K]$.
This  baseline mimics the conventional multiple testing regime with static, 
i.i.d.sampling and a pre-specified sample size.

Figure \ref{fig:twogroup_1} shows the resulting TPR and FDR as functions of the sample size. 
Across all settings, e-PS generally achieves a higher TPR than the competing methods.
The only exception is  when $d=1$, where the greedy method shows an intial advantage,
but the gap diminishes over time.
Thw two-group setting is, in principle, relatively unfavorable to e-PS since the 
non-null signals are homogenous.
Nevertheless, e-PS continues to exhibit an empirical power advantage in most settings.
The comparison with the fixed-horizon baseline demonstrates the advantage of 
adaptive experimentation.




\begin{figure}[!htbp]
    \centering
      \includegraphics[height = 5.5cm, width = 13.35cm]{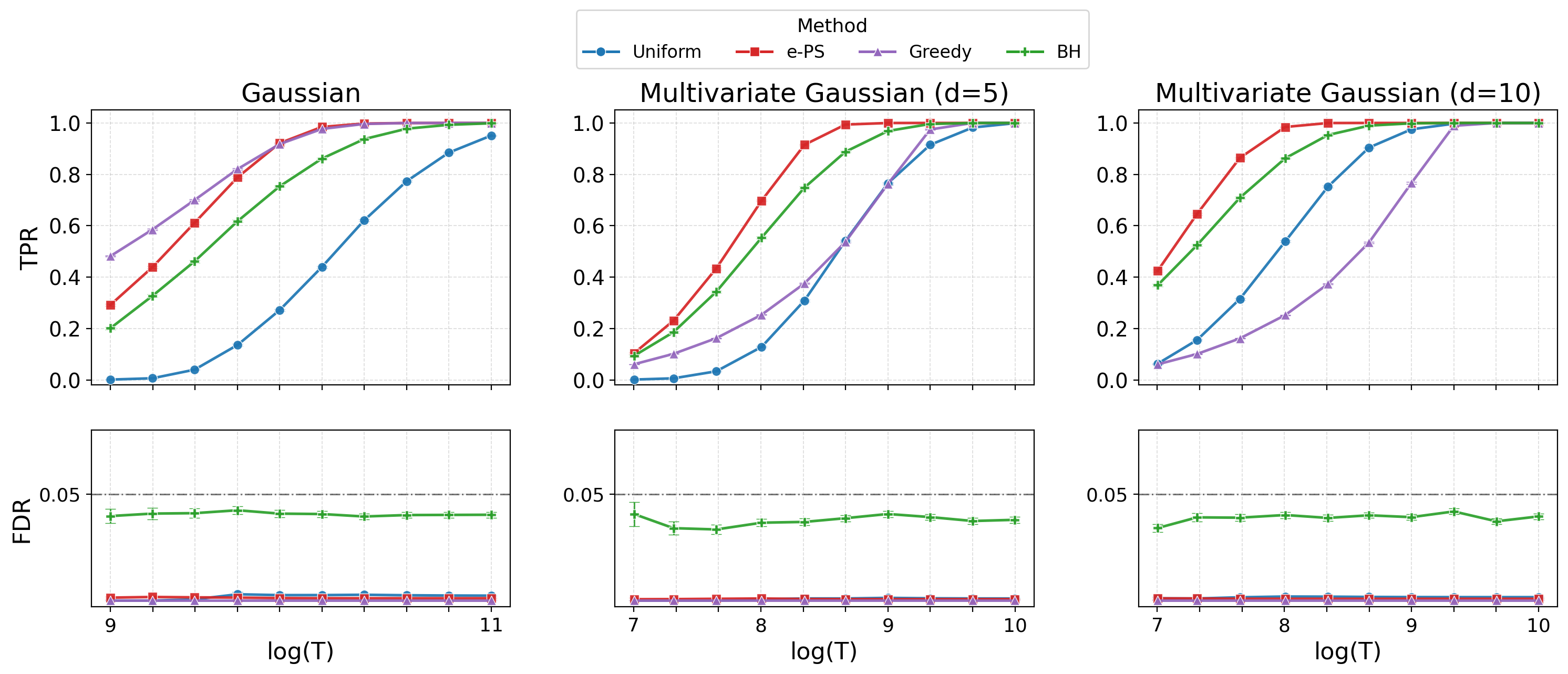}
    \caption{Realized TPR (top) and FDR (bottom) as functions of the sample size
    in the two group-model experiments described in Section \ref{sec:twogroup_simu}. The results are averaged over 
   $200$ independent trials.}
\label{fig:twogroup_1}
\end{figure}


\subsection{Mean detection with general e-values construction}\label{sec:general_simu}
In this section, we consider e-process constructions beyond likelihood ratio-based.
Specifically, we consider scalar outcomes and the mean testing problem for $K=200$ arms: for each $k\in[K]$, 
\$
H_k^{(0)}: \E[Y_1(k)] \leq 0 \text{ versus } H_k^{(1)}: \E[Y_1(k)] > 0,
\$ 
where we only assume that $Y_t(k)$ has a symmetric distribution and $Y_t(k) \in [-B,B]$ 
for some parameter $B>0$.

In the experiments, we consider two distribution families
for $P_k$:  
$\text{Unif}([\mu_k-0.1,\mu_k+0.1])$ and $\text{Unif}([\mu_k-1,\mu_k+1])$,
with $\mu_k$ as the parameter of interest.
For each setting, we randomly select $20$ indices from $[K]$
and assign their means the values $0.1,0.2,\dots,2$, respectively;
the remaining $\mu_k$'s are set to be zero. 
The mean parameters are resampled independent 
for each trial 
and then held fixed throughout the time horizon within 
that trial.
By construction, all outcomes lie in $[-1,3]$, and therefore $[-5,5]$.
We therefore take $B = 5$ in the algorithms.


For each $k\in [K]$ and $t\ge 1$, we construct the e-increment as
\$
e_{k,t} = 1 + \frac{\lambda}{B}Y_t \cdot \ind\{A_t = k\},
\$ 
where $\lambda = 0.99$ and $B = 5$.
Based on the e-process, we implement e-PS, with $v_{k,t}$ estimated in a hold-out set, and compare it  
with uniform and greedy sampling. As in Section~\ref{sec:twogroup_simu}, we 
also evaluate the fixed-horizon baseline that applies the BH procedure to fixed-horizon 
p-values. Under each setting, we conduct the experiment for $200$ independent trials 
with the target FDR $\alpha = 0.05$. 

Figure \ref{fig:meandetection_1} presents the realized TPR and FDP of all the candidate methods
as functions of the sample size. In both settings, we can see e-PS generally 
outperforms the competing methods, with an increasing advantege over time.

\begin{figure}[h!]
    \centering
      \includegraphics[width = 0.6\textwidth]{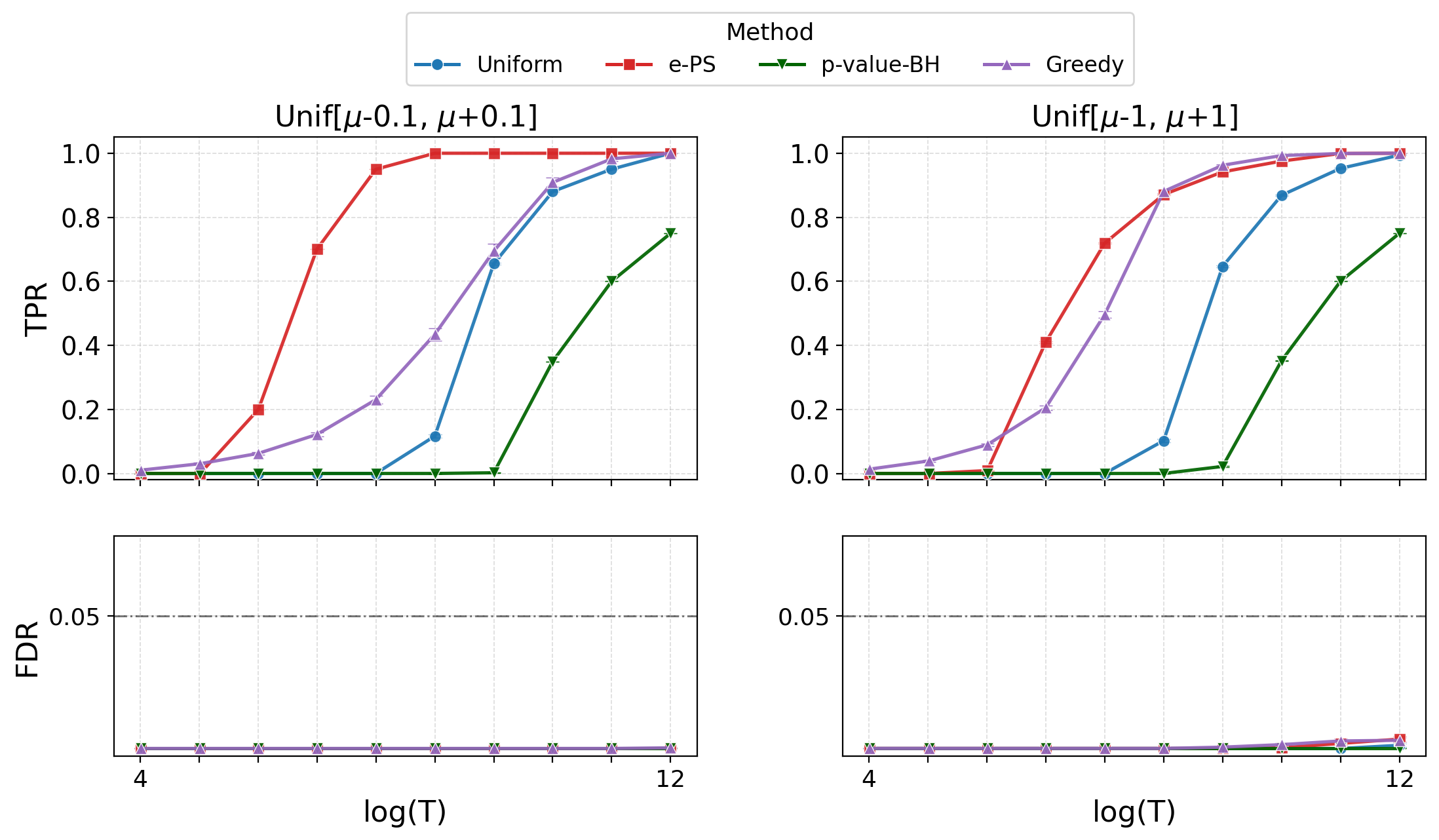}
    \caption{Realized TPR (top) and FDR (bottom) as functions of the sample size 
    in the mean-detection experiments described in Section \ref{sec:general_simu}. The results are averaged  
    over $200$ independent trials.}
\label{fig:meandetection_1}
\end{figure}

\end{document}